\documentclass[11pt,oneside,english]{amsart}
\RequirePackage{amsthm,amsmath,mathtools}

\usepackage[T1]{fontenc}
\usepackage{geometry,comment}
\usepackage{subfig}
\usepackage{setspace}
\usepackage{babel,verbatim}
\usepackage{float}
\usepackage{amstext}
\usepackage{amssymb,mathrsfs}
\usepackage{graphicx}
\usepackage[numbers]{natbib}
\RequirePackage[colorlinks=true,linkcolor=black,
	citecolor=blue,urlcolor=blue]{hyperref}

\numberwithin{equation}{section}
\numberwithin{figure}{section}
  \theoremstyle{plain}
  
  \theoremstyle{plain}
  \newtheorem*{lyxalgorithm*}{\protect\algorithmname}
\theoremstyle{plain}
\newtheorem{thm}{\protect\theoremname}
  \theoremstyle{remark}
  \newtheorem*{rem*}{\protect\remarkname}
  \theoremstyle{remark}
  \newtheorem{rem}[]{\protect\remarkname}
  \theoremstyle{plain}
  \newtheorem*{assumption*}{\protect\assumptionname}
  \theoremstyle{plain}
  \newtheorem{lem}[thm]{\protect\lemmaname}
  \theoremstyle{plain}
  
  \theoremstyle{plain}
  \newtheorem{prop}[thm]{\protect\propositionname}
  \theoremstyle{definition}
  \newtheorem*{example*}{\protect\examplename}
  
  \theoremstyle{definition}
  \newtheorem*{algorithm}{\protect\algorithmname}

\usepackage{tikz}
\usetikzlibrary{shapes.multipart,shapes,arrows,decorations.pathreplacing}

\usepackage{algorithm,algpseudocode,multirow,bbm,enumitem}
\usepackage{array}
\newcolumntype{L}{>{\centering\arraybackslash}m{1.55cm}}
\newcolumntype{C}{>{\centering\arraybackslash}m{4.35cm}}
\newcolumntype{T}{>{\centering\arraybackslash}m{4.85cm}}

\newcommand{\Oh}{\mathcal{O}}
\newcommand{\C}{\mathcal{C}}

\newcommand{\U}{\mathrm{U}}
\newcommand{\Exp}{\mathrm{Exp}}
\newcommand{\e}{\mathbb{E}}
\newcommand{\V}{\mathbb{V}}
\newcommand{\p}{\mathbb{P}}

\newcommand{\R}{\mathbb{R}}
\newcommand{\N}{\mathbb{N}}
\newcommand{\ov}[1]{\overline{#1}}
\newcommand{\1}{\mathbbm{1}}
\newcommand{\D}{\mathrm{d}}

\newcommand{\MC}{{\mathrm{MC}}}
\newcommand{\ML}{{\mathrm{ML}}}
\newcommand{\sgn}{{\mathrm{sgn}}}
\newcommand{\eqd}{\overset{d}{=}}
\newcommand{\cov}{\mathrm{cov}}
\newcommand{\blambda}{{\boldsymbol\lambda}}
\newcommand{\bzero}{{\boldsymbol0}}

 \usepackage{relsize,scalefnt}
\newcommand{\pp}{{\mathsmaller{(+)}}}
\newcommand{\pn}{{\mathsmaller{(-)}}}
\newcommand{\ppm}{{\mathsmaller{(\pm)}}}

\usepackage{lineno}
\usepackage{pgfplots} 
\usepgfplotslibrary{fillbetween}

\usepackage{caption}
\usepackage{nameref}
\DeclareCaptionLabelFormat{algnonumber}{}
\makeatother
  \providecommand{\algorithmname}{Algorithm}
  \providecommand{\assumptionname}{Assumption}
  \providecommand{\examplename}{Example}
  \providecommand{\lemmaname}{Lemma}
  \providecommand{\propositionname}{Proposition}
  \providecommand{\remarkname}{Remark}
\providecommand{\corollaryname}{Corollary}
\providecommand{\theoremname}{Theorem}

\begin{document}

\title[Stick-breaking approximation for tempered L\'evy processes]{Monte Carlo algorithm for the extrema of tempered stable processes}

\author{Jorge I. Gonz\'{a}lez C\'{a}zares \& Aleksandar Mijatovi\'{c}}

\address{Department of Statistics, University of Warwick, \& The Alan Turing Institute, UK}

\email{jorge.gonzalez-cazares@warwick.ac.uk}

\address{Department of Statistics, University of Warwick, \& The Alan Turing Institute, UK}

\email{a.mijatovic@warwick.ac.uk}



\begin{abstract}
We develop a novel Monte Carlo algorithm for the 
vector consisting of the supremum, the time at which the 
supremum is attained and the position at a given (constant) time
of an exponentially tempered L\'evy process. 
The algorithm, 
based on the increments of the process without tempering, 
converges geometrically fast (as a function of the 
computational cost) for discontinuous and locally Lipschitz functions of the 
vector. We prove that the corresponding multilevel Monte Carlo 
estimator has optimal computational complexity (i.e. of order 
$\varepsilon^{-2}$ if the mean squared error is at most $\varepsilon^2$)
	and provide its central limit theorem (CLT). Using the CLT we construct confidence 
intervals for barrier option prices and various risk measures based on drawdown 
under the tempered stable (CGMY) model calibrated/estimated on real-world data. 
We provide non-asymptotic and asymptotic comparisons of our algorithm with existing approximations,
leading to rule-of-thumb guidelines for users to the best method for a given
set of parameters. We
illustrate the performance of the algorithm with numerical examples. 
\end{abstract}

\maketitle

\section{Introduction}
\label{sec:intro}

\subsection{Setting and motivation}
The class of tempered stable processes is very popular in the financial modelling of asset prices of risky assets, see e.g.~\cite{tankov2015}. A tempered stable process $X=(X_t)_{t\geq0}$ naturally addresses the shortcomings of diffusion models by allowing the large (often heavy-tailed and asymmetric) sudden movements of the asset price observed in the markets, while preserving the exponential moments required in exponential L\'evy models $S_0 e^X$ of asset prices~\cite{CGMY,schoutens2003levy,KouLevy,tankov2015}. Of particular interest in this context are the expected drawdown (the current decline from a historical peak) and its duration (the elapsed time since the historical peak), see e.g.~\cite{DrawdownSornette,DrawdownVecer,CarrDrawdown,FutureDrawdowns, MR3556778}, as well as barrier option prices~\cite{MR1919609,MR2202995,MR2519843,MR3723380} and the estimation of ruin probabilities in insurance~\cite{MR2013414,MR2099651,MR3338431}. In these application areas, the key quantities are the historic maximum $\ov X_T$ at time $T$, the time $\tau_T(X)$ at which this maximum was attained during the time interval $[0,T]$ and 
the value $X_T$ of the process $X$ at time $T$. 

In this paper we focus on the Monte Carlo (MC) estimation of $\e [g(X_T,\ov X_T,\tau_T(X))]$, where the payoff $g$ is (locally) Lipschitz or of barrier type (cf. Proposition~\ref{prop:error-rates} below), covering the aforementioned applications. We construct a novel \emph{tempered stick-breaking algorithm} (\nameref{alg:TSB}), applicable  to a tempered L\'evy process, if the increments of the process without tempering can be simulated, which clearly holds if $X$ is a tempered stable processes. We show that the bias of \nameref{alg:TSB} decays geometrically fast in its computational cost for functions $g$ that are either locally Lipschitz or of barrier-type (see Subsection~\ref{sec:complexities} for details). We prove that the corresponding multilevel Monte Carlo (MLMC) estimator has optimal computational complexity (i.e. of order $\varepsilon^{-2}$ if the mean squared error is at most $\varepsilon^2$) and establish the central limit theorem (CLT) for the MLMC estimator. Using the CLT we construct confidence intervals for barrier option prices and various risk measures based on drawdown under the tempered stable (CGMY) model. \nameref{alg:TSB} combines the stick-breaking  algorithm in~\cite{LevySupSim} with the exponential change-of-measure for L\'evy processes, also applied in~\cite{Poirot2006} for the MC pricing of European options. A short 
\href{https://youtu.be/PKvSg2tKqfs}{YouTube}~\cite{Presentation_AM} video describes \nameref{alg:TSB} and the results of this paper.

\subsection{Comparison with the literature}

Exact simulation of the drawdown is currently out of reach. Under the assumption that the increments of the L\'evy process $X$ can be simulated (an assumption \textit{not} satisfied by tempered stable models of infinite variation), the algorithm SB-Alg in~\cite{LevySupSim} has a geometrically small bias, outperforming significantly other algorithms for which the computational complexity analysis has been carried out. For instance, the computational complexity analysis for the procedures presented in~\cite{MR3217440,MR3500619}, applicable to tempered stable processes of finite variation, has not been carried out making a direct quantitative comparison with SB-Alg~\cite{LevySupSim} not possible at present.
If the increments cannot be sampled, a general approach utilises the Gaussian approximation of small jumps, in which case the algorithm SBG-Alg~\cite{SBG} outperforms its competitors (e.g. random walk approxiamion, see~\cite{SBG} for details), while retaining polynomially small bias. Thus it suffices to compare \nameref{alg:TSB} below with SB-Alg~\cite{LevySupSim} and SBG-Alg~\cite{SBG}. Table~\ref{tab:compare} below provides a summary of the properties of \nameref{alg:TSB}, SB-Alg and SBG-Alg as well as the asymptotic computational complexities of the corresponding MC and MLMC estimators based on these algorithms (see also Subsection~\ref{subsec:SB_SBG_comparison} below for a detailed comparison).

\begin{table}[ht]
{\begin{spacing}{1.15}\scalefont{.72}
	\begin{tabular}{|L|C|C|T|}
		\hline
		\textbf{Algorithm}
		& \textbf{\nameref{alg:TSB}}
		& \textbf{SB-Alg}~\cite{LevySupSim} 
		& \textbf{SBG-Alg}~\cite{SBG}\\
		\hline
		\textbf{Class of L\'evy processes Algorithm applies to}
		& Tempered L\'evy process for which the increments of the process without tempering can be simulated (includes \emph{all} tempered stable processes!) 
		& L\'evy process whose increments can be simulated (among tempered stable, includes \emph{only} finite variation processes)
		& L\'evy process whose jumps of magnitude greater than any $\delta>0$ can be simulated (includes \emph{all} tempered stable processes)\\
		\hline\hline
		\textbf{Bias and level variance}
		& Both decay exponentially as $\Oh(e^{-\vartheta_g n})$ where $\vartheta>0$ depends on $g$ and $\beta$ (see Proposition~\ref{prop:error-rates} below)
		& Both decay exponentially as $\Oh(e^{-\vartheta_g n})$ where $\vartheta>0$ depends on $g$ and $\beta$ (see~\cite[Props.~1--3]{LevySupSim}) 
		& Bias and level variance decay polynomially as $\Oh(n^{-p})$ and $\Oh(n^{-q})$, resp., where $p\ge q>0$ depend on $g$ and $\beta$ (see~\cite[Sec.~3.2]{SBG} for bias and~\cite[Sec.~6.5.2]{SBG} for level variance) 
		\\
		\hline
		\textbf{MC complexity }
		& $\Oh(\varepsilon^{-2}\log(1/\varepsilon))$ for locally Lipschitz or barrier-type $g$ (see Section~\ref{sec:complexities} below)
		& $\Oh(\varepsilon^{-2}\log(1/\varepsilon))$ for locally Lipschitz or barrier-type $g$  (see~\cite[Sec.~2.4]{LevySupSim})
		& $\Oh(\varepsilon^{-2-\beta_*})$ for locally Lipschitz $g$; otherwise, complexity is larger (see~\cite[Tab.~2]{SBG})
		\\
		\hline
		\textbf{MLMC complexity}
		& $\Oh(\varepsilon^{-2})$ for locally Lipschitz or barrier-type $g$ (see Section~\ref{sec:complexities} below)
		& $\Oh(\varepsilon^{-2})$ for locally Lipschitz or barrier-type $g$  (see~\cite[Sec.~2.4]{LevySupSim})
		& $\Oh(\varepsilon^{-2\min\{\beta_*,1\}})$ for locally Lipschitz $g$; otherwise, complexity is larger (see~\cite[Tab.~3]{SBG})
		\\
		\hline
	\end{tabular}
\end{spacing}

\caption{Summary of the properties \nameref{alg:TSB}, SB-Alg~\cite{LevySupSim} and SBG-Alg~\cite{SBG}. The index $\beta_*$, defined in~\eqref{eq:BG+} below, is \emph{slightly} larger than the Blumenthal-Getoor index $\beta$, see Section~\ref{sec:Proofs} below for details. The bias and level variance are parametrised by computational effort $n$ as $n\to\infty$, while the MC and MLMC complexities are parametrised by the accuracy $\varepsilon$ (i.e., the mean squared error is at most $\varepsilon^2$) as $\varepsilon\to 0$.}
	\label{tab:compare}
}\end{table}

The stick-breaking (SB) representation in~\eqref{eq:chi} plays a central role in algorithms \nameref{alg:TSB}, SB-Alg and SBG-Alg. The SB representation was used in~\cite{LevySupSim} to obtain an approximation $\chi_n$ of $\chi:=(X_T,\ov{X}_T,\tau_T(X))$ that converges geometrically fast in the computational cost when the increments of $X$ can be simulated exactly. In~\cite{SBG}, the same representation was used in conjunction with a small-jump Gaussian approximation for arbitrary L\'evy  processes. In the present work we address a situation in between the two aforementioned papers using \nameref{alg:TSB} below. \nameref{alg:TSB} preserves the geometric convergence in the computational cost of SB-Alg, while being applicable to general tempered stable processes (unlike SB-Alg~\cite{LevySupSim} in the infinite variation case)
and asymptotically outperforming SBG-Alg~\cite{SBG}, see Tables~\ref{tab:compare} for an overview.

\subsection{Organisation}

The remainder of the paper is structured as follows. In Section~\ref{sec:main} we recall the SB representation and construct \nameref{alg:TSB}. In Section~\ref{sec:complexities} we describe the geometric decay of the bias and the strong error in $L^p$ and explain what the computational complexities of the MC and MLMC estimators are. We discuss briefly in Subsection~\ref{subsec:UnbiasedEstimators} the construction and properties of unbiased estimators based on \nameref{alg:TSB}. In Subsection~\ref{subsec:SB_SBG_comparison} we provide an in-depth comparison of \nameref{alg:TSB} with the SB and SBG algorithms, identifying where each algorithm outperforms the others. In Section~\ref{sec:stable} we consider the case of tempered stable processes and illustrate the previously described results
with numerical examples. The proofs of all the results except Theorem~\ref{thm:chi_exp_temp}, which is stated and proved in Section~\ref{sec:main}, are given in Section~\ref{sec:Proofs} below.

\section{The tempered stick-breaking algorithm}
\label{sec:main}

Let $T>0$ be a time horizon and $\blambda=(\lambda_+,\lambda_-)\in\R_+^2$ a vector with non-negative coordinates, different from the origin $\bzero=(0,0)$. A stochastic process $X=(X_t)_{t\in[0,T]}$ is said to be a \emph{tempered L\'evy process} under the probability measure $\p_\blambda$ if its characteristic function satisfies 
\[
t^{-1}\log\e_\blambda\big[e^{iuX_t}\big]
=iub_\blambda-\frac{1}{2}\sigma^2u^2+\int_{\R\setminus\{0\}}(e^{iux}-1-iux\cdot\1_{(-1,1)}(x))\nu_\blambda(\D x),\quad \text{for all }u\in\R,\,t>0,
\]
where $\e_\blambda$ denotes the expectation under the probability measure $\p_\blambda$, and the generating (or L\'evy) triplet $(\sigma^2,\nu_\blambda,b_\blambda)$ is given by 
\begin{equation}\label{eq:lambda}
\nu_\blambda(\D x) := e^{-\lambda_{\sgn(x)}|x|}\nu(\D x)
\qquad\text{and}\qquad
b_\blambda := b + \int_{(-1,1)}x\big(e^{-\lambda_{\sgn(x)}|x|}-1\big)\nu(\D x),
\end{equation}
where $\sigma^2\in\R_+$, $b\in\R$ and $\nu$ is a L\'evy measure on $\R\setminus\{0\}$, i.e. $\nu$ satisfies $\int_{(-1,1)}x^2\nu(\D x)<\infty$ and $\nu(\R\setminus(-1,1))<\infty$ (all L\'evy triplets in this paper are given with respect to the cutoff function $x\mapsto\1_{(-1,1)}(x)$ and the sign function in~\eqref{eq:lambda} is defined as $\sgn(x):=\1_{\{x>0\}}-\1_{\{x<0\}}$). The triplet $(\sigma^2,\nu_\blambda,b_\blambda)$ determines uniquely the law of $X$ via the L\'evy-Khintchine formula for the characteristic function of $X_t$ for $t>0$ given in the displays above (see details in~\cite[Thms~7.10 \& 8.1 and Def.~8.2]{MR3185174}).

Our aim is to sample from the law of the statistic $(X_T,\ov X_T,\tau_T)$ consisting of the position $X_T$ of the process $X$ at $T$, the supremum $\ov X_T:=\sup\{X_s:s\in[0,T]\}$ of $X$ on the time interval $[0,T]$ and the time $\tau_T:=\inf\{s\in[0,T]:\ov X_s=\ov X_T\}$ at which the supremum was attained in $[0,T]$. By~\cite[Thm~1]{LevySupSim} there exists a coupling $(X,Y,\ell)$ under a probability measure $\p_\blambda$, such that $\ell=(\ell_n)_{n\in\N}$ is a stick-breaking process on $[0,T]$ based on the uniform law $\U(0,1)$ (i.e. $L_0:=T$, $L_n:=L_{n-1}U_n$ and $\ell_n:=L_{n-1}-L_n$ for $n\in\N$ where $U_n$ are iid $\U(0,1)$), independent of the  L\'evy process $Y$ with law equal to that of $X$, and the SB representation holds $\p_\blambda$-a.s.:
\begin{equation}\label{eq:chi}
\chi := (X_T,\ov X_T,\tau_T)=\sum_{n=1}^\infty 
	\big(\xi_n,\max\{\xi_n,0\},\ell_n\cdot\1_{\{\xi_n>0\}}\big),
\quad\text{where}\quad
\xi_n := Y_{L_{n-1}}-Y_{L_n},\enskip n\in\N.
\end{equation}
We stress that $\ell$ is \emph{not} independent of $X$. In fact  
$(\ell,Y)$ can be expressed naturally through the geometry of the path of $X$
(see~\cite[Thm~1]{MR2978134} and the coupling in~\cite{LevySupSim}), 
but further details of the coupling are not important for our purposes.  
The key step in the construction of our algorithm is given by the following theorem.  Its proof is based on 
the coupling described above and the change-of-measure 
theorem for L\'evy processes~\cite[Thms~33.1 \& 33.2]{MR3185174}. 

\begin{thm}\label{thm:chi_exp_temp} 
Denote by $\sigma B$, $Y^\pp$, $Y^\pn$ the independent 
L\'evy processes with generating triplets $(\sigma^2,0,0)$, 
$(0,\nu_\blambda|_{(0,\infty)},0)$, $(0,\nu_\blambda|_{(-\infty,0)},0)$, 
respectively, satisfying $Y_t=\sigma B_t+Y_t^\pp+Y_t^\pn+b_\blambda t$ for all $t\in[0,T]$,
	$\p_\blambda$-a.s. Let $\e_\blambda$ (resp. $\e_\bzero$) be the expectation under 
	$\p_\blambda$ (resp. $\p_\bzero$) and define 
\begin{align}
\label{eq:upsilon}
\Upsilon_\blambda 
&:= \exp\big(-\lambda_+Y^\pp_T  +\lambda_-Y^\pn_T - \mu_\blambda T\big),
	\qquad\text{where}\\
\label{eq:RN-drift}
\mu_\blambda 
&:= \int_\R\big(e^{-\lambda_{\sgn(x)}|x|}-1
	+\lambda_{\sgn(x)}|x|\cdot\1_{(-1,1)}(x)\big)\nu(\D x),
\end{align}
Then for any $\sigma(\ell,\xi)$-measurable random variable $\zeta$ 
with $\e_\blambda|\zeta|<\infty$ we have 
$\e_\blambda[\zeta]
= \e_\bzero[\zeta\Upsilon_\blambda]$. 
\end{thm}

\begin{proof}
The exponential change-of-measure theorem for L\'evy processes (see~\cite[Thms~33.1 \& 33.2]{MR3185174}) implies that for any measurable function $F$ with $\e_\blambda |F((Y_t)_{t\in[0,T]})|<\infty$, we have the identity $\e_\blambda[F((Y_t)_{t\in[0,T]})]=\e_\bzero[F((Y_t)_{t\in[0,T]})\Upsilon_\blambda]$, where $\Upsilon_\blambda$ is defined in~\eqref{eq:upsilon} in the statement of Theorem~\ref{thm:chi_exp_temp}. Since the stick-breaking process $\ell$ is independent of $Y$ under both $\p_\blambda$ and $\p_\bzero$, this property extends to measurable functions of $(\ell, (Y_t)_{t\in[0,T]})$ and thus to the measurable functions of $(\ell,\xi)$, as claimed. 
\end{proof}

By the equality in~\eqref{eq:chi}, the measurable function $\zeta$ of the sequences $\ell$ and $\xi$ in Theorem~\ref{thm:chi_exp_temp} may be either equal to $g(\chi)$ (for any integrable function $g$ of 
the statistic $\chi$) 
or its approximation
$g(\chi_n)$, where $\chi_n$ is as 
introduced in~\cite{LevySupSim}: 
\begin{equation}\label{eq:chi_n}\begin{split}
\chi_n
&:=\big(Y_{L_n},\max\{Y_{L_n},0\},L_n\cdot\1_{\{Y_{L_n}>0\}}\big)
	+\sum_{k=1}^n\big(\xi_k,\max\{\xi_k,0\},\ell_k\cdot\1_{\{\xi_k>0\}}\big).
\end{split}\end{equation}
Theorem~\ref{thm:chi_exp_temp} enables us to sample $\chi_n$ under the probability measure $\p_\bzero$, which for any tempered stable process $X$ makes the increments of $Y$ stable and thus easy to simulate. Under $\p_\bzero$, the law of $Y_t$ equals that of $Y_t^\pp+Y_t^\pn+\sigma B_t+b t$, where $\sigma B_t+b t$ is normal $N(bt,\sigma^2 t)$ with mean $bt$ and variance $\sigma^2 t$ and the L\'evy processes $Y^\pp$ and $Y^\pn$ have triplets $(0,\nu|_{(0,\infty)},0)$ and $(0,\nu|_{(-\infty,0)},0)$, respectively. Denote their distribution functions by $F^\ppm(t,x):=\p_\bzero(Y_t^\ppm\le x)$, $x\in\R$, $t>0$. 

\begin{lyxalgorithm*}[TSB-Alg]
	Unbiased simulation of $g(\chi_n)$ under $\p_\blambda$
	\label{alg:TSB}
	\begin{algorithmic}[1]
		\Require{Tempering parameter $\blambda\in\R_+^2\setminus\{\bzero\}$, generating triplet $(\sigma^2,\nu,b)$,  time horizon $T>0$, 
			test function $g$, approximation level $n\in\N$} 
		\State{Set $L_0=T$ and compute $\mu_\blambda$ in~\eqref{eq:RN-drift}}
		\For{$k=1,\ldots,n$}
		\State{Sample $U_k\sim U(0,1)$ and put $L_k=U_k L_{k-1}$ 
			and $\ell_k=L_{k-1}-L_{k}$ \label{TSB_alg:line_3}}
		\State{Sample $\xi_k^\ppm\sim F^\ppm(\ell_k,\cdot)$, $G_k\sim N(\ell_k b,\sigma^2\ell_k)$ and put 
			$\xi_k=\xi_k^\pp+\xi_k^\pn+G_k$\label{TSB_alg:line_4}}
		\EndFor
		\State{Sample $\zeta_n^\ppm\sim F^\ppm(L_{n},\cdot)$, 
			$H_n\sim N(L_n b,\sigma^2L_n)$ and put 
			$\zeta_n=\zeta_n^\pp+\zeta_n^\pn+H_n$\label{TSB_alg:line_6}}
		\State{Set $\chi_n=\sum_{k=1}^{n}(\xi_k,\max\{\xi_k,0\},\ell_k\1_{\{\xi_k>0\}})
				+(\zeta_n,\max\{\zeta_n,0\},L_{n}\1_{\{\zeta_n>0\}})$ \& 
			$Y_T^\ppm=\zeta_n^\ppm+\sum_{k=1}^{n}\xi_k^\ppm$}
		\State{\Return $g(\chi_n)\exp(-\lambda_+Y_T^\pp+\lambda_-Y_T^\pn-\mu_\blambda T)$}
	\end{algorithmic}
\end{lyxalgorithm*}

Note that the output $g(\chi_n)\Upsilon_\blambda$ of \nameref{alg:TSB} is sampled under $\p_\bzero$ and, by Theorem~\ref{thm:chi_exp_temp} above, is unbiased since $\e_\bzero[g(\chi_n)\Upsilon_\blambda]=\e_\blambda[g(\chi_n)]$. As our aim is to obtain MC and MLMC estimators for $\e_\blambda[g(\chi)]$, our next task is to understand the expected error of \nameref{alg:TSB}, see~\eqref{eq:bias} in  Subsection~\ref{subsec:Bias_of_TSB_alg} below. In~\cite{LevySupSim} it was proved that the approximation $\chi_n$ converges geometrically fast in computational effort (or equivalently as $n\to\infty$) to $\chi$ if the increments of $Y$ can be sampled under~$\p_\blambda$ (see~\cite{LevySupSim} for more details and a discussion of the benefits of the ``correction term'' $(Y_{L_n},\max\{Y_{L_n},0\},L_n\cdot\1_{\{Y_{L_n}>0\}})$ in~\eqref{eq:chi_n}). Theorem~\ref{thm:chi_exp_temp} allows us to weaken this requirement in the context of tempered L\'evy processes, by requiring that we be able to sample the increments of $Y$ under~$\p_\bzero$. The main application of \nameref{alg:TSB} is to general tempered stable processes as the simulation of their increments is currently out of reach for many cases of interest (see Section~\ref{subsec:SB_SBG_comparison} below for comparison with existing methods when it is not), making the main algorithm in~\cite{LevySupSim} not applicable. Moreover, Theorem~\ref{thm:chi_exp_temp} allows us to retain the geometric convergence of $\chi_n$ established in~\cite{LevySupSim}, see Section~\ref{sec:complexities} below for details.

\section{MC and MLMC estimators based on \nameref{alg:TSB}}
\label{sec:complexities}

\subsection{Bias of \nameref{alg:TSB}}
\label{subsec:Bias_of_TSB_alg}
An application of Theorem~\ref{thm:chi_exp_temp} implies that the bias of \nameref{alg:TSB} equals 
\begin{equation}
\label{eq:bias}
\e_\blambda[g(\chi)]-\e_\bzero[g(\chi_n)\Upsilon_\blambda]
=\e_\bzero\big[\Delta_n^g\big],
	\quad\text{where $\Delta_n^g:=(g(\chi)-g(\chi_n))\Upsilon_\blambda$}.
\end{equation}
The natural coupling $(\chi,\chi_n, Y_T^\pp, Y_T^\pn)$
in~\eqref{eq:bias} is defined by
$Y_T^\ppm:=\sum_{k=1}^{\infty}\xi_k^\ppm$,
$\xi_k:=\xi_k^\pp+\xi_k^\pn+ \eta_k$ for all $k\in\N$,
$\chi$ in~\eqref{eq:chi}
and 
$\chi_n$ in~\eqref{eq:chi_n} with $Y_{L_n}:=\sum_{k=n+1}^\infty \xi_k$, where, conditional on the stick-breaking process $\ell=(\ell_k)_{k\in\N}$, 
the random variables $\{\xi_k^\ppm,\eta_k:k\in\N\}$ are independent and distributed as
$\xi_k^\ppm\sim F^\ppm(\ell_k,\cdot)$ and $\eta_k\sim N(\ell_k b,\sigma^2 \ell_k)$ for $k\in\N$.

The following result presents the decay of the strong error $\Delta_n^g$ for Lipschitz, locally Lipschitz and two classes of barrier-type discontinuous payoffs that arise frequently in applications. Observe that, in all cases and under the corresponding mild assumptions, the $p$-th moment of the strong error $\Delta_n^g$ decays exponentially fast in $n$ with a rate $\vartheta>0$ that depends on the payoff $g$, the index $\beta_*$ defined in~\eqref{eq:BG+} below and the degree $p$ of the considered moment. In Proposition~\ref{prop:error-rates} and throughout the paper, 
the notation 
$f(n)=\Oh(g(n))$ 
as $n\to\infty$ 
for functions 
$f,g:\N\to(0,\infty)$ 
means 
$\limsup_{n\to\infty}f(n)/g(n)<\infty$. 
Put differently, $f(n)=\Oh(g(n))$ is equivalent to $f(n)$ being bounded above by $C_0 g(n)$
for some constant $C_0>0$ and all $n\in\N$.

\begin{prop}
\label{prop:error-rates}
	Let $\blambda=(\lambda_+,\lambda_-)$, $\nu$ and $\sigma^2$ be as in Section~\ref{sec:main} and 
fix $p\ge 1$. Then, for the 
classes of payoffs $g(\chi)=g(X_T,\ov X_T, \tau_T)$  below,
the strong error of~\nameref{alg:TSB}  
decays  as follows (as $n\to\infty$).
\begin{description}[leftmargin=0cm]
\item[(Lipschitz)] 
Suppose $|g(x,y,t)-g(x,y',t')|\le K(|y-y'|+|t-t'|)$ for some $K$ and all $(x,y,y',t,t')\in\R\times\R_+^2\times[0,T]^2$.Then $\e_{\bzero}[|\Delta_n^g|^p]=\Oh(e^{-\vartheta_p n})$, where $\vartheta_p\in[\log(3/2),\log2]$ is in~\eqref{eq:eta} below.
\item[(locally Lipschitz)] 
Let $|g(x,y,t)-g(x,y',t')|\le K(|y-y'|+|t-t'|)e^{\max\{y,y'\}}$ for some constant~$K>0$ and all $(x,y,y',t,t')\in\R\times\R_+^2\times[0,T]^2$. If $\lambda_+\geq q > 1$ and 
		$\int_{[1,\infty)}e^{p(q-\lambda_+)x}\nu(\D x)<\infty$, then 
$\e_{\bzero}[|\Delta_n^g|^p]=\Oh(e^{-( \vartheta_{pr}/r)n})$, where 
		$r:=(1-1/q)^{-1}>1$ and $\vartheta_{pr}\in[\log(3/2),\log2]$ is as in~\eqref{eq:eta}. 
\item[(barrier-type 1)] Suppose $g(\chi)=h(X_T)\cdot\1\{\ov X_T\le M\}$ for some $M>0$ and a measurable bounded function $h:\R\to\R$. If $\p_{\bzero}(M<\ov X_T\le M+x)\le Kx$ for some $K>0$ and all $x\ge0$, then for $\alpha_*\in (1,2]$ in~\eqref{eq:alpha} and any $\gamma\in(0,1)$ we have $\e_{\bzero}[|\Delta_n^g|^p]=\Oh(e^{-[\gamma\log(2)/(\gamma+\alpha_*)]n})$. Moreover, we may take $\gamma=1$ if any of the following hold: $\sigma^2>0$ or $\int_{(-1,1)}|x|\nu(\D x)<\infty$ or Assumption~(\nameref{asm:s}) below. 
\item[(barrier-type 2)] Suppose $g(\chi)=h(X_T,\ov X_T)\cdot\1\{\tau_T\le s\}$, 
where $s\in(0,T)$, $h$ is measurable and bounded with 
$|h(x,y)-h(x,y')|\le K|y-y'|$ for some $K>0$ and all $(x,y,y')\in\R\times\R_+^2$. 
		If~$\sigma^2>0$ or $\nu(\R\setminus\{0\})=\infty$, then 
$\e_{\bzero}[|\Delta_n^g|^p]=\Oh(e^{-n/e})$. 
\end{description}
\end{prop}

\begin{rem}
	\label{rem:bounded_moments}
{\normalfont (i)} The proof of Proposition~\ref{prop:error-rates}, given in Section~\ref{sec:Proofs} below, is based on Theorem~\ref{thm:chi_exp_temp} and analogous bounds in~\cite{LevySupSim} (for Lipschitz, locally Lipschitz and barrier-type 1 payoffs) and~\cite{SBG} (for barrier-type 2 payoffs). In particular, in the proof of Proposition~\ref{prop:error-rates} below, we need not assume $\lambda_+>0$ to apply~\cite[Prop.~1]{LevySupSim}, which works under our standing assumption $\blambda\neq \bzero$.\\
{\normalfont (ii)} For barrier option payoffs under a tempered stable process $X$ (i.e. barrier-type 1 class in Proposition~\ref{prop:error-rates}), 
we may take $\gamma=1$ since $X$ satisfies either $\int_{(-1,1)}|x|\nu(\D x)<\infty$ or Assumption~(\nameref{asm:s}).\\
{\normalfont (iii)} The restriction $p\ge 1$ is not essential as we may consider any $p>0$ at the cost of a smaller (but still geometric) convergence rate. In particular, our standing assumption $\blambda\ne\bzero$ (and $\lambda_+>1$ in the locally Lipschitz case) guarantees the finiteness of the $p$-moment of the strong error $\Delta_n^g$ for any $p>0$. However, the restriction $p\ge 1$ covers the cases $p\in\{1,2\}$ required for the MC and MLMC complexity analyses and ensures that the corresponding rate $\vartheta_s$ in~\eqref{eq:eta} lies in $[\log(3/2),\log2]$. In fact, for any payoff $g$ in Proposition~\ref{prop:error-rates} we have $\e_\bzero[|\Delta_{k}^g|^p]=\Oh(e^{-\vartheta_gk})$ for $p\in\{1,2\}$ and a positive rate $\vartheta_g>0$ bounded away from zero: $\vartheta_g\geq 0.23$ (resp. $\log(3/2)$, $(1-1/\lambda_+)\log(3/2)$) for barrier-type 1 \& 2 (resp. Lipschitz, locally Lipschitz) payoffs. 

\end{rem}

\subsection{Computational complexity and the CLT for the MC and MLMC estimators}
\label{subsec:TSB-Alg_complexity}

Consider the MC estimator 
\begin{equation}
\label{eq:MC}
\hat\theta^{g,n}_\MC := \frac{1}{N}\sum_{i=1}^N \theta_i^{g,n},
\end{equation}
where $\{\theta_i^{g,n}\}_{i\in\N}$ is iid output of \nameref{alg:TSB} with $\theta_1^{g,n}\eqd g(\chi_n)\Upsilon_\blambda$ (under $\p_\bzero$) and $n,N\in\N$. The corresponding MLMC estimator is given by 
\begin{equation}
\label{eq:MLMC}
\hat\theta^{g,n}_\ML: = \sum_{k=1}^n\frac{1}{N_k}\sum_{i=1}^{N_k}D^g_{k,i},
\end{equation}
where $\{D^g_{k,i}\}_{k,i\in\N}$ is an array of independent 
variables satisfying 
$D^g_{k,i}\eqd (g(\chi_{k})-g(\chi_{k-1}))\Upsilon_\blambda$ 
and $D^g_{1,i}\eqd g(\chi_1)\Upsilon_\blambda$
(under $\p_\bzero$), for $i\in\N$, $k\ge 2$ and 
$n,N_1,\ldots,N_n\in\N$. Note that \nameref{alg:TSB} can be 
easily adapted to sample the variable 
$D^g_{k,i}$ by drawing the ingredients for 
$(\chi_k,\Upsilon_\blambda)$ and computing 
$(\chi_{k-1},\chi_k,\Upsilon_\blambda)$ deterministically from the output, 
see~\cite[Subsec.~2.4]{LevySupSim} for further details. In the following, we refer to $\V_\bzero[D_{k,1}^g]$ as the level variance of the MLMC estimator. 

The computational complexity analysis of the MC and MLMC estimators is given in the next result (the usual notation $\lceil x\rceil:=\inf\{k\in\N:k\ge x\}$, $x\in\R_+$, is used for the ceiling function). In Proposition~\ref{prop:MC_MLMC} and throughout the paper, the computational cost of an algorithm is measured as the total number of operations carried out by the algorithm. In particular, we assume that the following operations have computational costs uniformly bounded by some constant (measured, for instance, in units of time):
simulation  from the uniform law, simulation from the laws $F^{(\pm)}(t,\cdot)$, $t>0$, evaluation of elementary mathematical operations such as addition, subtraction, multiplication, division, as well as the evaluation of elementary functions such as $\exp$, $\log$, $\sin$, $\cos$, $\tan$ and $\arctan$.

\begin{prop}
\label{prop:MC_MLMC}
Let the payoff $g$ from Proposition~\ref{prop:error-rates} also satisfy $\e_\bzero[g(\chi)^2\Upsilon_\blambda^2]<\infty$. For any $\varepsilon>0$, let $n(\varepsilon) :=\inf\{k\in\N:|\e_{\bzero}[g(\chi_k)\Upsilon_\blambda]-\e_\blambda[g(\chi)]|
\le\varepsilon/\sqrt{2}\}$. Let $c$ be an upper bound on the expected computational cost of line~\ref{TSB_alg:line_6} in~\nameref{alg:TSB} for a time-horizon bounded by $T$ and let $\V_\bzero[\cdot]$ denote the variance under the probability measure $\p_\bzero$.\\
\textbf{(MC)} Suppose $n=n(\varepsilon)$ and 
$N=\lceil2\varepsilon^{-2}
	\V_\bzero[g(\chi_n)\Upsilon_\blambda]\rceil$, 
then the MC estimator $\hat\theta^{g,n}_\MC$ is 
$\varepsilon$-accurate, i.e.  
$\e_\bzero[|\hat\theta^{g,n}_\MC
	-\e_\blambda[g(\chi)]|^2]\le\varepsilon^2$, with expected cost $\C_\MC(\varepsilon):=c(n+1)N=\Oh(\varepsilon^{-2}\log(1/\varepsilon))$ as $\varepsilon\to0$.\\
\textbf{(MLMC)} Suppose $n=n(\varepsilon)$ and set
\begin{equation}
\label{eq:N_k}
	N_k:=\bigg\lceil2\varepsilon^{-2}
	\sqrt{\V_\bzero[D_{k,1}^g]/k}
	\bigg(\sum_{j=1}^{n}\sqrt{j\V_\bzero[D_{j,1}^g]}\bigg)
		\bigg\rceil,
	\quad k\in\{1,\ldots,n\}.
\end{equation}
Then the MLMC estimator $\hat\theta^{g,n}_\ML$ is $\varepsilon$-accurate and the corresponding expected cost equals 
\begin{equation}
\label{eq:MLMC_cost}
\C_\ML(\varepsilon)
:=2c\varepsilon^{-2}
	\bigg(\sum_{k=1}^{n}
		\sqrt{k\V_\bzero[D_{k,1}^g]}\bigg)^2= \Oh(\varepsilon^{-2})\quad\text{as $\varepsilon\to0$.}
\end{equation}
\end{prop}

Proposition~\ref{prop:error-rates} (with $p=1$) implies that the bias in~\eqref{eq:bias} equals $\e_\bzero[\Delta_n^g]=\Oh(e^{-\vartheta_g n})$ as $n\to\infty$ for some $\vartheta_g>0$. Thus, the integer $n(\varepsilon)$ in Proposition~\ref{prop:MC_MLMC} is finite for all payoffs $g$ considered in Proposition~\ref{prop:error-rates} and, moreover, $n(\varepsilon)=\Oh(\log(1/\varepsilon))$ as $\varepsilon\to0$ in all cases. In addition, by Remark~\ref{rem:bounded_moments}(i) above, the variance of $\theta^{g,k}_1$ is bounded in $k\in\N$:
$$
\V_\bzero[\theta^{g,k}_1]\leq \e_\bzero[g(\chi_k)^2\Upsilon_\blambda^2]\leq
2\e_\bzero[g(\chi)^2\Upsilon_\blambda^2]+2\e_\bzero[(\Delta_{k}^g)^2]
\to2\e_\bzero[g(\chi)^2\Upsilon_\blambda^2]<\infty\quad
\text{as $k\to\infty$.}$$
Note that barrier-type payoffs $g$ considered in Proposition~\ref{prop:error-rates} satisfy the second moment assumption, while in the Lipschitz or locally Lipschitz cases it is sufficient to assume additionally that $\lambda_+$ is either positive or strictly greater than one, respectively. Moreover,
$\V_\bzero[D_{k,1}^g]
\le 2\e_\bzero[(\Delta_{k}^g)^2+(\Delta_{k-1}^g)^2]
=\Oh(e^{-\vartheta_gk})$ for $\vartheta_g>0$ bounded away from zero (see Remark~\ref{rem:bounded_moments}(iii) above). These facts and the standard complexity analysis for MLMC (see e.g.~\cite[App.~A]{SBG} and the references therein) imply that the estimators $\hat\theta^{g,n}_\MC$ and $\hat\theta^{g,n}_\ML$ are $\varepsilon$-accurate with the stated computational costs, proving Proposition~\ref{prop:MC_MLMC}.

We stress that payoffs 
$g$ 
in Proposition~\ref{prop:MC_MLMC}
include discontinuous  payoffs in the supremum $\ov{X}_T$ (barrier-type~1) or 
the time $\tau_T$ this supremum is attained (barrier-type~2), with 
the corresponding computational complexities of the MC and MLMC estimators 
given by $\Oh(\varepsilon^{-2}\log(1/\varepsilon))$ and 
$\Oh(\varepsilon^{-2})$, respectively. 
This theoretical prediction matches the numerical performance 
of \nameref{alg:TSB} for barrier options and the modified ulcer index, see Section~\ref{subsec:MC_stable} below.

In order to obtain confidence intervals\footnote{The confidence intervals derived in this paper do not account for model uncertainty or the uncertainty in the estimation or calibration of the parameters.} (CIs) for the MC and MLMC estimators $\hat\theta^{g,n}_\MC$ and $\hat\theta^{g,n}_\ML$, a CLT can be very helpful. In fact, the CLT is necessary to construct a CI if the constants in the bounds on the bias in Proposition~\ref{prop:error-rates} are not explicitly known (e.g. for barrier-type~1 payoffs, the constant depends on the unknown value of the density of the supremum $\ov X_T$ at the barrier), see the discussion in~\cite[Sec.~2.3]{LevySupSim}. Moreover, even if the bias can be controlled explicitly, the concentration inequalities typically lead to wider CIs than those based on the CLT, see~\cite{MR3297771,MR3983217}. The following result establishes the CLT for the MC and MLMC estimators valid for  payoffs considered in Proposition~\ref{prop:error-rates}. 

\begin{thm}[CLT for \nameref{alg:TSB}]
\label{thm:CLT}
Let  $g$ be any of the payoffs in  Proposition~\ref{prop:error-rates}, satisfying in addition $\e_\bzero[g(\chi)^2\Upsilon_\blambda^2]<\infty$. Let $\vartheta_g\in(0,\log2]$ be the rate satisfying $\e_\bzero[|\Delta_n^g|]=\Oh(e^{-\vartheta_gn})$, given in Proposition~\ref{prop:error-rates} and Remark~\ref{rem:bounded_moments}(iii) above (with $p=1$). Fix $c_0>1/\vartheta_g$, let $n=n(\varepsilon):=\lceil c_0\log(1/\varepsilon)\rceil$ and suppose $N$ and $N_1,\ldots,N_n$ are as in Proposition~\ref{prop:MC_MLMC}. Then the MC and MLMC estimators $\hat\theta_\MC^{g,n}$ and $\hat\theta_\ML^{g,n}$ respectively satisfy the following CLTs ($Z$ is a standard normal random variable): 
\begin{equation}
\label{eq:CLT}
\sqrt{2}\varepsilon^{-1}(\hat\theta_\MC^{g,n(\varepsilon)}-\e_\blambda[g(\chi)])
	\overset{d}{\to}Z,
\qquad\text{and}\qquad
\sqrt{2}\varepsilon^{-1}(\hat\theta_\ML^{g,n(\varepsilon)}-\e_\blambda[g(\chi)])
	\overset{d}{\to}Z,
\qquad\text{as }\varepsilon\to 0.
\end{equation}
\end{thm}

Theorem~\ref{thm:CLT} works well in practice: in Figure~\ref{fig:CLT} of Section~\ref{subsec:MC_stable} below we construct CIs (as a function of decreasing accuracy $\varepsilon$) for an MLMC estimator of a barrier option price under a tempered stable model. The rate $c_0$ can be taken arbitrarily close to $1/\vartheta_g$, where $\vartheta_g$ is the corresponding geometric rate the of decay of the error for the payoff $g$ in Proposition~\ref{prop:error-rates} ($\vartheta_g$ is bounded away from zero by Remark~\ref{rem:bounded_moments}(iii) above), ensuring that the bias of the estimators vanishes in the limit.

By Lemma~\ref{lem:V_k_DCT} below, the definition of the sample sizes $N$ and $N_1,\ldots,N_n$ in Proposition~\ref{prop:MC_MLMC} implies that the variances of the estimators $\hat\theta_\MC^{g,n(\varepsilon)}$ and $\hat\theta_\ML^{g,n(\varepsilon)}$ (under $\p_\bzero$) satisfy 
\[
\frac{\V_\bzero[\hat\theta_\MC^{g,n(\varepsilon)}]}{\varepsilon^{2}/2}=\frac{\V_\bzero[\theta_1^{g,n(\varepsilon)}]}{\varepsilon^2 N/2}\to 1\quad
\&\quad
\frac{\V_\bzero[\hat\theta_\ML^{g,n(\varepsilon)}]}{\varepsilon^{2}/2}=\sum_{k=1}^{n(\varepsilon)}\frac{\V_\bzero[D_{k,1}^g]}{\varepsilon^2 N_k/2}\to 1\quad
\text{ as $\varepsilon\to0$}.
\]
Hence, the CLT in~\eqref{eq:CLT} can also be expressed as 
\begin{equation*}
(\hat\theta_\MC^{g,n(\varepsilon)}-\e_\blambda[g(\chi)])/\V_\bzero[\hat\theta_\MC^{g,n(\varepsilon)}]^{1/2} 
	\overset{d}{\to}Z
\quad\text{and}\quad
(\hat\theta_\ML^{g,n(\varepsilon)}-\e_\blambda[g(\chi)])/\V_\bzero[\hat\theta_\ML^{g,n(\varepsilon)}]^{1/2} 
	\overset{d}{\to}Z,
\quad\text{as }\varepsilon\to 0.
\end{equation*}
Since the variances $\V_\bzero[\hat\theta_\MC^{g,n(\varepsilon)}]$ and $\V_\bzero[\hat\theta_\ML^{g,n(\varepsilon)}]$ can be estimated from the sample, this is in fact how the CLT is often applied in practice. The proof of Theorem~\ref{thm:CLT} is based on the CLT for triangular arrays and amounts to verifying Lindeberg's condition for the estimators $\hat\theta_\MC^{g,n}$ and $\hat\theta_\ML^{g,n}$, see Section~\ref{sec:Proofs} below.

\subsection{Unbiased estimators}
\label{subsec:UnbiasedEstimators}
It is known that when the MLMC complexity is optimal, it is possible to construct unbiased estimators at the cost of inflating the variance of the estimator but without altering the optimal computational complexity $\Oh(\varepsilon^{-2})$ as $\varepsilon\to 0$. Such debiasing techniques are based around randomising the number of levels $\sup\{k\in\N:N_k>0\}$ and number of samples $(N_k)_{k\in\N}$ at each level of the variables $\{D^g_{k,i}\}_{k,i\in\N}$ in the MLMC estimator in~\eqref{eq:MLMC}, see e.g.~\cite{MR2890424,MR3422533,MR3782809}. More precisely, following~\cite[Thm~7]{MR3782809}, for any $g$ in Proposition~\ref{prop:error-rates} and any parameter $N\in\N$, we may construct random integers $(N_k)_{k\in\N}$ with explicit means $\e_\bzero[N_k]>0$ and $\sup\{k\in\N:N_k>0\}<\infty$ a.s. and the estimator
\begin{equation}
\label{eq:Matti_estimator}
\hat{P}_N\coloneqq\sum_{k=1}^\infty\frac{1}{\e_\bzero [N_k]}\sum_{i=1}^{N_k}D_{k,i}^g,
\end{equation}
is unbiased $\e_\bzero[\hat{P}_N]=\e_\blambda[g(\chi)]$ and its variance $\V_\bzero[\hat{P}_N]$ and expected computational cost (under $\p_\bzero$) are proportional to $1/N$ and $N$, respectively, as $N\to\infty$. The MC complexity analysis of the estimator $\hat{P}_N$ is then almost identical to that of the classical MC estimator for exact simulation algorithms. 

There are several parametric ways of constructing the random variables $(N_k)_{k\in\N}$ (see~\cite{MR3782809}) and it is also possible to optimise for the parameters appearing in these constructions as a function of the considered payoff $g$ and other characteristics of $X$ (see, e.g.~\cite[Sec.~2.5]{LevySupSim}). The details of such optimisations have been omitted in the present work in the interest of brevity since they coincide with those found in~\cite[Sec.~2.5 \& App.~A.3]{LevySupSim}.

\subsection{Comparisons}
\label{subsec:SB_SBG_comparison}

This subsection performs non-asymptotic and asymptotic performance comparisons of estimators based on \nameref{alg:TSB}. The main aim is to develop rule-of-thumb principles guiding the user to the most effective estimator. In Subsection~\ref{subsec:MCvMLMC}, for a given value of accuracy $\varepsilon$, we compare the computational complexity of the MC and MLMC estimators based on \nameref{alg:TSB}. The MLMC estimator based on \nameref{alg:TSB} is compared with the ones based on SB-Alg~\cite{LevySupSim} with rejection sampling (available only when the jumps of $X$ are of finite variation) and SBG-Alg~\cite{SBG} in Subsections~\ref{subsec:SBA_Masuda} and~\ref{subsec:SBG_comparison}, respectively. In both cases we analyse the behaviour of the computational complexity in two regimes: 
(I) $\varepsilon$ tending to $0$ and fixed time horizon~$T$; 
(II) fixed $\varepsilon$ and time horizon $T$ tending to $0$ or $\infty$.

Regime~(II) is useful when there is a limited benefit to arbitrary accuracy in $\varepsilon$ but the constants may vary greatly with the time horizon $T$. For example, in option pricing, estimators with accuracy smaller than a basis point are of limited interest. For simplicity, in the remainder of this subsection the payoff $g$ is assumed to be Lipschitz. However, analogous comparisons can be made for other payoffs under appropriate assumptions.

\subsubsection{Comparison between the MC and MLMC estimators based on \nameref{alg:TSB}}
\label{subsec:MCvMLMC}

Recall first that both MC and MLMC estimators have the same bias, since the latter estimator is a telescoping sum of a sequence of the former estimators, controlled by $n(\varepsilon)$  given in Theorem~\ref{thm:CLT} above.

\emph{Regime (I).}
Propositions~\ref{prop:error-rates} and~\ref{prop:MC_MLMC} imply that MLMC estimator outperforms the MC estimator as $\varepsilon\to 0$. Moreover, since 
$\V_\bzero[g(\chi_n)\Upsilon_\blambda]
\to\V_\bzero[g(\chi)\Upsilon_\blambda]$ and 
$\varepsilon^2\C_\ML(\varepsilon)
\to 2c(\sum_{k=1}^\infty(k\V_\bzero[D_{k,1}^g])^{1/2})^2<\infty$ as $\varepsilon\to0$, the MLMC estimator is preferable to the MC estimator for $\varepsilon > 0$ satisfying 
\[
n(\varepsilon)
>\bigg(\sum_{k=1}^\infty\sqrt{k\V_\bzero[D_{k,1}^g]}\bigg)^2/\V_\bzero[g(\chi)\Upsilon_\blambda].
\]
Since the payoff $g$ is Lipschitz, Proposition~\ref{prop:error-rates} implies that this condition is close to 
\[
\log(1/\varepsilon)
>\vartheta_1
\bigg(\sum_{k=1}^\infty\sqrt{k(2^{-k}-e^{-2\vartheta_1k})}\bigg)^2,
\]
where we recall that $\vartheta_1\in[\log(3/2),\log2]$ is defined in~\eqref{eq:eta} below. In particular, the latter condition is equivalent to $\varepsilon<0.0915$ if $\vartheta_1=\log(3/2)$, or $\varepsilon<5.06\times 10^{-5}$ if $\vartheta_1=\log2$. 

\emph{Regime (II).} Assume $\varepsilon>0$ is fixed. In this case, the estimators MC and MLMC share the value of $n=n(\varepsilon)$, which is $\Oh(\max\{1,\log T\})$ as either $T\to 0$ or $T\to\infty$. Moreover, the variances $\V_\bzero[g(\chi_n)\Upsilon_\blambda]$ (appearing in $\C_\MC$) and $\V_\bzero[D^g_{k,1}]$, $k\in\N$ (appearing in $\C_\ML$, see Proposition~\ref{prop:MC_MLMC} above) are all proportional to $\Oh((T+T^2)e^{(\mu_{2\blambda}-2\mu_\blambda)T})$ as either $T\to0$ or $T\to\infty$. Therefore, by Proposition~\ref{prop:MC_MLMC}, the quotient $\C_\MC/\C_\ML$ is proportional to a constant as $T\to 0$ and a multiple of $\log T$ as $T\to\infty$.

In conclusion, the MLMC estimator is preferable to the MC estimator when either $\varepsilon$ is small or $T$ is large. Otherwise, when $\varepsilon$ is not small and $T$ is small, both estimators have similar performance.

\subsubsection{Comparison with SB-Alg}
\label{subsec:SBA_Masuda}

In the special case when the jumps of $X$ have finite variation (equivalently, $\int_{(-1,1)}|x|\nu(\D x)<\infty$), the increments $X_t$ can be simulated under $\p_\blambda$ using rejection sampling (see~\cite{MR2763192,MR3969059}), making SB-Alg~\cite{LevySupSim} applicable to sample $\chi_n$ (see~\eqref{eq:chi_n} for its definition) under~$\p_\blambda$. The rejection sampling is performed for each of the increments of the subordinators 
\[
\widetilde Y^\ppm_t := \pm Y^\ppm_t + d_\pm t,
\quad\text{where}\quad
d_+ := \int_{(0,1)}x\nu(\D x),
\quad\text{and}\quad
d_- := \int_{(-1,0)}|x|\nu(\D x),
\] 
and the processes $Y^\ppm$ are as in 
Theorem~\ref{thm:chi_exp_temp}. The algorithm proposes samples 
under $\p_\bzero$ and rejects independently with probability 
$\exp(-\lambda_\pm\widetilde Y_t^\ppm)$. 
Let $\blambda_+ := (\lambda_+,0)$ and 
$\blambda_-:=(0,\lambda_-)$, then the expected number of 
proposals required for each sample equals 
$\exp(\gamma^\ppm_\blambda t)
=1/\e_\bzero[\exp(-\lambda_\pm\widetilde Y_t^\ppm)]$,
 where we define 
\begin{equation}
\label{eq:gamma_pm}
\gamma^\ppm_\blambda 
:= \lambda_\pm d_\pm-\mu_{\blambda_\pm} 
= \int_{\R_\pm}(1-e^{-\lambda_{\pm}|x|})\nu(\D x) \in [0,\infty).
\end{equation} 
(Note that $\mu_{p\blambda}-p \mu_\blambda 
	= p (\gamma^\pp_{\blambda}+\gamma^\pn_{\blambda}) 
	- (\gamma^\pp_{p\blambda}+\gamma^\pn_{p\blambda})$, 
see~\eqref{eq:RN-drift}.)

We need the following elementary lemma to analyse the computational 
complexity of SB-Alg with rejection sampling. 
\begin{lem}
\label{lem:asymp_cost} 
{\normalfont(a)} 
Let $\ell$ be a stick-breaking process on $[0,1]$, then for any $n\in\N$ we have
\begin{equation}
	\label{eq:asymp_cost2}
	0\le  n + \int_0^1\frac{1}{x}\big(e^{cx}-1\big)\D x 
	- \sum_{k=1}^n \e[e^{c\ell_k}] 
	\le 2^{-n}\int_0^1\frac{1}{x}\big(e^{cx}-1\big)\D x.
\end{equation}
{\normalfont(b)} We have 
$c^{-1}e^{-c}(1+c^2)\int_0^1x^{-1}(e^{cx}-1)\D x\to1$ as 
either $c\to 0$ or $c\to\infty$. 
\end{lem}

Assume that the simulation of the increments $Y^\ppm_t$ under $\p_\bzero$ have constant cost independent of the time horizon $t$ (we also assume without loss of generality that the simulation of uniform random variables and the evaluation of operators such as sums, products and the exponential functions have constant cost). Since the SB-Alg requires the rejection sampling to be carried out over random stick-breaking lengths, the expected cost of the SB-Alg with rejection sampling is, by~\eqref{eq:asymp_cost2} in Lemma~\ref{lem:asymp_cost}, asymptotically proportional to 
\begin{equation}
\label{eq:asymp_cost}
\sum_{k=1}^n \e[e^{\gamma^\pp_\blambda\ell_k}
		+ e^{\gamma^\pn_\blambda\ell_k}]
= 2n + (1+\Oh(2^{-n}))\int_0^1\frac{1}{x}
	\big(e^{\gamma_\blambda^\pp Tx} 
		+ e^{\gamma_\blambda^\pn Tx} - 2\big)\D x,
\quad\text{as }n\to\infty.
\end{equation}
In fact, by Lemma~\ref{lem:asymp_cost}(a), the absolute 
value of the term $\Oh(2^{-n})$ is bounded by $2^{-n}$. 
Moreover, by Lemma~\ref{lem:asymp_cost}(b), the integral 
in~\eqref{eq:asymp_cost} may be replaced with an 
explicit expression 
\begin{equation}
\label{eq:Gamma}
\Gamma_{\blambda}(T)
:=\frac{\gamma_\blambda^\pp T}{1+(\gamma_\blambda^\pp T)^2}
	e^{\gamma_\blambda^\pp T}
+\frac{\gamma_\blambda^\pn T}{1+(\gamma_\blambda^\pn T)^2}
	e^{\gamma_\blambda^\pn T},
\end{equation}
as $T$ tends to either $0$ or $\infty$. 

Table~\ref{tab:TSB_SBA} shows how SB-Alg with rejection sampling compares to \nameref{alg:TSB} above. 

\begin{table}[h]
	\begin{tabular}{|c|c|c|c|} 
		\hline
		Algorithm 
		& Bias
		& Level variance 
		& Cost \\
		\hline
		SB-Alg 
		& $e^{-\vartheta_1n}(\sqrt{T}+T)$
		&
		$
		e^{-\log(2)n}(T+T^2)$ & 
		$2n + \Gamma_{\blambda}(T)$
		\\
		\hline
		\nameref{alg:TSB} 
		& $e^{-\vartheta_1n}(\sqrt{T}+T)$
		& 
		$
			e^{-\log(2)n}(T+T^2)
				e^{(\mu_{2\blambda}-2\mu_{\blambda})T}$ & 
		$2n$ \\
		\hline
	\end{tabular}
	\caption{Level variance and expected cost of SB-Alg and \nameref{alg:TSB}, up to multiplicative constants that do not depend on the time horizon $T$. The bounds on the level variances (defined in the first paragraph of Subsection~\ref{subsec:TSB-Alg_complexity} above) follow from~\cite[Thm~2]{LevySupSim} for SB-Alg and Proposition~\ref{prop:error-rates} for \nameref{alg:TSB}.}
	\label{tab:TSB_SBA}
\end{table}

\emph{Regime (I).} 
By Table~\ref{tab:TSB_SBA}, we can deduce that the MC and MLMC estimators of both algorithms have the complexities $\Oh(\varepsilon^{-2}\log(1/\varepsilon))$ and $\Oh(\varepsilon^{-2})$ as $\varepsilon\to0$, respectively, for all the payoffs considered in Proposition~\ref{prop:error-rates}. 

\emph{Regime (II).} 
Assume $\varepsilon$ is fixed. The biases of both algorithms agree and equal $\Oh(e^{-\vartheta_1n}(\sqrt{T}+T))$, implying $n=\log((\sqrt{T}+T)/\varepsilon)/\vartheta_1+\Oh(1)$ (with $\vartheta_1$ defined in~\eqref{eq:eta}). The level variances of SB-Alg and \nameref{alg:TSB} are $\Oh(e^{-\log(2)n}(T+T^2))$ and $\Oh(e^{-\log(2)n}(T+T^2)e^{(\mu_{2\blambda}-2\mu_\blambda)T})$, with costs $\Oh(n+\Gamma_{\blambda}(T))$ and $\Oh(n)$, respectively. Thus, by~\eqref{eq:Gamma}, the ratios of the level variance and cost converge to $1$ as $T\to 0$, so the ratio of the complexities of both algorithms converges to $1$, implying that one should use \nameref{alg:TSB} as its performance in this regime is no worse than that of the other algorithm but its implementation is simpler.  For moderate or large values of $T$, by~\cite[Eq.~(A.3)]{SBG}, the computational complexity of the MLMC estimator based on \nameref{alg:TSB} is proportional to $\varepsilon^{-2}e^{(\mu_{2\blambda}-2\mu_\blambda)T}(T+T^2)$and that of SB-Alg is proportional to  $\varepsilon^{-2}(1+\Gamma_{\blambda}(T))(T+T^2)$, where the proportionality constants in both cases do not depend on the time horizon $T$. Since both constants are exponential in $T$, for large $T$ we need only compare $\max\{\gamma_\blambda^\pp,\gamma_\blambda^\pn\}$ with $\mu_{2\blambda}-2\mu_\blambda$. Indeed,  \nameref{alg:TSB} has a smaller complexity than SB-Alg for all sufficiently large $T$ if and only if $\mu_{2\blambda}-2\mu_\blambda<\max\{\gamma_\blambda^\pp,\gamma_\blambda^\pn\}$. In Subsection~\ref{subsec:SB-Alg_TS} below, we provide an explicit criterion for the tempered stable process in terms of the parameters, see Figure~\ref{fig:SB-TSB}. 

In conclusion, when $X$ has jumps of finite variation, it is preferable to use \nameref{alg:TSB} if $T$ is small or $e^{(\mu_{2\blambda}-2\mu_\blambda)T}<1+\Gamma_\blambda(T)$. Moreover, this typically holds if the Blumenthal-Getoor index of $X$ is larger than $\log_2(3/2)<0.6$, see Subsection~\ref{subsec:SB-Alg_TS} and Figure~\ref{fig:SB-TSB} below for details. 

\subsubsection{Comparison with SBG-Alg}
\label{subsec:SBG_comparison}

Given any cutoff level $\kappa\in(0,1]$, the algorithm SBG-Alg approximates the L\'evy process $X$ (under $\p_\blambda$ with the generating triplet $(\sigma^2, \nu_\blambda, b_\blambda)$) with a jump-diffusion $X^{(\kappa)}$ and samples exactly the vector $(X^{(\kappa)}_T,\ov X^{(\kappa)}_T,\tau_T(X^{(\kappa)}))$, see~\cite{SBG} for details. The bias of SBG-Alg is $\Oh(\min\{\sqrt{T}\kappa^{1-\beta_*/2},\kappa\} (1+\log^+(T\kappa^{-\beta_*})))$ with corresponding expected cost $\Oh(T\kappa^{-\beta_*})$. Thus, when parametrised by computational effort $n$, the bias of SBG-Alg is $\Oh(T^{1/\beta_*}n^{-1/\beta_*}\log n)$ while the bias of \nameref{alg:TSB} is $\Oh(e^{-\vartheta_1 n}(\sqrt{T}+T))$.

\emph{Regime (I).} 
The complexities of the MC and MLMC estimators based on the SBG-Alg are by~\cite[Thm~22]{SBG} equal to $\Oh(\varepsilon^{-2-\beta_*})$ and $\Oh(\varepsilon^{-\max\{2,2\beta_*\}}(1+\1_{\{\beta_*=1\}}\log(\varepsilon)^2))$, respectively, where $\beta_*$, defined in~\eqref{eq:BG+} below, is \emph{slightly} larger than the Blumenthal-Getoor index of $X$. Since, by Subsection~\ref{subsec:TSB-Alg_complexity} above, the complexities of the MC and MLMC estimators based on \nameref{alg:TSB} are $\Oh(\varepsilon^{-2}\log(1/\varepsilon))$ and $\Oh(\varepsilon^{-2})$, respectively, the complexity of \nameref{alg:TSB} is always dominated by that of SBG-Alg as $\varepsilon\to0$. 

\emph{Regime (II).} 
Suppose $\varepsilon>0$ is fixed. Then, as in  Subsection~\ref{subsec:SBA_Masuda} above, the computational complexity of the MLMC estimator based on \nameref{alg:TSB} is $\Oh(\varepsilon^{-2}(T+T^2)e^{(\mu_{2\blambda}-2\mu_\blambda)T})$, where the multiplicative constant does not depend on the time horizon $T$. By~\cite[Thms~3 \&~22]{SBG} and~\cite[Eq.~(A.3)]{SBG}, the complexity of the MLMC estimator based on the SBG-Alg is $\Oh(\varepsilon^{-2}(C_1 T+C_2T^2))$ if $\beta_*< 1$, 
$\Oh(\varepsilon^{-2}(C_1T+C_2T^2\log^2(1/\varepsilon)))$ if $\beta_*= 1$, and 
$\Oh(\varepsilon^{-2}(C_1T+C_2T^2\varepsilon^{-2(\beta_*-1)}))$ otherwise, 
where $C_1 := e^{r\beta_*}/(1-e^{r(\beta_*/2-1)})^{2}$, 
$C_2 := e^{r\beta_*}/(1-e^{r(\beta_*-1)})^{2}\cdot\1_{\{\beta_*\ne 1\}} 
	+(e/2)^2\cdot\1_{\{\beta_*= 1\}}$ and 
$r := (2/|\beta_*-1|)\log(1+|\beta_*-1|/\beta_*)\cdot\1_{\{\beta_*\ne 1\}} 
	+ 2\cdot\1_{\{\beta_*=1\}}$. All other multiplicative constants in these 
bounds do not depend on the time horizon $T$. Thus \nameref{alg:TSB} 
outperforms SBG-Alg if and only if we are in one of the following cases: 
\begin{itemize}
\item $\beta_*< 1$ and $(1+T)e^{(\mu_{2\blambda}-2\mu_\blambda)T}
	<C_1+C_2T$, 
\item $\beta_*=1$ and $(1+T)e^{(\mu_{2\blambda}-2\mu_\blambda)T}
	<C_1+C_2T\log^2(1/\varepsilon)$, 
\item $\beta_*>1$ and $(1+T)e^{(\mu_{2\blambda}-2\mu_\blambda)T}
	<C_1+C_2T\varepsilon^{-2(\beta_*-1)}$.
\end{itemize}
Note that the constant $C_2$ is unbounded for $\beta_*$ close to 1, favouring \nameref{alg:TSB}. 

In conclusion, \nameref{alg:TSB} is simpler than SBG-Alg~\cite{SBG} and it outperforms it asymptotically as $\varepsilon\to0$. Moreover, \nameref{alg:TSB} performs better than SBG-Alg for a fixed accuracy $\varepsilon>0$ if either (I) $\beta_*<1$ and the time horizon $T\ll 1$ satisfies the inequality $T<\log(C_1)/(\mu_{2\blambda}-2\mu_\blambda)$ or (II) $\beta_*\geq 1$ and $T$ is not large. In Subsection~\ref{subsec:SBG-Alg_TS} below, we apply this criterion to tempered stable process, see Figure~\ref{fig:SBG-TSB} for the case (I) $\beta_*< 1$ and Figure~\ref{fig:SBG-TSB2} for the case (II) $\beta_*\ge 1$. 

\subsection{Variance reduction via exponential martingales}
\label{subsec:var_red}

It follows from Subsection~\ref{subsec:SB_SBG_comparison}  that the performance of~\nameref{alg:TSB} deteriorates if the expectation $\e_\bzero [\Upsilon_\blambda^2] = \exp((\mu_{2\blambda}-2\mu_\blambda)T)$ is large, making the variance of the estimator large. This problems is mitigated by using the control variates method, a variance reduction technique, based on exponential $\p_\bzero$-martingales, at (nearly) no additional computational cost. 

Suppose we are interested in estimating $\e_\blambda[\zeta]=\e_\bzero[\zeta\Upsilon_\blambda]$, where $\zeta$ is either $g(\chi_n)$ (MC) or $g(\chi_n)-g(\chi_{n-1})$ (MLMC). We propose drawing samples of $\tilde\zeta$ under $\p_\bzero$, instead of $\zeta\Upsilon_\blambda$, where 
\[
\tilde\zeta
:=\zeta\Upsilon_\blambda 
- w_0 (\Upsilon_\blambda-1) 
- w_+ (\Upsilon_{\blambda_+}-1)
- w_- (\Upsilon_{\blambda_-}-1),
\]
and $w_0,w_+,w_-\in\R$ are constants to be determined (recall $\blambda_+=(\lambda_+,0)$ and $\blambda_-=(0,\lambda_-)$). Clearly $\e_\bzero[\tilde\zeta]=\e_\bzero[\zeta\Upsilon_\blambda]$ since the variables $\Upsilon_\blambda$, $\Upsilon_{\blambda_+}$ and $\Upsilon_{\blambda_-}$ have unit mean under $\p_\bzero$. We  choose $w_0,w_+$ and $w_-$ to minimise the variance of $\tilde\zeta$, by minimising the quadratic form~\cite[Sec.~4.1.2]{MR1999614} of $\V_\bzero[\tilde\zeta]$:
\[
\left(
\begin{array}{c}
-1\\
\omega_0\\
\omega_+\\
\omega_-
\end{array}
\right)^{\top}
\left(
\begin{array}{cccc}
\V_\bzero[\zeta\Upsilon_\blambda]
& \cov_\bzero(\zeta\Upsilon_\blambda,\Upsilon_\blambda)
& \cov_\bzero(\zeta\Upsilon_\blambda,\Upsilon_{\blambda_+})
& \cov_\bzero(\zeta\Upsilon_\blambda,\Upsilon_{\blambda_-})\\
\cov_\bzero(\zeta\Upsilon_\blambda,\Upsilon_\blambda)
& \V_\bzero[\Upsilon_\blambda] 
& \cov_\bzero(\Upsilon_\blambda,\Upsilon_{\blambda_+}) 
& \cov_\bzero(\Upsilon_\blambda,\Upsilon_{\blambda_-})\\
\cov_\bzero(\zeta\Upsilon_\blambda,\Upsilon_{\blambda_+})
& \cov_\bzero(\Upsilon_\blambda,\Upsilon_{\blambda_+}) 
& \V_\bzero[\Upsilon_{\blambda_+}] 
& \cov_\bzero(\Upsilon_{\blambda_+},\Upsilon_{\blambda_-})\\
\cov_\bzero(\zeta\Upsilon_\blambda,\Upsilon_{\blambda_-})
& \cov_\bzero(\Upsilon_\blambda,\Upsilon_{\blambda_-})
& \cov_\bzero(\Upsilon_{\blambda_+},\Upsilon_{\blambda_-}) 
& \V_\bzero[\Upsilon_{\blambda_-}]
\end{array}
\right)
\left(
\begin{array}{c}
-1\\
\omega_0\\
\omega_+\\
\omega_-
\end{array}
\right).
\]
The solution, in terms of the covariance matrix (under $\p_\bzero$) of the vector $(\zeta\Upsilon_\blambda,\Upsilon_\blambda,\Upsilon_{\blambda_+},\Upsilon_{\blambda_-})$, is: 
\[
\left(
\begin{array}{c}
w_0\\
w_+\\
w_-
\end{array}
\right)
= 
\left(
\begin{array}{ccc}
2\V_\bzero[\Upsilon_\blambda] 
& \cov_\bzero(\Upsilon_\blambda,\Upsilon_{\blambda_+}) 
& \cov_\bzero(\Upsilon_\blambda,\Upsilon_{\blambda_-})\\
\cov_\bzero(\Upsilon_\blambda,\Upsilon_{\blambda_+}) 
& 2\V_\bzero[\Upsilon_{\blambda_+}] 
& \cov_\bzero(\Upsilon_{\blambda_+},\Upsilon_{\blambda_-})\\
\cov_\bzero(\Upsilon_\blambda,\Upsilon_{\blambda_-})
& \cov_\bzero(\Upsilon_{\blambda_+},\Upsilon_{\blambda_-}) 
& 2\V_\bzero[\Upsilon_{\blambda_-}]
\end{array}
\right)^{-1}
\left(
\begin{array}{c}
\cov_\bzero(\zeta\Upsilon_\blambda,\Upsilon_\blambda)\\
\cov_\bzero(\zeta\Upsilon_\blambda,\Upsilon_{\blambda_+})\\
\cov_\bzero(\zeta\Upsilon_\blambda,\Upsilon_{\blambda_-})
\end{array}
\right).
\]

In practice, these covariances are estimated from the same samples that were drawn to estimate $\e_\bzero[\zeta\Upsilon_\blambda]$. The additional cost is (nearly) negligible as all the variables in the exponential martingales are by-products of \nameref{alg:TSB}. It is difficult to establish theoretical guarantees for the improvement of $\tilde \zeta$ over $\zeta\Upsilon_\blambda$. However, since most of the variance of the estimator based on $\zeta\Upsilon_\blambda$ comes from $\Upsilon_\blambda$, the proposal $\tilde \zeta$ typically performs well in applications, see e.g. the CIs in Figures~\ref{fig:ARA-Convergence-2} and~\ref{fig:ARA-Convergence-3} for a tempered stable process. 

\section{Extrema of tempered stable processes}
\label{sec:stable}

\subsection{Description of the model.}
\label{subsec:prelim-ts}

In the present section we apply our results to the tempered stable 
process $X$. More precisely, given 
$\blambda\in\R_+^2\setminus\{\bzero\}$, the tempered stable 
L\'evy measure $\nu_\blambda$ specifies the L\'evy measure $\nu$ 
via~\eqref{eq:lambda}: 
\begin{equation}
\label{eq:nu_stable}
\frac{\nu(\D x)}{\D x}
=\frac{c_+}{x^{\alpha_++1}}
		\cdot\1_{(0,\infty)}(x)
	+ \frac{c_-}{|x|^{\alpha_-+1}}
		\cdot\1_{(-\infty,0)}(x),
\end{equation}
where $c_\pm\ge 0$ and $\alpha_\pm\in(0,2)$. The drift $b$ 
is given by the tempered stable drift $b_\blambda\in\R$ 
via~\eqref{eq:lambda} and the constant $\mu_\blambda$ 
is given in~\eqref{eq:RN-drift}. Both $b$ and $\mu_\blambda$ 
can be computed using~\eqref{eq:b_stable} and~\eqref{eq:c_stable} 
below. \nameref{alg:TSB} samples from the distribution functions 
$F^\ppm(t,x)=\p_\bzero(Y_t^\ppm\le x)$, where $Y^\ppm$ are 
the spectrally one-sided stable processes defined in 
Theorem~\ref{thm:chi_exp_temp}, using the 
Chambers-Mallows-Stuck algorithm~\cite{MR415982}. 
We included a version of  this algorithm in the appendix, given 
explicitly in terms of the drift $b$ and the parameters in the L\'evy 
measure $\nu$, see Algorithm~\ref{alg:cms} below.

Next, we provide explicit formulae for $b$ and $\mu_\blambda$ 
in terms of special functions. We begin by expressing $b$ in terms of 
$b_\blambda$ (see~\eqref{eq:lambda} above): 
\begin{equation}\label{eq:b_stable}
b=b_\blambda 
-c_+B_{\alpha_+,\lambda_+}+c_-B_{\alpha_-,\lambda_-},
\quad\text{where}\quad
B_{a,r} := \int_0^1 (e^{-rx}-1)x^{-a}\D x.
\end{equation}
We have $B_{a,0}=0$ for any $a\ge0$ and, for $r>0$,  
\[
B_{a,r} 
=\sum_{n=1}^\infty \frac{(-r)^n}{n!(n-a-1)}
=\begin{cases}
(e^{-r}-1+r^{a-1}\gamma(2-a,r))/(1-a),
&a\in(0,2)\setminus\{1\},\\
-\gamma-\log r-E_1(r),
&a=1,
\end{cases}
\]
where $\gamma(a,r)=\int_0^r e^{-x}x^{a-1}\D x
=\sum_{n=0}^\infty (-1)^{n}r^{n+a}/(n!(n+a))$, $a>0$, is the lower 
incomplete gamma function, $E_1(r)=\int_r^\infty e^{-x}x^{-1}\D x$, 
$r>0$, is the exponential integral and $\gamma=0.57721566\ldots$ 
is the Euler-Mascheroni constant. 
Similarly, to compute $\mu_\blambda$ from~\eqref{eq:RN-drift}, 
note that 
\begin{equation}\label{eq:c_stable}
\mu_\blambda =c_+C_{\alpha_+,\lambda_+}+c_-C_{\alpha_-,\lambda_-},
\quad\text{where}\quad
C_{a,r} 
:= \int_0^\infty (e^{-rx}-1+rx\cdot \1_{(0,1)}(x))x^{-a-1}\D x.
\end{equation}
Clearly, $C_{a,0}=0$ for any $a\ge0$ and, for $r>0$, 
\begin{align*}
C_{a,r} 
&=-\frac{1}{a}\Big(e^{-r}+r(1+B_{a,r})\Big)+r^{a}\Gamma(-a,r),
\end{align*}
where $\Gamma(a,r)=\int_r^\infty e^{-x}x^{a-1}dx$ is the upper 
incomplete gamma function. When $a<1$, this simplifies to 
$C_{a,r}=r^a \Gamma(-a) + r/(1-a)$, where $\Gamma(\cdot)$ is the 
gamma function. 

As discussed in Subsection~\ref{subsec:SB_SBG_comparison} (see also Table~\ref{tab:TSB_SBA} above), the performance of \nameref{alg:TSB} deteriorates for large values of the constant $\mu_{2\blambda}-2\mu_{\blambda}$. As a consequence of the formulae above, it is easy to see that this constant is proportional to $\lambda_+^{\alpha_+}/(2-\alpha_+) + \lambda_-^{\alpha_-}/(2-\alpha_-)$ as either $\lambda_\pm\to\infty$ or $\alpha_\pm\to 2$, see Figure~\ref{fig:c_lambda}.

\begin{figure}[ht]
	\begin{center}
		{\scalefont{.8}
			\begin{tikzpicture} 
				\begin{axis} 
					[
					title={Constant $\mu_{2\blambda}-2\mu_\blambda$},
					ymin=0,
					ymax=33,
					xmin=0,
					xmax=9.99,
					xlabel={$l$},
					width=7.5cm,
					height=4.5cm,
					axis on top=true,
					axis x line=bottom, 
					axis y line=right,
					axis line style={->},
					x label style={at={(axis description cs:1,0.1)},anchor=north},
					legend style={at={(.3,.39)},anchor=south east}
					]
					
					\addplot[
					densely dotted, 
					color=black,
					]
					coordinates {
						(0.1353,6.644)(0.1423,6.71)(0.1496,6.778)(0.1572,6.846)(0.1653,6.915)(0.1738,6.984)(0.1827,7.054)(0.192,7.125)(0.2019,7.197)(0.2122,7.269)(0.2231,7.342)(0.2346,7.416)(0.2466,7.491)(0.2592,7.566)(0.2725,7.642)(0.2865,7.719)(0.3012,7.796)(0.3166,7.875)(0.3329,7.954)(0.3499,8.034)(0.3679,8.115)(0.3867,8.196)(0.4066,8.278)(0.4274,8.362)(0.4493,8.446)(0.4724,8.531)(0.4966,8.616)(0.522,8.703)(0.5488,8.79)(0.5769,8.879)(0.6065,8.968)(0.6376,9.058)(0.6703,9.149)(0.7047,9.241)(0.7408,9.334)(0.7788,9.428)(0.8187,9.522)(0.8607,9.618)(0.9048,9.715)(0.9512,9.812)(1.0,9.911)(1.051,10.01)(1.105,10.11)(1.162,10.21)(1.221,10.32)(1.284,10.42)(1.35,10.52)(1.419,10.63)(1.492,10.74)(1.568,10.84)(1.649,10.95)(1.733,11.06)(1.822,11.17)(1.916,11.29)(2.014,11.4)(2.117,11.52)(2.226,11.63)(2.34,11.75)(2.46,11.87)(2.586,11.99)(2.718,12.11)(2.858,12.23)(3.004,12.35)(3.158,12.47)(3.32,12.6)(3.49,12.73)(3.669,12.85)(3.857,12.98)(4.055,13.11)(4.263,13.25)(4.482,13.38)(4.711,13.51)(4.953,13.65)(5.207,13.79)(5.474,13.92)(5.755,14.06)(6.05,14.21)(6.36,14.35)(6.686,14.49)(7.029,14.64)(7.389,14.79)(7.768,14.93)(8.166,15.08)(8.585,15.24)(9.025,15.39)(9.488,15.54)(9.974,15.7)
					};
					
					\addplot[
					dashdotted, 
					color=black,
					]
					coordinates {
						(0.1353,2.277)(0.1423,2.323)(0.1496,2.37)(0.1572,2.417)(0.1653,2.466)(0.1738,2.516)(0.1827,2.567)(0.192,2.619)(0.2019,2.672)(0.2122,2.726)(0.2231,2.781)(0.2346,2.837)(0.2466,2.894)(0.2592,2.953)(0.2725,3.012)(0.2865,3.073)(0.3012,3.135)(0.3166,3.199)(0.3329,3.263)(0.3499,3.329)(0.3679,3.396)(0.3867,3.465)(0.4066,3.535)(0.4274,3.606)(0.4493,3.679)(0.4724,3.754)(0.4966,3.829)(0.522,3.907)(0.5488,3.986)(0.5769,4.066)(0.6065,4.148)(0.6376,4.232)(0.6703,4.318)(0.7047,4.405)(0.7408,4.494)(0.7788,4.585)(0.8187,4.677)(0.8607,4.772)(0.9048,4.868)(0.9512,4.967)(1.0,5.067)(1.051,5.169)(1.105,5.274)(1.162,5.38)(1.221,5.489)(1.284,5.6)(1.35,5.713)(1.419,5.828)(1.492,5.946)(1.568,6.066)(1.649,6.189)(1.733,6.314)(1.822,6.441)(1.916,6.571)(2.014,6.704)(2.117,6.84)(2.226,6.978)(2.34,7.119)(2.46,7.263)(2.586,7.409)(2.718,7.559)(2.858,7.712)(3.004,7.867)(3.158,8.026)(3.32,8.189)(3.49,8.354)(3.669,8.523)(3.857,8.695)(4.055,8.871)(4.263,9.05)(4.482,9.233)(4.711,9.419)(4.953,9.609)(5.207,9.803)(5.474,10.0)(5.755,10.2)(6.05,10.41)(6.36,10.62)(6.686,10.83)(7.029,11.05)(7.389,11.28)(7.768,11.5)(8.166,11.74)(8.585,11.97)(9.025,12.22)(9.488,12.46)(9.974,12.71)
					};
					
					\addplot[
					dashed, 
					color=black,
					]
					coordinates {
						(0.1353,1.078)(0.1423,1.111)(0.1496,1.145)(0.1572,1.18)(0.1653,1.216)(0.1738,1.253)(0.1827,1.291)(0.192,1.331)(0.2019,1.371)(0.2122,1.413)(0.2231,1.456)(0.2346,1.5)(0.2466,1.546)(0.2592,1.593)(0.2725,1.641)(0.2865,1.691)(0.3012,1.743)(0.3166,1.796)(0.3329,1.851)(0.3499,1.907)(0.3679,1.965)(0.3867,2.025)(0.4066,2.087)(0.4274,2.15)(0.4493,2.216)(0.4724,2.283)(0.4966,2.353)(0.522,2.424)(0.5488,2.498)(0.5769,2.574)(0.6065,2.653)(0.6376,2.733)(0.6703,2.817)(0.7047,2.902)(0.7408,2.991)(0.7788,3.082)(0.8187,3.176)(0.8607,3.273)(0.9048,3.372)(0.9512,3.475)(1.0,3.581)(1.051,3.69)(1.105,3.802)(1.162,3.918)(1.221,4.037)(1.284,4.16)(1.35,4.287)(1.419,4.417)(1.492,4.552)(1.568,4.691)(1.649,4.833)(1.733,4.981)(1.822,5.132)(1.916,5.289)(2.014,5.45)(2.117,5.616)(2.226,5.787)(2.34,5.963)(2.46,6.145)(2.586,6.332)(2.718,6.525)(2.858,6.723)(3.004,6.928)(3.158,7.139)(3.32,7.356)(3.49,7.58)(3.669,7.811)(3.857,8.049)(4.055,8.294)(4.263,8.547)(4.482,8.807)(4.711,9.075)(4.953,9.352)(5.207,9.637)(5.474,9.93)(5.755,10.23)(6.05,10.54)(6.36,10.87)(6.686,11.2)(7.029,11.54)(7.389,11.89)(7.768,12.25)(8.166,12.62)(8.585,13.01)(9.025,13.4)(9.488,13.81)(9.974,14.23)
					};
					
					\addplot[
					densely dashed, 
					color=black,
					]
					coordinates {
						(0.1353,0.5999)(0.1423,0.6244)(0.1496,0.6499)(0.1572,0.6764)(0.1653,0.704)(0.1738,0.7327)(0.1827,0.7626)(0.192,0.7938)(0.2019,0.8262)(0.2122,0.8599)(0.2231,0.895)(0.2346,0.9315)(0.2466,0.9695)(0.2592,1.009)(0.2725,1.05)(0.2865,1.093)(0.3012,1.138)(0.3166,1.184)(0.3329,1.232)(0.3499,1.283)(0.3679,1.335)(0.3867,1.39)(0.4066,1.446)(0.4274,1.505)(0.4493,1.567)(0.4724,1.631)(0.4966,1.697)(0.522,1.767)(0.5488,1.839)(0.5769,1.914)(0.6065,1.992)(0.6376,2.073)(0.6703,2.158)(0.7047,2.246)(0.7408,2.337)(0.7788,2.433)(0.8187,2.532)(0.8607,2.635)(0.9048,2.743)(0.9512,2.855)(1.0,2.971)(1.051,3.093)(1.105,3.219)(1.162,3.35)(1.221,3.487)(1.284,3.629)(1.35,3.777)(1.419,3.932)(1.492,4.092)(1.568,4.259)(1.649,4.433)(1.733,4.614)(1.822,4.802)(1.916,4.998)(2.014,5.202)(2.117,5.414)(2.226,5.635)(2.34,5.865)(2.46,6.105)(2.586,6.354)(2.718,6.613)(2.858,6.883)(3.004,7.164)(3.158,7.456)(3.32,7.76)(3.49,8.077)(3.669,8.407)(3.857,8.75)(4.055,9.107)(4.263,9.479)(4.482,9.865)(4.711,10.27)(4.953,10.69)(5.207,11.12)(5.474,11.58)(5.755,12.05)(6.05,12.54)(6.36,13.05)(6.686,13.59)(7.029,14.14)(7.389,14.72)(7.768,15.32)(8.166,15.94)(8.585,16.59)(9.025,17.27)(9.488,17.98)(9.974,18.71)
					};
					
					\addplot[
					solid, 
					color=black,
					]
					coordinates {
						(0.1353,0.3752)(0.1423,0.3945)(0.1496,0.4147)(0.1572,0.436)(0.1653,0.4583)(0.1738,0.4818)(0.1827,0.5065)(0.192,0.5325)(0.2019,0.5598)(0.2122,0.5885)(0.2231,0.6186)(0.2346,0.6504)(0.2466,0.6837)(0.2592,0.7188)(0.2725,0.7556)(0.2865,0.7944)(0.3012,0.8351)(0.3166,0.8779)(0.3329,0.9229)(0.3499,0.9702)(0.3679,1.02)(0.3867,1.072)(0.4066,1.127)(0.4274,1.185)(0.4493,1.246)(0.4724,1.31)(0.4966,1.377)(0.522,1.447)(0.5488,1.522)(0.5769,1.6)(0.6065,1.682)(0.6376,1.768)(0.6703,1.859)(0.7047,1.954)(0.7408,2.054)(0.7788,2.159)(0.8187,2.27)(0.8607,2.386)(0.9048,2.509)(0.9512,2.637)(1.0,2.773)(1.051,2.915)(1.105,3.064)(1.162,3.221)(1.221,3.386)(1.284,3.56)(1.35,3.743)(1.419,3.934)(1.492,4.136)(1.568,4.348)(1.649,4.571)(1.733,4.806)(1.822,5.052)(1.916,5.311)(2.014,5.583)(2.117,5.87)(2.226,6.171)(2.34,6.487)(2.46,6.819)(2.586,7.169)(2.718,7.537)(2.858,7.923)(3.004,8.329)(3.158,8.756)(3.32,9.205)(3.49,9.677)(3.669,10.17)(3.857,10.7)(4.055,11.24)(4.263,11.82)(4.482,12.43)(4.711,13.06)(4.953,13.73)(5.207,14.44)(5.474,15.18)(5.755,15.96)(6.05,16.77)(6.36,17.63)(6.686,18.54)(7.029,19.49)(7.389,20.49)(7.768,21.54)(8.166,22.64)(8.585,23.8)(9.025,25.02)(9.488,26.31)(9.974,27.65)
					};
					
					\legend {
						$\alpha=.2$, $\alpha=.4$, $\alpha=.6$, $\alpha=.8$, $\alpha=1$};
				\end{axis}
			\end{tikzpicture}
			\begin{tikzpicture} 
				\begin{axis} 
					[
					title={Constant $\mu_{2\blambda}-2\mu_\blambda$},
					ymin=0,
					ymax=180,
					xmin=0,
					xmax=4.99,
					xlabel={$l$},
					width=7.5cm,
					height=4.5cm,
					axis on top=true,
					axis x line=bottom, 
					axis y line=right,
					axis line style={->},
					x label style={at={(axis description cs:1,0.1)},anchor=north},
					legend style={at={(.3,.39)},anchor=south east}
					]
					
					\addplot[
					densely dotted, 
					color=black,
					]
					coordinates {
						(0.1353,0.3752)(0.1423,0.3945)(0.1496,0.4147)(0.1572,0.436)(0.1653,0.4583)(0.1738,0.4818)(0.1827,0.5065)(0.192,0.5325)(0.2019,0.5598)(0.2122,0.5885)(0.2231,0.6186)(0.2346,0.6504)(0.2466,0.6837)(0.2592,0.7188)(0.2725,0.7556)(0.2865,0.7944)(0.3012,0.8351)(0.3166,0.8779)(0.3329,0.9229)(0.3499,0.9702)(0.3679,1.02)(0.3867,1.072)(0.4066,1.127)(0.4274,1.185)(0.4493,1.246)(0.4724,1.31)(0.4966,1.377)(0.522,1.447)(0.5488,1.522)(0.5769,1.6)(0.6065,1.682)(0.6376,1.768)(0.6703,1.859)(0.7047,1.954)(0.7408,2.054)(0.7788,2.159)(0.8187,2.27)(0.8607,2.386)(0.9048,2.509)(0.9512,2.637)(1.0,2.773)(1.051,2.915)(1.105,3.064)(1.162,3.221)(1.221,3.386)(1.284,3.56)(1.35,3.743)(1.419,3.934)(1.492,4.136)(1.568,4.348)(1.649,4.571)(1.733,4.806)(1.822,5.052)(1.916,5.311)(2.014,5.583)(2.117,5.87)(2.226,6.171)(2.34,6.487)(2.46,6.819)(2.586,7.169)(2.718,7.537)(2.858,7.923)(3.004,8.329)(3.158,8.756)(3.32,9.205)(3.49,9.677)(3.669,10.17)(3.857,10.7)(4.055,11.24)(4.263,11.82)(4.482,12.43)(4.711,13.06)(4.953,13.73)(5.207,14.44)(5.474,15.18)(5.755,15.96)(6.05,16.77)(6.36,17.63)(6.686,18.54)(7.029,19.49)(7.389,20.49)(7.768,21.54)(8.166,22.64)(8.585,23.8)(9.025,25.02)(9.488,26.31)(9.974,27.65)
					};
					
					\addplot[
					dashdotted, 
					color=black,
					]
					coordinates {
						(0.1353,0.2618)(0.1423,0.2779)(0.1496,0.2951)(0.1572,0.3134)(0.1653,0.3327)(0.1738,0.3533)(0.1827,0.3752)(0.192,0.3984)(0.2019,0.423)(0.2122,0.4492)(0.2231,0.4769)(0.2346,0.5064)(0.2466,0.5377)(0.2592,0.571)(0.2725,0.6063)(0.2865,0.6438)(0.3012,0.6836)(0.3166,0.7259)(0.3329,0.7708)(0.3499,0.8184)(0.3679,0.869)(0.3867,0.9228)(0.4066,0.9798)(0.4274,1.04)(0.4493,1.105)(0.4724,1.173)(0.4966,1.246)(0.522,1.323)(0.5488,1.404)(0.5769,1.491)(0.6065,1.583)(0.6376,1.681)(0.6703,1.785)(0.7047,1.896)(0.7408,2.013)(0.7788,2.137)(0.8187,2.27)(0.8607,2.41)(0.9048,2.559)(0.9512,2.717)(1.0,2.885)(1.051,3.064)(1.105,3.253)(1.162,3.454)(1.221,3.668)(1.284,3.895)(1.35,4.136)(1.419,4.391)(1.492,4.663)(1.568,4.951)(1.649,5.257)(1.733,5.582)(1.822,5.928)(1.916,6.294)(2.014,6.683)(2.117,7.097)(2.226,7.536)(2.34,8.002)(2.46,8.496)(2.586,9.022)(2.718,9.58)(2.858,10.17)(3.004,10.8)(3.158,11.47)(3.32,12.18)(3.49,12.93)(3.669,13.73)(3.857,14.58)(4.055,15.48)(4.263,16.44)(4.482,17.46)(4.711,18.53)(4.953,19.68)(5.207,20.9)(5.474,22.19)(5.755,23.56)(6.05,25.02)(6.36,26.57)(6.686,28.21)(7.029,29.95)(7.389,31.81)(7.768,33.77)(8.166,35.86)(8.585,38.08)(9.025,40.43)(9.488,42.93)(9.974,45.59)
					};
					
					\addplot[
					dashed, 
					color=black,
					]
					coordinates {
						(0.1353,0.2067)(0.1423,0.2217)(0.1496,0.2377)(0.1572,0.255)(0.1653,0.2735)(0.1738,0.2933)(0.1827,0.3145)(0.192,0.3374)(0.2019,0.3618)(0.2122,0.388)(0.2231,0.4162)(0.2346,0.4464)(0.2466,0.4787)(0.2592,0.5134)(0.2725,0.5507)(0.2865,0.5906)(0.3012,0.6334)(0.3166,0.6793)(0.3329,0.7286)(0.3499,0.7814)(0.3679,0.8381)(0.3867,0.8989)(0.4066,0.964)(0.4274,1.034)(0.4493,1.109)(0.4724,1.189)(0.4966,1.276)(0.522,1.368)(0.5488,1.467)(0.5769,1.574)(0.6065,1.688)(0.6376,1.81)(0.6703,1.941)(0.7047,2.082)(0.7408,2.233)(0.7788,2.395)(0.8187,2.569)(0.8607,2.755)(0.9048,2.955)(0.9512,3.169)(1.0,3.399)(1.051,3.645)(1.105,3.909)(1.162,4.193)(1.221,4.497)(1.284,4.823)(1.35,5.173)(1.419,5.548)(1.492,5.95)(1.568,6.381)(1.649,6.844)(1.733,7.34)(1.822,7.872)(1.916,8.443)(2.014,9.056)(2.117,9.712)(2.226,10.42)(2.34,11.17)(2.46,11.98)(2.586,12.85)(2.718,13.78)(2.858,14.78)(3.004,15.85)(3.158,17.0)(3.32,18.24)(3.49,19.56)(3.669,20.98)(3.857,22.5)(4.055,24.13)(4.263,25.88)(4.482,27.75)(4.711,29.77)(4.953,31.92)(5.207,34.24)(5.474,36.72)(5.755,39.38)(6.05,42.24)(6.36,45.3)(6.686,48.59)(7.029,52.11)(7.389,55.89)(7.768,59.94)(8.166,64.29)(8.585,68.95)(9.025,73.95)(9.488,79.31)(9.974,85.06)
					};
					
					\addplot[
					densely dashed, 
					color=black,
					]
					coordinates {
						(0.1353,0.1943)(0.1423,0.2105)(0.1496,0.228)(0.1572,0.247)(0.1653,0.2676)(0.1738,0.2898)(0.1827,0.314)(0.192,0.3401)(0.2019,0.3685)(0.2122,0.3992)(0.2231,0.4324)(0.2346,0.4684)(0.2466,0.5074)(0.2592,0.5497)(0.2725,0.5955)(0.2865,0.6451)(0.3012,0.6988)(0.3166,0.757)(0.3329,0.82)(0.3499,0.8883)(0.3679,0.9623)(0.3867,1.042)(0.4066,1.129)(0.4274,1.223)(0.4493,1.325)(0.4724,1.436)(0.4966,1.555)(0.522,1.685)(0.5488,1.825)(0.5769,1.977)(0.6065,2.142)(0.6376,2.32)(0.6703,2.513)(0.7047,2.723)(0.7408,2.949)(0.7788,3.195)(0.8187,3.461)(0.8607,3.749)(0.9048,4.062)(0.9512,4.4)(1.0,4.766)(1.051,5.163)(1.105,5.593)(1.162,6.059)(1.221,6.564)(1.284,7.111)(1.35,7.703)(1.419,8.344)(1.492,9.039)(1.568,9.792)(1.649,10.61)(1.733,11.49)(1.822,12.45)(1.916,13.49)(2.014,14.61)(2.117,15.83)(2.226,17.14)(2.34,18.57)(2.46,20.12)(2.586,21.79)(2.718,23.61)(2.858,25.57)(3.004,27.7)(3.158,30.01)(3.32,32.51)(3.49,35.22)(3.669,38.15)(3.857,41.33)(4.055,44.77)(4.263,48.5)(4.482,52.54)(4.711,56.92)(4.953,61.66)(5.207,66.79)(5.474,72.36)(5.755,78.38)(6.05,84.91)(6.36,91.98)(6.686,99.64)(7.029,107.9)(7.389,116.9)(7.768,126.7)(8.166,137.2)(8.585,148.7)(9.025,161.0)(9.488,174.4)(9.974,189.0)
					};
					
					\addplot[
					solid, 
					color=black,
					]
					coordinates {
						(0.1353,0.2582)(0.1423,0.2825)(0.1496,0.3092)(0.1572,0.3383)(0.1653,0.3701)(0.1738,0.405)(0.1827,0.4431)(0.192,0.4849)(0.2019,0.5305)(0.2122,0.5805)(0.2231,0.6351)(0.2346,0.695)(0.2466,0.7604)(0.2592,0.832)(0.2725,0.9104)(0.2865,0.9961)(0.3012,1.09)(0.3166,1.193)(0.3329,1.305)(0.3499,1.428)(0.3679,1.562)(0.3867,1.709)(0.4066,1.87)(0.4274,2.046)(0.4493,2.239)(0.4724,2.45)(0.4966,2.681)(0.522,2.933)(0.5488,3.209)(0.5769,3.512)(0.6065,3.842)(0.6376,4.204)(0.6703,4.6)(0.7047,5.033)(0.7408,5.507)(0.7788,6.026)(0.8187,6.594)(0.8607,7.215)(0.9048,7.894)(0.9512,8.637)(1.0,9.451)(1.051,10.34)(1.105,11.31)(1.162,12.38)(1.221,13.55)(1.284,14.82)(1.35,16.22)(1.419,17.74)(1.492,19.42)(1.568,21.24)(1.649,23.25)(1.733,25.43)(1.822,27.83)(1.916,30.45)(2.014,33.32)(2.117,36.46)(2.226,39.89)(2.34,43.65)(2.46,47.76)(2.586,52.25)(2.718,57.17)(2.858,62.56)(3.004,68.45)(3.158,74.9)(3.32,81.95)(3.49,89.67)(3.669,98.11)(3.857,107.4)(4.055,117.5)(4.263,128.5)(4.482,140.6)(4.711,153.9)(4.953,168.4)(5.207,184.2)(5.474,201.6)(5.755,220.5)(6.05,241.3)(6.36,264.0)(6.686,288.9)(7.029,316.1)(7.389,345.9)(7.768,378.5)(8.166,414.1)(8.585,453.1)(9.025,495.8)(9.488,542.5)(9.974,593.5)
					};
					
					\legend {
						$\alpha=1$, $\alpha=1.2$, $\alpha=1.4$, $\alpha=1.6$, $\alpha=1.8$};
				\end{axis}
			\end{tikzpicture}
		}\caption{\small 
			The picture shows the map $l\mapsto \mu_\blambda$, where 
			$c_+=c_-=1$, $\alpha_+=\alpha_-=\alpha$, $\blambda=(l,l)$ 
			and $\alpha\in\{0.2,\ldots,1.8\}$. 
		}\label{fig:c_lambda}
	\end{center}
\end{figure}
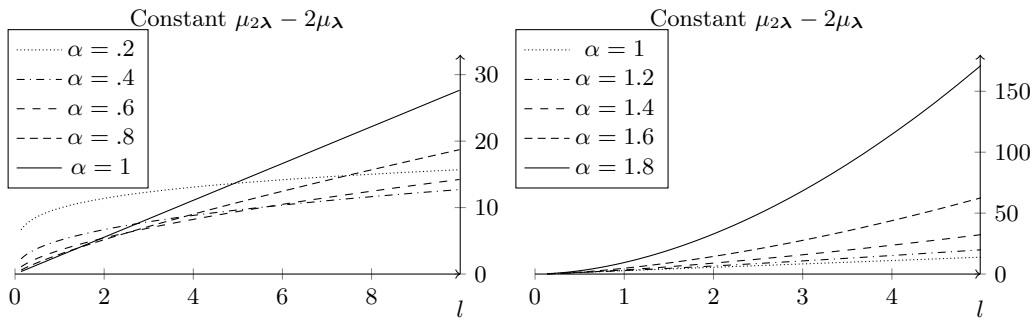

It is natural to expect the performance of \nameref{alg:TSB} to 
deteriorate as $\lambda_\pm\to\infty$ or $\alpha_\pm\to 2$ 
because the variance of the Radon-Nikodym derivative 
$\Upsilon_\blambda$ tends to $\infty$, making the variance of 
all the estimators very large. 

\subsection{Monte Carlo and multilevel Monte Carlo estimators for tempered stable processes}
\label{subsec:MC_stable}

Let $X$ denote the tempered stable process with generating triplet 
$(\sigma^2,\nu_\blambda,b_\blambda)$ given in 
Subsection~\ref{subsec:prelim-ts}. Then $S_t=S_0e^{X_t}$ models 
the risky asset for some initial value $S_0>0$ and all $t\ge 0$. 
Observe that $\ov S_t=S_0 e^{\ov X_t}$ and that the time at which 
$S$ attains its supremum on $[0,t]$ is also $\tau_t$. Since the tails 
$\nu(\R\setminus(-x,x))$ of $\nu$ decay polynomially 
as $x\to\infty$ (see~\eqref{eq:nu_stable}), $S_t$ and $\ov S_t$ 
satisfy the following moment conditions under $\p_\blambda$: 
$\e_\blambda[S_t^r]<\infty$ 
(resp. $\e_\blambda[\ov{S}_t^r]<\infty$) 
if and only if $r\in[-\lambda_-,\lambda_+]$ 
(resp. $r\le \lambda_+$). 
In the present subsection we apply \nameref{alg:TSB} 
to estimate $\e_\blambda [g(\chi)]$ using the MC and MLMC 
estimators in~\eqref{eq:MC} and~\eqref{eq:MLMC}, respectively, 
for payoffs $g$ that arise in applications and calibrated/estimated 
values of tempered stable model parameters. 

\subsubsection{Monte Carlo estimators}
\label{subsec:MC} 

Consider the estimator in~\eqref{eq:MC} for the  following payoff functions $g$: (I) the payoff of the up-and-out barrier call option $g(\chi)=\max\{S_T-K,0\}\1{\{\ov S_T\le M\}}$ and (II) the function $g(\chi)=(S_T/\ov S_T-1)^2$ associated to the ulcer index (UI) given by $100\sqrt{\e_\blambda g(\chi)}$ (see~\cite{Investopedia_UI}). We plot the estimated value along with a band, representing the empirically constructed CIs of level $95\%$ via Theorem~\ref{thm:CLT}, see Figures~\ref{fig:ARA-Convergence-1} and~\ref{fig:ARA-Convergence-2}. We set the initial value $S_0=100$, assume $b_\blambda=0$ and consider multiple time horizons $T$ (given as fractions of a calendar year): we take $T\in\{\frac{30}{365},\frac{90}{365}\}$ in (I) and $T\in\{\frac{14}{365},\frac{28}{365}\}$ in (II). These time horizons are common in their respective applications (I) option pricing and (II) real-world risk assessment (see, e.g.~\cite{MR3038608,CGMY}). The values of the other parameters, given in Table~\ref{tab:fit} below, are either calibrated or estimated: in (I) we are pricing an option and thus use calibrated parameters with respect to the risk-neutral measure obtained in~\cite[Table~3]{MR3038608} for the USD/JPY foreign exchange (FX) rate and in (II) we are assessing risks under the real-world probability measure and hence use the maximum likelihood estimates in~\cite[Table~2]{CGMY} for various time series of historic asset returns. In both, (I) and (II), it is assumed that $\alpha_\pm=\alpha$; (II) additionally assumes that $c_+=c_-$ (i.e., $X$ is a CGMY process). The stocks MCD, BIX and SOX in (II) were chosen with $\alpha>1$ to stress test the performance of \nameref{alg:TSB} when $\mu_{2\blambda}-2\mu_\blambda$ is big, forcing the variance of the estimator to be large, see the discussion at the end of Subsection~\ref{subsec:prelim-ts} above and Figure~\ref{fig:c_lambda}. 

\begin{table}[h]
	\begin{tabular}{|c|c|c|c|c|c|c|c|} 
		\hline
		Stock & $\sigma$ & $\alpha$ & $c_+$ & $c_-$ & $\lambda_+$ & $\lambda_-$ & $\mu_{2\blambda}-2\mu_\blambda$ \\
		\hline
		USD/JPY (v1) & 0.0007 & 0.66 & 0.1305 & 0.0615 & 6.5022 & 3.0888 & 0.9658\\
		\hline
		USD/JPY (v2) & 0.0001 & 1.5 & 0.0069 & 0.0063 & 1.932 & 0.4087 & 0.0395\\
		\hline
		MCD & 0 & 1.50683 & 0.08 & 0.08 & 25.4 & 25.4 & 41.47 \\
		\hline
		BIX & 0 & 1.2341 & 0.32 & 0.32 & 37.42 & 47.76 & 96.6 \\
		\hline
		SOX & 0 & 1.3814 & 0.44 & 0.44 & 34.73 & 34.76 & 196.81 \\
		\hline
	\end{tabular}
	\caption{Calibrated/estimated parameters of the tempered stable 
		model. The first two sets of parameters were calibrated 
		in~\cite[Table~3]{MR3038608} based on FX option data. 
		The bottom three sets of parameters were estimated 
		using the MLE in~\cite[Table~2]{CGMY}, based on a time series 
		of equity returns.}
	\label{tab:fit}
\end{table}


\begin{figure}[ht]
	\centering
	\begin{tikzpicture} 
	\begin{axis} 
	[
	title={Maturity: 30 days},
	ymin=1,
	ymax=3,
	xmin=1,
	xmax=12,
	xlabel={\small $n$},
	width=7cm,
	height=4cm,
	axis on top=true,
	axis x line=bottom, 
	axis y line=left,
	axis line style={->},
	x label style={at={(axis description cs:1.05,0.3)},anchor=north},
	legend style={at={(1.3,.15)},anchor=south east}
	]
	
	\addplot[
	solid, mark=o, mark options={scale=.75, solid, thin},
	color=black,
	]
	coordinates {
		(1,2.698)(2,2.612)(3,2.571)(4,2.553)(5,2.544)(6,2.541)(7,2.539)(8,2.538)(9,2.538)(10,2.538)(11,2.538)(12,2.538)(13,2.538)(14,2.538)(15,2.538)
	};
	
	\addplot[
	dotted, mark=o, mark options={scale=.75, solid, thin},
	color=black,
	]
	coordinates {
		(1,2.703)(2,2.617)(3,2.577)(4,2.558)(5,2.55)(6,2.546)(7,2.544)(8,2.543)(9,2.543)(10,2.543)(11,2.543)(12,2.543)(13,2.543)(14,2.543)(15,2.543)
	};
	
	\addplot[
	dotted, mark=o, mark options={scale=.75, solid, thin},
	color=black,
	]
	coordinates {
		(1,2.692)(2,2.607)(3,2.566)(4,2.548)(5,2.539)(6,2.535)(7,2.534)(8,2.533)(9,2.533)(10,2.533)(11,2.532)(12,2.532)(13,2.532)(14,2.532)(15,2.532)
	};
	
	\addplot[
	solid, mark=diamond, mark options={scale=.75, solid, thin},
	color=black,
	]
	coordinates {
		(1,1.698)(2,1.446)(3,1.302)(4,1.221)(5,1.175)(6,1.148)(7,1.132)(8,1.123)(9,1.117)(10,1.114)(11,1.112)(12,1.111)(13,1.11)(14,1.11)(15,1.109)
	};
	
	\addplot[
	dotted, mark=diamond, mark options={scale=.75, solid, thin},
	color=black,
	]
	coordinates {
		(1,1.702)(2,1.449)(3,1.305)(4,1.224)(5,1.178)(6,1.151)(7,1.135)(8,1.126)(9,1.12)(10,1.117)(11,1.115)(12,1.114)(13,1.113)(14,1.113)(15,1.113)
	};
	
	\addplot[
	dotted, mark=diamond, mark options={scale=.75, solid, thin},
	color=black,
	]
	coordinates {
		(1,1.694)(2,1.442)(3,1.298)(4,1.218)(5,1.172)(6,1.145)(7,1.129)(8,1.12)(9,1.114)(10,1.111)(11,1.109)(12,1.108)(13,1.107)(14,1.107)(15,1.106)
	};
	
	\legend {\small v1,,,\small v2};
	\end{axis}
	\end{tikzpicture}
	\begin{tikzpicture} 
	\begin{axis} 
	[
	title={90 days},
	ymin=.25,
	ymax=1.75,
	xmin=1,
	xmax=12,
	xlabel={\small $n$},
	width=7cm,
	height=4cm,
	axis on top=true,
	axis x line=bottom, 
	axis y line=left,
	axis line style={->},
	x label style={at={(axis description cs:1.05,0.3)},anchor=north},
	legend style={at={(1.5,.15)},anchor=south east}
	]
	
	\addplot[
	solid, mark=o, mark options={scale=.75, solid, thin},
	color=black,
	]
	coordinates {
		(1,1.631)(2,1.495)(3,1.431)(4,1.402)(5,1.389)(6,1.382)(7,1.379)(8,1.378)(9,1.377)(10,1.377)(11,1.377)(12,1.376)(13,1.376)(14,1.376)(15,1.376)
	};
	
	\addplot[
	dotted, mark=o, mark options={scale=.75, solid, thin},
	color=black,
	]
	coordinates {
		(1,1.637)(2,1.5)(3,1.436)(4,1.407)(5,1.394)(6,1.387)(7,1.384)(8,1.383)(9,1.382)(10,1.382)(11,1.382)(12,1.382)(13,1.382)(14,1.382)(15,1.382)
	};
	
	\addplot[
	dotted, mark=o, mark options={scale=.75, solid, thin},
	color=black,
	]
	coordinates {
		(1,1.626)(2,1.49)(3,1.426)(4,1.397)(5,1.383)(6,1.377)(7,1.374)(8,1.373)(9,1.372)(10,1.372)(11,1.371)(12,1.371)(13,1.371)(14,1.371)(15,1.371)
	};
	
	\addplot[
	solid, mark=diamond, mark options={scale=.75, solid, thin},
	color=black,
	]
	coordinates {
		(1,0.9426)(2,0.7073)(3,0.5924)(4,0.5328)(5,0.502)(6,0.4853)(7,0.4761)(8,0.4705)(9,0.4673)(10,0.4654)(11,0.4642)(12,0.4635)(13,0.4631)(14,0.4628)(15,0.4626)
	};
	
	\addplot[
	dotted, mark=diamond, mark options={scale=.75, solid, thin},
	color=black,
	]
	coordinates {
		(1,0.9457)(2,0.71)(3,0.5948)(4,0.5352)(5,0.5043)(6,0.4876)(7,0.4783)(8,0.4727)(9,0.4695)(10,0.4676)(11,0.4664)(12,0.4657)(13,0.4652)(14,0.465)(15,0.4648)
	};
	
	\addplot[
	dotted, mark=diamond, mark options={scale=.75, solid, thin},
	color=black,
	]
	coordinates {
		(1,0.9395)(2,0.7047)(3,0.5899)(4,0.5305)(5,0.4997)(6,0.4831)(7,0.4739)(8,0.4683)(9,0.4652)(10,0.4632)(11,0.462)(12,0.4613)(13,0.4609)(14,0.4606)(15,0.4605)
	};
	\end{axis}\end{tikzpicture}
	\caption{\small 
\emph{Solid} lines indicate the estimated value of $\e_\blambda [g(\chi_n)]$ for the USD/JPY currency pair and the function $g(\chi)= \max\{S_T-K,0\}\1{\{\ov S_T \le M\}}$ using $N=10^6$ samples with $T\in\{\frac{30}{365},\frac{90}{365}\}$, $K=95$ and $M=102$. \emph{Dotted} lines (the symmetric bands around each solid line) indicate the confidence interval of the estimates with confidence level $95\%$. These confidence intervals are very tight, making them hard to discern in the plot. 
	}\label{fig:ARA-Convergence-1}
\end{figure}
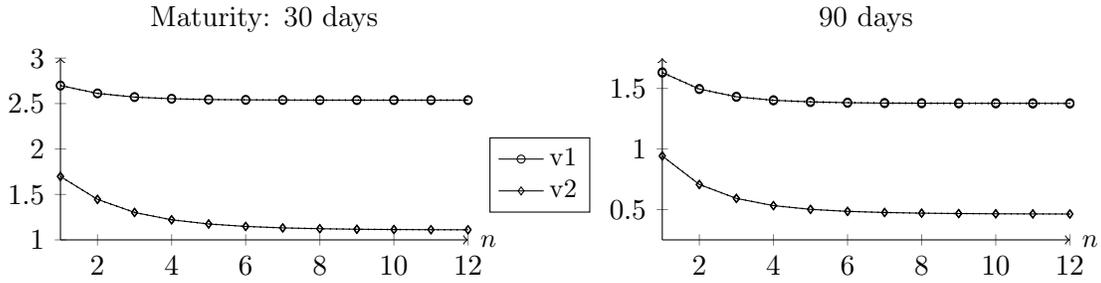

\begin{figure}[ht]
	\centering
	\begin{tikzpicture} 
	\begin{axis} 
	[
	title={Maturity: 14 days},
	ymin=3,
	ymax=7,
	xmin=1,
	xmax=12,
	xlabel={\small $n$},
	width=7.2cm,
	height=4cm,
	axis on top=true,
	axis x line=bottom, 
	axis y line=left,
	axis line style={->},
	x label style={at={(axis description cs:1.05,0.3)},anchor=north},
	legend style={at={(1.4,.15)},anchor=south east}
	]
	
	\addplot[
	solid, mark=o, mark options={scale=.75, solid, thin},
	color=black,
	]
	coordinates {
(1,3.577)(2,3.834)(3,4.009)(4,4.124)(5,4.198)(6,4.245)(7,4.275)(8,4.293)(9,4.304)(10,4.311)(11,4.316)(12,4.318)
	};
	
	\addplot[
	dotted, mark=o, mark options={scale=.75, solid, thin},
	color=black,
	]
	coordinates {
(1,3.581)(2,3.838)(3,4.013)(4,4.128)(5,4.202)(6,4.249)(7,4.278)(8,4.297)(9,4.308)(10,4.315)(11,4.319)(12,4.322)
	};
	
	\addplot[
	dotted, mark=o, mark options={scale=.75, solid, thin},
	color=black,
	]
	coordinates {
(1,3.573)(2,3.831)(3,4.005)(4,4.12)(5,4.194)(6,4.241)(7,4.271)(8,4.289)(9,4.301)(10,4.308)(11,4.312)(12,4.314)
	};
	
	\addplot[
	solid, mark=square, mark options={scale=.75, solid, thin},
	color=black,
	]
	coordinates {
(1,3.087)(2,3.306)(3,3.452)(4,3.545)(5,3.603)(6,3.639)(7,3.661)(8,3.673)(9,3.681)(10,3.685)(11,3.687)(12,3.688)
	};
	
	\addplot[
	dotted, mark=square, mark options={scale=.75, solid, thin},
	color=black,
	]
	coordinates {
(1,3.094)(2,3.313)(3,3.459)(4,3.553)(5,3.611)(6,3.647)(7,3.668)(8,3.681)(9,3.688)(10,3.693)(11,3.695)(12,3.696)
	};
	
	\addplot[
	dotted, mark=square, mark options={scale=.75, solid, thin},
	color=black,
	]
	coordinates {
(1,3.079)(2,3.298)(3,3.444)(4,3.537)(5,3.595)(6,3.631)(7,3.653)(8,3.665)(9,3.673)(10,3.677)(11,3.679)(12,3.681)
	};
	
	\addplot[
	solid, mark=diamond, mark options={scale=.75, solid, thin},
	color=black,
	]
	coordinates {
(1,5.51)(2,5.936)(3,6.227)(4,6.413)(5,6.533)(6,6.61)(7,6.656)(8,6.686)(9,6.705)(10,6.716)(11,6.722)(12,6.726)
	};
	
	\addplot[
	dotted, mark=diamond, mark options={scale=.75, solid, thin},
	color=black,
	]
	coordinates {
(1,5.591)(2,6.018)(3,6.311)(4,6.497)(5,6.618)(6,6.696)(7,6.742)(8,6.772)(9,6.79)(10,6.802)(11,6.808)(12,6.812)
	};
	
	\addplot[
	dotted, mark=diamond, mark options={scale=.75, solid, thin},
	color=black,
	]
	coordinates {
(1,5.428)(2,5.852)(3,6.142)(4,6.328)(5,6.446)(6,6.523)(7,6.569)(8,6.599)(9,6.618)(10,6.629)(11,6.635)(12,6.639)
	};
	
	\legend {\small MCD,,,
		\small BIX,,, 
		\small SOX};
	\end{axis}
	\end{tikzpicture}
	\begin{tikzpicture} 
	\begin{axis} 
	[
	title={Maturity: 28 days},
	ymin=4,
	ymax=10,
	xmin=1,
	xmax=12,
	xlabel={\small $n$},
	width=7.2cm,
	height=4cm,
	axis on top=true,
	axis x line=bottom, 
	axis y line=left,
	axis line style={->},
	x label style={at={(axis description cs:1.05,0.3)},anchor=north},
	legend style={at={(1.4,.15)},anchor=south east}
	]
	
	\addplot[
	solid, mark=o, mark options={scale=.75, solid, thin},
	color=black,
	]
	coordinates {
(1,5.0)(2,5.375)(3,5.63)(4,5.799)(5,5.91)(6,5.981)(7,6.025)(8,6.053)(9,6.07)(10,6.081)(11,6.088)(12,6.092)
	};
	
	\addplot[
	dotted, mark=o, mark options={scale=.75, solid, thin},
	color=black,
	]
	coordinates {
(1,5.011)(2,5.386)(3,5.642)(4,5.811)(5,5.922)(6,5.992)(7,6.037)(8,6.065)(9,6.082)(10,6.093)(11,6.099)(12,6.103)
	};
	
	\addplot[
	dotted, mark=o, mark options={scale=.75, solid, thin},
	color=black,
	]
	coordinates {
(1,4.989)(2,5.363)(3,5.619)(4,5.788)(5,5.899)(6,5.969)(7,6.013)(8,6.041)(9,6.059)(10,6.069)(11,6.076)(12,6.08)
	};
	
	\addplot[
	solid, mark=square, mark options={scale=.75, solid, thin},
	color=black,
	]
	coordinates {
(1,4.309)(2,4.63)(3,4.839)(4,4.979)(5,5.069)(6,5.124)(7,5.158)(8,5.177)(9,5.19)(10,5.197)(11,5.201)(12,5.203)
	};
	
	\addplot[
	dotted, mark=square, mark options={scale=.75, solid, thin},
	color=black,
	]
	coordinates {
(1,4.36)(2,4.681)(3,4.89)(4,5.03)(5,5.121)(6,5.176)(7,5.21)(8,5.229)(9,5.242)(10,5.249)(11,5.253)(12,5.255)
	};
	
	\addplot[
	dotted, mark=square, mark options={scale=.75, solid, thin},
	color=black,
	]
	coordinates {
(1,4.258)(2,4.578)(3,4.787)(4,4.927)(5,5.017)(6,5.072)(7,5.106)(8,5.125)(9,5.138)(10,5.145)(11,5.149)(12,5.151)
	};
	
	\addplot[
	solid, mark=diamond, mark options={scale=.75, solid, thin},
	color=black,
	]
	coordinates {
(1,6.898)(2,7.55)(3,8.029)(4,8.275)(5,8.483)(6,8.586)(7,8.648)(8,8.687)(9,8.709)(10,8.723)(11,8.732)(12,8.737)
	};
	
	\addplot[
	dotted, mark=diamond, mark options={scale=.75, solid, thin},
	color=black,
	]
	coordinates {
(1,7.548)(2,8.238)(3,8.734)(4,8.998)(5,9.238)(6,9.338)(7,9.402)(8,9.441)(9,9.463)(10,9.477)(11,9.486)(12,9.491)
	};
	
	\addplot[
	dotted, mark=diamond, mark options={scale=.75, solid, thin},
	color=black,
	]
	coordinates {
(1,6.182)(2,6.793)(3,7.256)(4,7.481)(5,7.653)(6,7.76)(7,7.823)(8,7.86)(9,7.883)(10,7.898)(11,7.907)(12,7.912)
	};
	\end{axis}\end{tikzpicture}
	\caption{\small 
		\emph{Solid} lines indicate the estimated value of the 
		ulcer index (UI) $100\sqrt{\e_\blambda [g(\chi_n)]}$ using 
		$N=10^7$ samples, where $g(\chi)= (S_T/\ov S_T-1)^2$. 
		\emph{Dotted} lines (the bands around each solid line) 
		indicate the $95\%$ confidence interval of the estimates.
	}\label{fig:ARA-Convergence-2}
\end{figure}
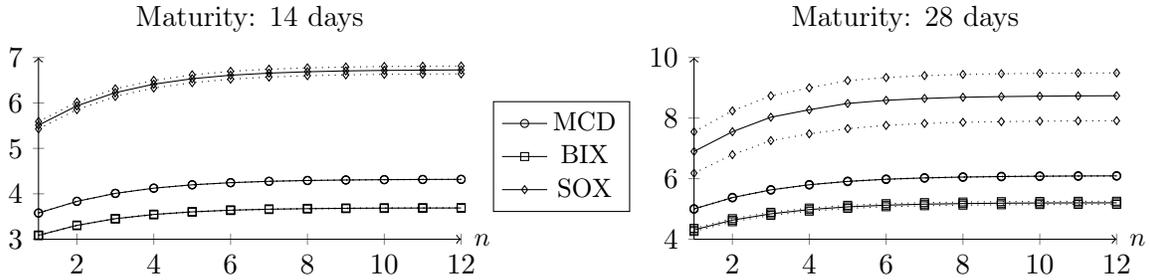

The visible increase in the variance of the estimations for the SOX, see Figure~\ref{fig:ARA-Convergence-2}, as the maturity increases from 14 to 28 days is a consequence of the unusually large value of $\mu_{2\blambda}-2\mu_\blambda$, see Table~\ref{tab:fit} above. If $\mu_{2\blambda}-2\mu_\blambda$ is not large, we may take long time horizons and relatively few samples to obtain an accurate MC estimate. In fact this appears to be the case for calibrated risk-neutral parameters (we were unable to find calibrated parameter values in the literature with a large value of $\mu_{2\blambda}-2\mu_\blambda$). However, if $\mu_{2\blambda}-2\mu_\blambda$ is large, then for any moderately large time horizon, an accurate MC estimate would require a large number of samples. In such cases, the control variates method from Subsection~\ref{subsec:var_red} becomes very useful.

To illustrate the added value of the control variates method from Subsection~\ref{subsec:var_red}, we apply it in the setting of Figure~\ref{fig:ARA-Convergence-2}, where we observed the widest CIs\footnote{Recall that the CIs derived in this paper do not account for model uncertainty or the uncertainty in the estimation or calibration of the parameters.}. Figure~\ref{fig:ARA-Convergence-3} displays the resulting estimators and CIs, showing that this method is beneficial in the case where the variance of the Radon-Nikodym derivative $\Upsilon_\blambda$ is large, i.e., when $(\mu_{2\blambda}-2\mu_\blambda)T$ is large. The confidence intervals for the SOX asset at $n=12$ shrank by a factor of $4.23$, in other words, the variance became $5.58\%$ of its original value. 

\begin{figure}[ht]
	\centering
	\begin{tikzpicture} 
	\begin{axis} 
	[
	title={Maturity: 14 days},
	ymin=3,
	ymax=7,
	xmin=1,
	xmax=12,
	xlabel={\small $n$},
	width=7.2cm,
	height=4cm,
	axis on top=true,
	axis x line=bottom, 
	axis y line=left,
	axis line style={->},
	x label style={at={(axis description cs:1.05,0.3)},anchor=north},
	legend style={at={(1.4,.15)},anchor=south east}
	]
	
	\addplot[
	solid, mark=o, mark options={scale=.75, solid, thin},
	color=black,
	]
	coordinates {
(1,3.576)(2,3.834)(3,4.01)(4,4.124)(5,4.199)(6,4.245)(7,4.274)(8,4.293)(9,4.304)(10,4.311)(11,4.315)(12,4.318)
	};
	
	\addplot[
	dotted, mark=o, mark options={scale=.75, solid, thin},
	color=black,
	]
	coordinates {
(1,3.579)(2,3.837)(3,4.012)(4,4.127)(5,4.201)(6,4.248)(7,4.277)(8,4.295)(9,4.307)(10,4.313)(11,4.318)(12,4.32)
	};
	
	\addplot[
	dotted, mark=o, mark options={scale=.75, solid, thin},
	color=black,
	]
	coordinates {
(1,3.574)(2,3.831)(3,4.007)(4,4.122)(5,4.196)(6,4.243)(7,4.272)(8,4.29)(9,4.301)(10,4.308)(11,4.312)(12,4.315)
	};
	
	\addplot[
	solid, mark=square, mark options={scale=.75, solid, thin},
	color=black,
	]
	coordinates {
(1,3.09)(2,3.307)(3,3.452)(4,3.546)(5,3.605)(6,3.64)(7,3.661)(8,3.674)(9,3.681)(10,3.685)(11,3.688)(12,3.689)
	};
	
	\addplot[
	dotted, mark=square, mark options={scale=.75, solid, thin},
	color=black,
	]
	coordinates {
(1,3.096)(2,3.313)(3,3.458)(4,3.552)(5,3.611)(6,3.646)(7,3.667)(8,3.68)(9,3.687)(10,3.692)(11,3.694)(12,3.695)
	};
	
	\addplot[
	dotted, mark=square, mark options={scale=.75, solid, thin},
	color=black,
	]
	coordinates {
(1,3.083)(2,3.301)(3,3.446)(4,3.54)(5,3.599)(6,3.634)(7,3.655)(8,3.668)(9,3.675)(10,3.679)(11,3.681)(12,3.683)
	};
	
	\addplot[
	solid, mark=diamond, mark options={scale=.75, solid, thin},
	color=black,
	]
	coordinates {
(1,5.614)(2,6.062)(3,6.326)(4,6.512)(5,6.633)(6,6.713)(7,6.763)(8,6.804)(9,6.823)(10,6.833)(11,6.843)(12,6.847)
	};
	
	\addplot[
	dotted, mark=diamond, mark options={scale=.75, solid, thin},
	color=black,
	]
	coordinates {
(1,5.706)(2,6.147)(3,6.409)(4,6.593)(5,6.713)(6,6.794)(7,6.843)(8,6.884)(9,6.904)(10,6.914)(11,6.924)(12,6.927)
	};
	
	\addplot[
	dotted, mark=diamond, mark options={scale=.75, solid, thin},
	color=black,
	]
	coordinates {
(1,5.521)(2,5.975)(3,6.243)(4,6.43)(5,6.552)(6,6.632)(7,6.682)(8,6.723)(9,6.742)(10,6.752)(11,6.762)(12,6.766)
	};
	
	\legend {\small MCD,,,
		\small BIX,,, 
		\small SOX};
	\end{axis}
	\end{tikzpicture}
	\begin{tikzpicture} 
	\begin{axis} 
	[
	title={Maturity: 28 days},
	ymin=4,
	ymax=10,
	xmin=1,
	xmax=12,
	xlabel={\small $n$},
	width=7.2cm,
	height=4cm,
	axis on top=true,
	axis x line=bottom, 
	axis y line=left,
	axis line style={->},
	x label style={at={(axis description cs:1.05,0.3)},anchor=north},
	legend style={at={(1.4,.15)},anchor=south east}
	]
	
	\addplot[
	solid, mark=o, mark options={scale=.75, solid, thin},
	color=black,
	]
	coordinates {
(1,5.005)(2,5.377)(3,5.63)(4,5.799)(5,5.908)(6,5.979)(7,6.023)(8,6.051)(9,6.069)(10,6.079)(11,6.086)(12,6.09)
	};
	
	\addplot[
	dotted, mark=o, mark options={scale=.75, solid, thin},
	color=black,
	]
	coordinates {
(1,5.013)(2,5.386)(3,5.638)(4,5.807)(5,5.917)(6,5.988)(7,6.031)(8,6.06)(9,6.077)(10,6.088)(11,6.094)(12,6.098)
	};
	
	\addplot[
	dotted, mark=o, mark options={scale=.75, solid, thin},
	color=black,
	]
	coordinates {
(1,4.996)(2,5.369)(3,5.622)(4,5.79)(5,5.9)(6,5.971)(7,6.014)(8,6.043)(9,6.06)(10,6.071)(11,6.077)(12,6.081)
	};
	
	\addplot[
	solid, mark=square, mark options={scale=.75, solid, thin},
	color=black,
	]
	coordinates {
(1,4.36)(2,4.712)(3,4.922)(4,5.078)(5,5.165)(6,5.219)(7,5.249)(8,5.269)(9,5.279)(10,5.287)(11,5.291)(12,5.294)
	};
	
	\addplot[
	dotted, mark=square, mark options={scale=.75, solid, thin},
	color=black,
	]
	coordinates {
(1,4.406)(2,4.768)(3,4.978)(4,5.134)(5,5.22)(6,5.273)(7,5.304)(8,5.324)(9,5.334)(10,5.341)(11,5.346)(12,5.348)
	};
	
	\addplot[
	dotted, mark=square, mark options={scale=.75, solid, thin},
	color=black,
	]
	coordinates {
(1,4.313)(2,4.655)(3,4.866)(4,5.023)(5,5.11)(6,5.163)(7,5.194)(8,5.214)(9,5.224)(10,5.232)(11,5.236)(12,5.238)
	};
	
	\addplot[
	solid, mark=diamond, mark options={scale=.75, solid, thin},
	color=black,
	]
	coordinates {
(1,6.373)(2,7.258)(3,7.635)(4,8.013)(5,8.261)(6,8.354)(7,8.435)(8,8.459)(9,8.475)(10,8.498)(11,8.507)(12,8.513)
	};
	
	\addplot[
	dotted, mark=diamond, mark options={scale=.75, solid, thin},
	color=black,
	]
	coordinates {
(1,6.556)(2,7.416)(3,7.806)(4,8.19)(5,8.441)(6,8.537)(7,8.619)(8,8.643)(9,8.659)(10,8.682)(11,8.691)(12,8.697)
	};
	
	\addplot[
	dotted, mark=diamond, mark options={scale=.75, solid, thin},
	color=black,
	]
	coordinates {
(1,6.185)(2,7.097)(3,7.459)(4,7.832)(5,8.077)(6,8.167)(7,8.246)(8,8.271)(9,8.287)(10,8.309)(11,8.319)(12,8.324)
	};
	\end{axis}\end{tikzpicture}
	\caption{\small 
		\emph{Solid} lines indicate the estimated value of the 
		ulcer index (UI) $100\sqrt{\e_\blambda g(\chi_n)}$ using 
		$N=10^7$ samples and the control variates method of 
		Subsection~\ref{subsec:var_red}, where $g(\chi)= (S_T/\ov S_T-1)^2$. 
		\emph{Dotted} lines (the bands around each solid line) 
		indicate the confidence interval of the estimates with 
		confidence level $95\%$.
	}\label{fig:ARA-Convergence-3}
\end{figure}
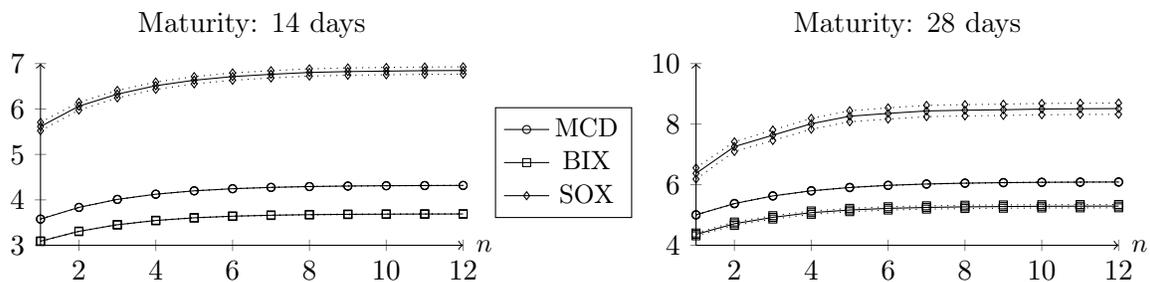

\subsubsection{Multilevel Monte Carlo estimators}

We will consider the MLMC estimator in~\eqref{eq:MLMC} with parameters $n,N_1,\ldots,N_n\in\N$ given by~\cite[Eq.~(A.1)--(A.2)]{SBG}. In this example we chose (I) the payoff $g(\chi)=\max\{S_T-95,0\}\1{\{\ov S_T \le 102\}}$ from Subsection~\ref{subsec:MC} above and (II) the payoff $g(\chi)=(S_T/\ov{S}_T-1)^2\1_{\{\tau_T<T/2\}}$, associated to the modified ulcer index (MUI) $100\sqrt{\e_\blambda[g(\chi)]}$, a risk measure which weighs trends more heavily than short-time fluctuations (see~\cite{SBG} and the references therein). 

The payoff in (I) is that of a barrier up-and-out option, so it is natural to use the risk-neutral parameters for the USD/JPY FX rate (see (v2) in Table~\ref{tab:fit}) over the time horizon $T=90/365$. For (II) we take the parameter values of the MCD stock in Table~\ref{tab:fit} with $T=28/365$. In both cases we set $S_0=100$. Figure~\ref{fig:MLMC_barrier} (resp. Figure~\ref{fig:MLMC_MUI}) shows the decay of the bias and level variance, the corresponding value of the constants $n,N_1,\ldots,N_n$ and the growth of the complexity for the first (resp. second) payoff. 

\begin{figure}[ht]
	\centering
	\begin{tikzpicture}
	\pgfplotsset{
		scale only axis,
		xmin=2, xmax=5
	} 
	\begin{semilogyaxis}[
	width=6.75cm,
	height=2.75cm,
	axis on top=true,
	axis x line=bottom,
	axis y line*=left,
	xlabel = {$k$},
	xmin = 2, xmax = 16,
	ymin = 2e-5, ymax = .5,
	title={Bias decay $|\e_\bzero[g(\chi_k)]-\e_\bzero[g(\chi_{k-1})]|$},
	legend style={at={(0.4,.05)},anchor=south east},
	x label style={at={(axis description cs:1,0.37)},anchor=north},
	]
	\addplot[
	solid,
	mark=+,
	color=black,
	]
	coordinates {(1,1.699)(2,0.2532)(3,0.1418)(4,0.08072)(5,0.04672)(6,0.02667)(7,0.01565)(8,0.009681)(9,0.005353)(10,0.003087)(11,0.001849)(12,0.001249)(13,0.0006588)(14,0.0003933)(15,0.0003069)(16,0.0001511)
	};
	\addlegendentry{\scriptsize Observed}
	\addplot[
	dashed,
	color=black,
	]
	coordinates {(1,0.06742)(2,0.0511)(3,0.03872)(4,0.02935)(5,0.02224)(6,0.01686)(7,0.01277)(8,0.009681)(9,0.007337)(10,0.00556)(11,0.004214)(12,0.003193)(13,0.00242)(14,0.001834)(15,0.00139)(16,0.001053)
	};
	\addlegendentry{\scriptsize Predicted}
	\end{semilogyaxis}
	\end{tikzpicture}
	\begin{tikzpicture}
	\pgfplotsset{
		scale only axis,
		xmin=2, xmax=5
	} 
	\begin{semilogyaxis}[
	width=6.75cm,
	height=2.75cm,
	axis on top=true,
	axis x line=bottom,
	axis y line*=left,
	xlabel = {$k$},
	xmin = 2, xmax = 16,
	ymin = 1e-4, ymax = 2,
	title={ Variance decay $\V_\bzero[g(\chi_k)-g(\chi_{k-1})]$},
	legend style={at={(0.4,.05)},anchor=south east},
	x label style={at={(axis description cs:1,0.37)},anchor=north},
	]
	\addplot[
	solid,
	mark=+,
	color=black,
	]
	coordinates {(1,5.213)(2,1.372)(3,0.7581)(4,0.425)(5,0.2411)(6,0.1355)(7,0.07869)(8,0.04825)(9,0.02692)(10,0.01501)(11,0.009293)(12,0.006123)(13,0.003204)(14,0.002026)(15,0.001555)(16,0.0007908)
	};
	\addlegendentry{\scriptsize Observed}
	\addplot[
	dashed,
	color=black,
	]
	coordinates {(1,0.336)(2,0.2547)(3,0.193)(4,0.1463)(5,0.1109)(6,0.08401)(7,0.06367)(8,0.04825)(9,0.03657)(10,0.02771)(11,0.021)(12,0.01592)(13,0.01206)(14,0.009142)(15,0.006928)(16,0.005251)
	};
	\addlegendentry{\scriptsize Predicted}
	
	\end{semilogyaxis}
	\end{tikzpicture}
	\begin{tikzpicture}
	\pgfplotsset{
		scale only axis,
		xmin=2, xmax=5
	} 
	\begin{semilogyaxis}[
	width=7cm,
	height=2.75cm,
	axis on top=true,
	axis x line=bottom,
	axis y line*=left,
	xlabel = {$k$},
	ylabel = {$\log_2 N_k$},
	xmin = 1, xmax = 16,
	ymin = 1e5, ymax = 1e10,
	title={ Number of levels and samples per level},
	legend style={at={(1,.57)},anchor=south east},
	x label style={at={(axis description cs:1,0.37)},anchor=north},
	y label style={at={(axis description cs:.3,1.05)},anchor=north,rotate=-90},
	]
	\addplot[
	solid,
	mark=o,
	color=black,
	]
	coordinates {(1,13868689)(2,5031021)(3,3053356)(4,1979881)(5,1333831)(6,912625)(7,643989)(8,471721)(9,332188)(10,235324)(11,176543)(12,137199)
	};
	\addlegendentry{\scriptsize $\varepsilon=2^{-9}$}
	\addplot[
	solid,
	mark=+,
	color=black,
	]
	coordinates {(1,231959398)(2,84145840)(3,51068606)(4,33114298)(5,22308852)(6,15264013)(7,10770971)(8,7889709)(9,5555978)(10,3935882)(11,2952744)(12,2294698)(13,1594915)(14,1222068)(15,1034447)
	};
	\addlegendentry{\scriptsize $\varepsilon=2^{-11}$}
	\addplot[
	solid,
	mark=x,
	color=black,
	]
	coordinates {(1,3745822356)(2,1358838531)(3,824687098)(4,534749945)(5,360256990)(6,246492614)(7,173936224)(8,127407843)(9,89721331)(10,63559020)(11,47682716)(12,37056181)(13,25755656)(14,19734693)(15,16704871)(16,11533447)
	};
	\addlegendentry{\scriptsize $\varepsilon=2^{-13}$}
	\end{semilogyaxis}
	\end{tikzpicture}
	\begin{tikzpicture} 
	\pgfplotsset{
		scale only axis,
		xmin=5, xmax=13
	}
	\begin{axis}[
	width=7cm,
	height=2.75cm,
	axis on top=true,
	axis x line=bottom,	
	axis y line*=left,
	xlabel={$\log(1/\varepsilon)$},
	xmin=4.8, xmax=9.1,
	ymin=15, ymax=24,
	title={Logarithm of the computational cost},
	legend style={at={(0.02,1)},anchor=north west},
	x label style={at={(axis description cs:.88,0.40)},anchor=north},
	]
	\addplot[
	solid,
	mark=+,
	mark options={scale=.8},
	color=black,
	]
	coordinates {(4.852,15.12)(5.545,16.64)(6.238,18.07)(6.931,19.49)(7.625,20.93)(8.318,22.34)(9.011,23.72)
	}; 
	\addlegendentry{\scriptsize Observed (MLMC)}
	\addplot[
	dashed,
	mark=+,
	mark options={scale=.8},
	color=black,
	]
	coordinates {(4.852,15.33)(5.545,16.72)(6.238,18.1)(6.931,19.49)(7.625,20.88)(8.318,22.26)(9.011,23.65)
	}; 
	\addlegendentry{\scriptsize Predicted (MLMC)}
	\end{axis}
	\end{tikzpicture}
	\caption{\small 
	The top pictures show the bias and level variance decay as a 
	function of $k$ for the payoff 
	$g(\chi)=\max\{S_T-K,0\}\1{\{\ov S_T \le M\}}$ 
	with $S_0=100$, $T=90/365$, $K=95$, $M=102$ and the parameter values 
	for the USD/JPY FX rate (v2) in Table~\ref{tab:fit}. 
	The theoretical predictions (\emph{dashed}) are based on 
	Proposition~\ref{prop:error-rates} for barrier-type 1 payoffs. 
	The bottom pictures show the corresponding value of 
	the complexities and parameters $n,N_1,\ldots,N_n$  
	associated to the precision levels 
	$\varepsilon\in\{2^{-9},2^{-10},\ldots,2^{-19}\}$. 
	}\label{fig:MLMC_barrier}
\end{figure}
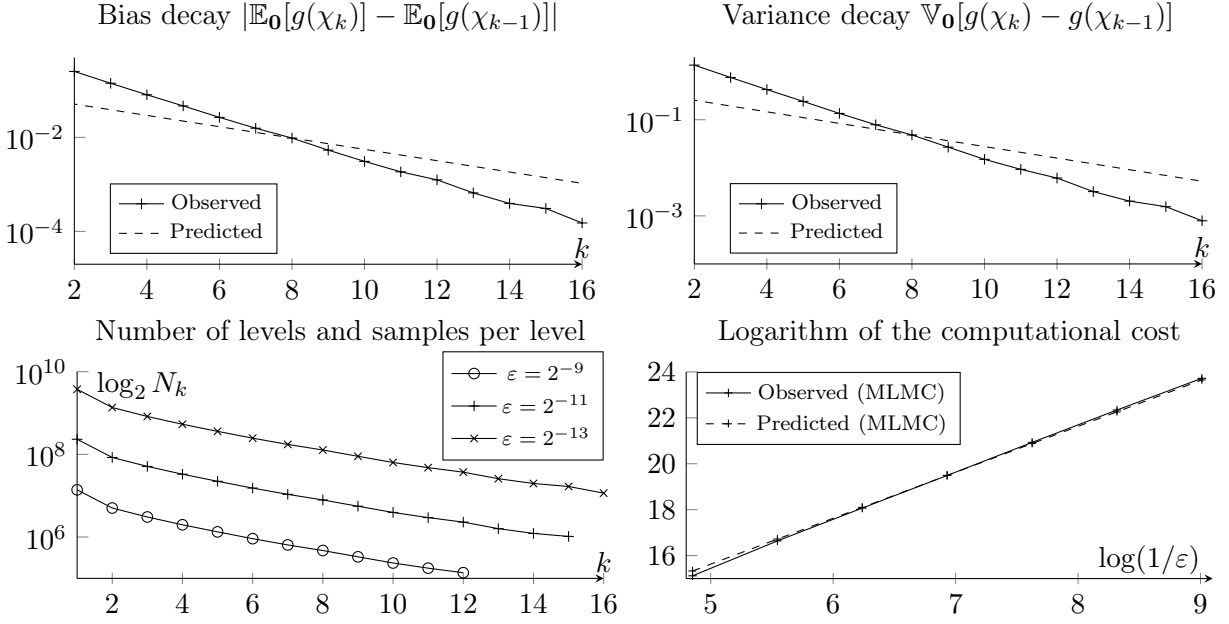

\begin{figure}[ht]
	\centering
	\begin{tikzpicture}
	\pgfplotsset{
		scale only axis,
		xmin=2, xmax=5
	} 
	\begin{semilogyaxis}[
	width=6.75cm,
	height=2.75cm,
	axis on top=true,
	axis x line=bottom,
	axis y line*=left,
	xlabel = {$k$},
	xmin = 2, xmax = 16,
	ymin = 1e-7, ymax = .0005,
	title={Bias decay $|\e_\bzero[g(\chi_k)]-\e_\bzero[g(\chi_{k-1})]|$},
	legend style={at={(0.4,.05)},anchor=south east},
	x label style={at={(axis description cs:1,0.37)},anchor=north},
	]
	\addplot[
	solid,
	mark=+,
	color=black,
	]
	coordinates {(1,0.00222)(2,0.0002603)(3,0.0001829)(4,0.0001279)(5,8.517e-5)(6,5.368e-5)(7,3.23e-5)(8,2.181e-5)(9,1.372e-5)(10,7.943e-6)(11,5.022e-6)(12,3.193e-6)(13,1.722e-6)(14,1.102e-6)(15,7.66e-7)(16,4.271e-7)
	};
	\addlegendentry{\scriptsize Observed}
	\addplot[
	dashed,
	color=black,
	]
	coordinates {(1,0.0002864)(2,0.0001983)(3,0.0001372)(4,9.5e-5)(5,6.576e-5)(6,4.552e-5)(7,3.151e-5)(8,2.181e-5)(9,1.51e-5)(10,1.045e-5)(11,7.234e-6)(12,5.007e-6)(13,3.466e-6)(14,2.399e-6)(15,1.661e-6)(16,1.15e-6)
	};
	\addlegendentry{\scriptsize Predicted}
	\end{semilogyaxis}
	\end{tikzpicture}
	\begin{tikzpicture}
	\pgfplotsset{
		scale only axis,
		xmin=2, xmax=5
	} 
	\begin{semilogyaxis}[
	width=6.75cm,
	height=2.75cm,
	axis on top=true,
	axis x line=bottom,
	axis y line*=left,
	xlabel = {$k$},
	xmin = 2, xmax = 16,
	ymin = 2e-11, ymax = 2e-5,
	title={ Variance decay $\V_\bzero[g(\chi_k)-g(\chi_{k-1})]$},
	legend style={at={(0.4,.05)},anchor=south east},
	x label style={at={(axis description cs:1,0.37)},anchor=north},
	]
	\addplot[
	solid,
	mark=+,
	color=black,
	]
	coordinates {(1,0.0002792)(2,1.814e-5)(3,1.056e-5)(4,7.797e-6)(5,5.642e-6)(6,2.038e-6)(7,5.357e-6)(8,4.058e-7)(9,3.102e-7)(10,8.492e-8)(11,7.299e-8)(12,1.936e-8)(13,1.214e-8)(14,4.886e-9)(15,1.988e-9)(16,3.126e-10)
	};
	\addlegendentry{\scriptsize Observed}
	\addplot[
	dashed,
	color=black,
	]
	coordinates {(1,5.33e-6)(2,3.689e-6)(3,2.554e-6)(4,1.768e-6)(5,1.224e-6)(6,8.47e-7)(7,5.863e-7)(8,4.058e-7)(9,2.809e-7)(10,1.944e-7)(11,1.346e-7)(12,9.317e-8)(13,6.449e-8)(14,4.464e-8)(15,3.09e-8)(16,2.139e-8)
	};
	\addlegendentry{\scriptsize Predicted}
	
	\end{semilogyaxis}
	\end{tikzpicture}
	\begin{tikzpicture}
	\pgfplotsset{
		scale only axis,
		xmin=2, xmax=5
	} 
	\begin{semilogyaxis}[
	width=7cm,
	height=2.75cm,
	axis on top=true,
	axis x line=bottom,
	axis y line*=left,
	xlabel = {$k$},
	ylabel = {$\log_2 N_k$},
	xmin = 1, xmax = 14,
	ymin = 2e4, ymax = 6e8,
	title={ Number of levels and samples per level},
	legend style={at={(1,.57)},anchor=south east},
	x label style={at={(axis description cs:1,0.37)},anchor=north},
	y label style={at={(axis description cs:.3,1.05)},anchor=north,rotate=-90},
	]
	\addplot[
	solid,
	mark=o,
	color=black,
	]
	coordinates {(1,1878436)(2,338620)(3,210954)(4,156965)(5,119426)(6,65530)(7,98350)(8,25322)(9,20872)
	};
	\addlegendentry{\scriptsize $\varepsilon=2^{-15}$}
	\addplot[
	solid,
	mark=+,
	color=black,
	]
	coordinates {(1,31098430)(2,5606013)(3,3492435)(4,2598618)(5,1977141)(6,1084869)(7,1628233)(8,419205)(9,345541)(10,171518)(11,151610)
	};
	\addlegendentry{\scriptsize $\varepsilon=2^{-17}$}
	\addplot[
	solid,
	mark=x,
	color=black,
	]
	coordinates {(1,508054030)(2,91585246)(3,57055797)(4,42453534)(5,32300473)(6,17723468)(7,26600382)(8,6848530)(9,5645076)(10,2802084)(11,2476835)(12,1221459)(13,929155)(14,568032)
	};
	\addlegendentry{\scriptsize $\varepsilon=2^{-19}$}
	\end{semilogyaxis}
	\end{tikzpicture}
	\begin{tikzpicture} 
	\pgfplotsset{
		scale only axis,
		xmin=5, xmax=13
	}
	\begin{axis}[
	width=7cm,
	height=2.75cm,
	axis on top=true,
	axis x line=bottom,	
	axis y line*=left,
	xlabel={$\log(1/\varepsilon)$},
	xmin=6.2, xmax=13.2,
	ymin=5, ymax=22,
	title={Logarithm of the computational cost},
	legend style={at={(0.02,1)},anchor=north west},
	x label style={at={(axis description cs:.88,0.40)},anchor=north},
	]
	\addplot[
	solid,
	mark=+,
	mark options={scale=.8},
	color=black,
	]
	coordinates {(6.238,5.602)(6.931,6.988)(7.625,8.374)(8.318,10.56)(9.011,12.24)(9.704,14.06)(10.4,15.59)(11.09,17.01)(11.78,18.43)(12.48,19.85)(13.17,21.24)
	}; 
	\addlegendentry{\scriptsize Observed (MLMC)}
	\addplot[
	dashed,
	mark=+,
	mark options={scale=.8},
	color=black,
	]
	coordinates {(6.238,7.305)(6.931,8.691)(7.625,10.08)(8.318,11.46)(9.011,12.85)(9.704,14.24)(10.4,15.62)(11.09,17.01)(11.78,18.4)(12.48,19.78)(13.17,21.17)
	}; 
	\addlegendentry{\scriptsize Predicted (MLMC)}
	\end{axis}
	\end{tikzpicture}
	\caption{\small 
			The top pictures show the bias and level variance decay as a 
		function of $k$ for the payoff 
		$g(\chi)=(S_T/\ov{S}_T-1)^2\1_{\{\tau_T<T/2\}}$ 
		with $S_0=100$, $T=28/365$ and the parameter values 
		for MCD in Table~\ref{tab:fit}. The theoretical predictions 
		(\emph{dashed}) are based on 
		Proposition~\ref{prop:error-rates} for barrier-type 2 payoffs. 
		The bottom pictures show the corresponding value of 
		the complexities and parameters $n,N_1,\ldots,N_n$ 
		associated to the precision levels 
		$\varepsilon\in\{2^{-9},2^{-10},\ldots,2^{-13}\}$. 
	}\label{fig:MLMC_MUI}
\end{figure}
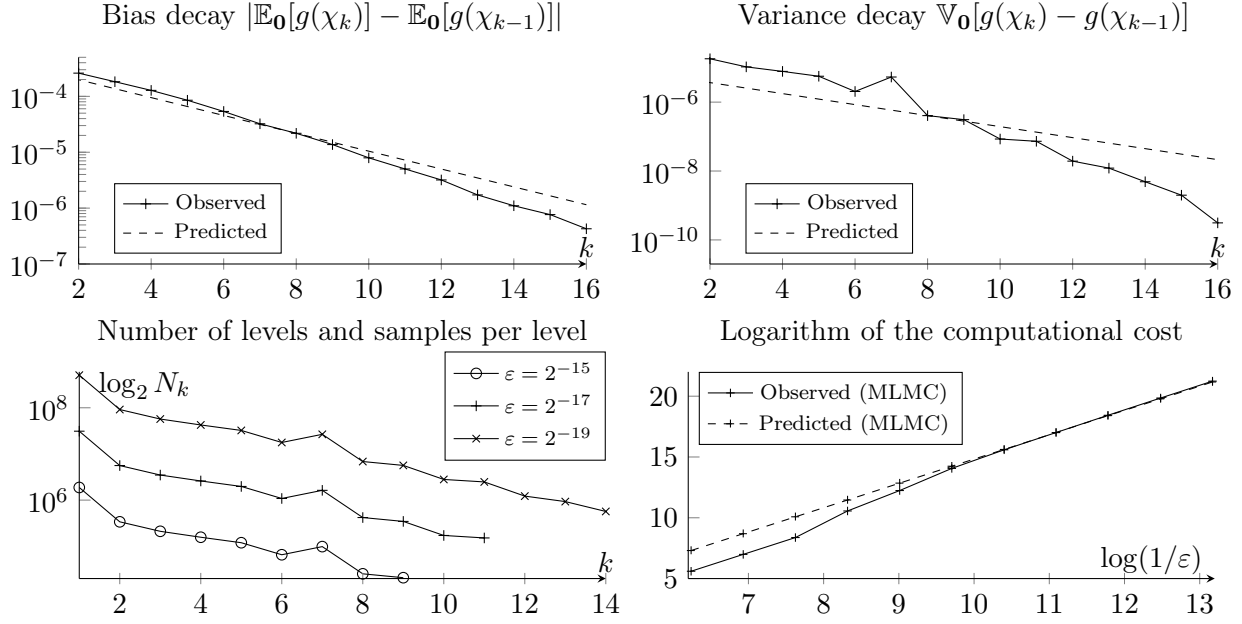

In Figure~\ref{fig:CLT} we plot the estimator $\hat\theta^{g,n}_\MC$ 
in Theorem~\ref{thm:CLT} for (I), parametrised by $\varepsilon\to0$. 
To further illustrate the CLT in Theorem~\ref{thm:CLT}, the figure 
also shows the CIs of confidence level $95\%$ constructed using 
the CLT. 

\begin{figure}[ht]
	\centering
	\begin{tikzpicture} 
		\begin{semilogxaxis} 
			[
			title={Maturity: 90 days},
			ymin=0,
			ymax=4.2,
			xmin=0,
			xmax=1,
			xlabel={\small $\varepsilon$},
			width=9cm,
			height=4.5cm,
			axis on top=true,
			axis x line=bottom, 
			axis y line=left,
			axis line style={->},
			x label style={at={(axis description cs:1.05,0.3)},anchor=north},
			legend style={at={(.5,.55)},anchor=south east}
			]
			
			\addplot[
			solid, mark=+, mark options={scale=.75, solid, thin},
			color=black,
			]
			coordinates {
(1.0,0.93045)(0.7788,2.6315)(0.60653,2.0084)(0.47237,1.4967)(0.36788,1.9413)(0.2865,1.0006)(0.22313,1.306)(0.17377,0.96196)(0.13534,1.2457)(0.1054,1.2746)(0.082085,1.2127)(0.063928,1.1677)(0.049787,1.1527)(0.038774,1.1866)(0.030197,1.1419)(0.023518,1.1741)(0.018316,1.0984)(0.014264,1.1188)(0.011109,1.1295)(0.0086517,1.1266)(0.0067379,1.1287)(0.0052475,1.1136)(0.0040868,1.1107)(0.0031828,1.1153)(0.0024788,1.1172)(0.0019305,1.1117)(0.0015034,1.1131)(0.0011709,1.1133)(0.00091188,1.1129)
			};
			
			\addplot[
			dotted, mark=+, mark options={scale=.75, solid, thin},
			color=black,
			]
			coordinates {
(1.0,2.182)(0.7788,3.8336)(0.60653,2.9565)(0.47237,2.5658)(0.36788,2.6088)(0.2865,1.5309)(0.22313,1.7199)(0.17377,1.3931)(0.13534,1.4819)(0.1054,1.4576)(0.082085,1.3826)(0.063928,1.2916)(0.049787,1.2489)(0.038774,1.26)(0.030197,1.2034)(0.023518,1.2189)(0.018316,1.136)(0.014264,1.147)(0.011109,1.1519)(0.0086517,1.1433)(0.0067379,1.1417)(0.0052475,1.1238)(0.0040868,1.1187)(0.0031828,1.1215)(0.0024788,1.122)(0.0019305,1.1156)(0.0015034,1.1161)(0.0011709,1.1156)(0.00091188,1.1147)
			};
			
			\addplot[
			dotted, mark=+, mark options={scale=.75, solid, thin},
			color=black,
			]
			coordinates {
(1.0,0)(0.7788,1.4295)(0.60653,1.0602)(0.47237,0.42754)(0.36788,1.2738)(0.2865,0.47035)(0.22313,0.89211)(0.17377,0.5308)(0.13534,1.0095)(0.1054,1.0917)(0.082085,1.0428)(0.063928,1.0438)(0.049787,1.0566)(0.038774,1.1132)(0.030197,1.0804)(0.023518,1.1294)(0.018316,1.0609)(0.014264,1.0907)(0.011109,1.1072)(0.0086517,1.1098)(0.0067379,1.1157)(0.0052475,1.1034)(0.0040868,1.1027)(0.0031828,1.1091)(0.0024788,1.1124)(0.0019305,1.1079)(0.0015034,1.1102)(0.0011709,1.111)(0.00091188,1.1111)
			};
		\end{semilogxaxis}
	\end{tikzpicture}
	\caption{\small 
		The \emph{solid} line indicates the estimated value of 
		the expectation $\e_\blambda g(\chi_n)$ for the payoff 
		$g(\chi)=\max\{S_T-K,0\}\1{\{\ov S_T \le M\}}$ 
		with $S_0=100$, $T=90/365$, $K=95$, $M=102$ and the 
		parameter values for the USD/JPY FX rate (v2) in 
		Table~\ref{tab:fit}. We use the MLMC estimator 
		$\hat\theta^{g,n}_\ML$ in~\eqref{eq:MLMC} and the 
		confidence intervals (\emph{dotted} lines) are constructed 
		using Theorem~\ref{thm:CLT} with confidence level $95\%$.
	}\label{fig:CLT}
\end{figure}
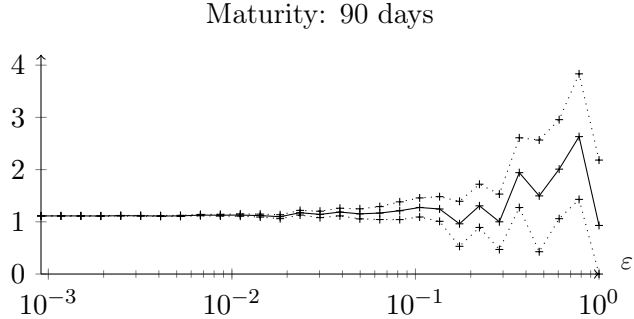

\subsection{Comparing \nameref{alg:TSB} with existing algorithms for tempered stable processes}

In this subsection we take the analysis from 
Subsection~\ref{subsec:SB_SBG_comparison} 
and apply it to the tempered stable case. 

\subsubsection{Comparison with SB-Alg}
\label{subsec:SB-Alg_TS}

Recall from Subsection~\ref{subsec:SBA_Masuda} that SB-Alg is only applicable when $\alpha_\pm<1$ and, under Regime (II), SB-Alg is preferable over \nameref{alg:TSB} for all sufficiently large $T$ if and only if $\max\{\gamma_\blambda^\pp,\gamma_\blambda^\pn\}\le\mu_{2\blambda}-2\mu_\blambda$, where $\gamma_\blambda^\ppm = -c_\pm \lambda_\pm^{\alpha_\pm} \Gamma(-\alpha_\pm)\ge 0$ is defined in~\eqref{eq:gamma_pm}. By the formulae in Subsection~\ref{subsec:prelim-ts}, it is easily seen that $\max\{\gamma_\blambda^\pp,\gamma_\blambda^\pn\}\le\mu_{2\blambda}-2\mu_\blambda$ is equivalent to 
\[
\min\{c_+\lambda_+^{\alpha_+}\Gamma(-\alpha_+),
	c_-\lambda_-^{\alpha_-}\Gamma(-\alpha_-)\}
\ge c_+\lambda_+^{\alpha_+}(2-2^{\alpha_+})\Gamma(-\alpha_+)
	+c_-\lambda_-^{\alpha_-}(2-2^{\alpha_-})\Gamma(-\alpha_-).
\]
Assuming that $\alpha_\pm=\alpha$, the inequality simplifies to 
$\alpha\le\phi(\varrho)$, where we define 
$\phi(x):=\log_2\big(1+\frac{x}{1+x}\big)$ and 
$\varrho := \min\{c_+\lambda_+^{\alpha},c_-\lambda_-^{\alpha}\}
	/\max\{c_+\lambda_+^{\alpha},c_-\lambda_-^{\alpha}\}$. 
In particular, a symmetric L\'evy measure yields $\varrho = 1$ 
and $\phi(1)=\log_2(3/2)=0.58496\ldots$, and a one-sided 
L\'evy measure gives $\varrho=0$ and $\phi(0)=0$. 

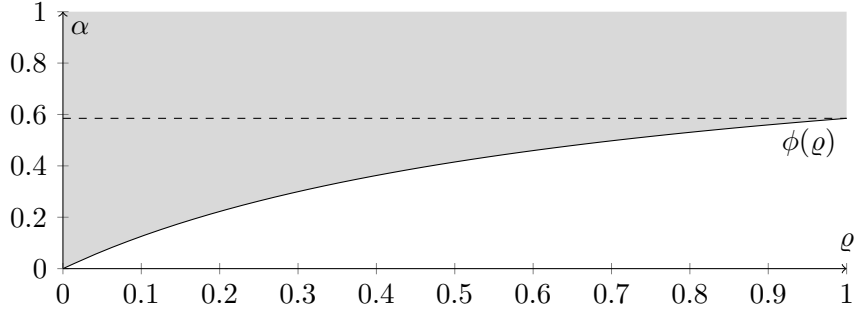
\begin{figure}[ht]
	\begin{center}
		\begin{tikzpicture} 
		\begin{axis} 
		[
		ymin=0,
		ymax=1,
		xmin=0,
		xmax=1,
		xlabel={$\varrho$},
		ylabel={$\alpha$},
		width=12cm,
		height=5cm,
		axis on top=true,
		axis x line=bottom, 
		axis y line=left,
		axis line style={->},
		x label style={at={(axis description cs:1,0.33)},anchor=north},
		ylabel style={at={(axis description cs:0.14,1)},anchor=north,rotate=-90},
		legend style={at={(.18,.3)},anchor=south east}
		]
		
		\addplot[
		dashed,
		color = black,
		]
		coordinates {
			(0,.58496)(1,.58496)
		};
	
		\path[name path=axis] (axis cs:0,0) -- (axis cs:0,1) -- (axis cs:1,1);

		\addplot [
		name path=f, 
		solid, smooth, 
		domain=0:1, 
		] {ln(1+x/(1+x))/ln(2)}
		node[pos=1] (endofplotsphi) {};
		
		\node [below left] at (endofplotsphi) {$\phi(\varrho)$};
		
		\addplot [
		thick,
		color=gray!30,
		fill=gray!30
		]
		fill between[
		of=f and axis,
		];
		
		\end{axis}
		\end{tikzpicture}
		\caption{\small 
	The picture shows the map $\varrho\mapsto\phi(\varrho)$, $\varrho\in[0,1]$. 
	Assuming $\alpha_\pm=\alpha$ and defining 
	$\varrho := \min\{c_+\lambda_+^{\alpha},c_-\lambda_-^{\alpha}\}
	/\max\{c_+\lambda_+^{\alpha},c_-\lambda_-^{\alpha}\}$,  
	\nameref{alg:TSB} is preferable to SB-Alg when $(\varrho,\alpha)$ 
	lies in the shaded region. 
		}\label{fig:SB-TSB}
	\end{center}
\end{figure}

\subsubsection{Comparison with SBG-Alg}
\label{subsec:SBG-Alg_TS}

Recall from Subsection~\ref{subsec:SBG_comparison} that \nameref{alg:TSB} is preferable to SBG-Alg when $\max\{\alpha_+,\alpha_-\}\ge 1$ (as it is equivalent to $\beta_\ast\ge1$). On the other hand, if $\alpha_\pm<1$ (equivalently, $\beta_\ast<1$), \nameref{alg:TSB} outperforms SB-Alg if and only if $(1+T)e^{(\mu_{2\blambda}-2\mu_\blambda)T}\le C_1+C_2T$. For large enough $T$, the SB-Alg will outperform \nameref{alg:TSB}; however, it is generally hard to determine when this happens. We illustrate the region of parameters $(T,\alpha)$ where \nameref{alg:TSB} is preferable in Figure~\ref{fig:SBG-TSB}, assuming $\alpha_\pm=\alpha\in(0,1)$ and all other parameters are as in the USD/JPY (v1) currency pair in Table~\ref{tab:fit}. 

\begin{figure}[ht]
\begin{center}
	\begin{tikzpicture} 
	\begin{axis} 
	[
	ymin=0,
	ymax=1,
	xmin=0,
	xmax=1,
	xlabel={$T$},
	ylabel={$\alpha$},
	width=12cm,
	height=4.5cm,
	axis on top=true,
	axis x line=bottom, 
	axis y line=left,
	axis line style={->},
	x label style={at={(axis description cs:1,0.35)},anchor=north},
	ylabel style={at={(axis description cs:0.14,1)},anchor=north,rotate=-90},
	legend style={at={(.18,.3)},anchor=south east}
	]
	
	\addplot[
	name path=f, 
	solid,
	mark options={scale=.8},
	color=black,
	]
	coordinates {(0,0)(0.004782,0.01)(0.01621,0.02)(0.03261,0.03)(0.05311,0.04)(0.07708,0.05)(0.1041,0.06)(0.1337,0.07)(0.1656,0.08)(0.1996,0.09)(0.2354,0.1)(0.2728,0.11)(0.3116,0.12)(0.3517,0.13)(0.393,0.14)(0.4352,0.15)(0.4784,0.16)(0.5223,0.17)(0.5669,0.18)(0.6121,0.19)(0.6579,0.2)(0.7041,0.21)(0.7507,0.22)(0.7975,0.23)(0.8447,0.24)(0.892,0.25)(0.9395,0.26)(0.9871,0.27)(1.035,0.28)(1.082,0.29)(1.13,0.3)(1.178,0.31)(1.225,0.32)(1.273,0.33)(1.32,0.34)(1.367,0.35)(1.414,0.36)(1.461,0.37)(1.507,0.38)(1.553,0.39)(1.599,0.4)(1.645,0.41)(1.691,0.42)(1.736,0.43)(1.781,0.44)(1.825,0.45)(1.869,0.46)(1.913,0.47)(1.957,0.48)(2.0,0.49)(2.043,0.5)(2.086,0.51)(2.128,0.52)(2.17,0.53)(2.212,0.54)(2.254,0.55)(2.295,0.56)(2.336,0.57)(2.377,0.58)(2.418,0.59)(2.459,0.6)(2.499,0.61)(2.54,0.62)(2.58,0.63)(2.621,0.64)(2.661,0.65)(2.702,0.66)(2.742,0.67)(2.783,0.68)(2.825,0.69)(2.866,0.7)(2.908,0.71)(2.951,0.72)(2.994,0.73)(3.038,0.74)(3.083,0.75)(3.129,0.76)(3.177,0.77)(3.225,0.78)(3.276,0.79)(3.328,0.8)(3.382,0.81)(3.439,0.82)(3.499,0.83)(3.563,0.84)(3.63,0.85)(3.702,0.86)(3.78,0.87)(3.864,0.88)(3.955,0.89)(4.057,0.9)(4.17,0.91)(4.299,0.92)(4.446,0.93)(4.62,0.94)(4.829,0.95)(5.092,0.96)(5.439,0.97)(5.944,0.98)(6.841,0.99)(7.638,0.9945)(9.973,0.999)(10.8,0.99945)(13.17,0.9999)(14.01,0.999945)(16.38,0.99999)(17.21,0.9999945)(19.58,0.999999)
	}; 
	
	\addplot[name path=axis] coordinates
	{(0.0,0)(0,2.01)(1.01,2.01)(1.01,0.864)};
	
	\addplot [
	thick,
	color=gray!30,
	fill=gray!30
	]
	fill between[
	of=f and axis,
	];
	
	\addplot [
	thick,
	color=gray!30,
	fill=gray!30
	]
	fill between[
	of=f and axis,
	];
	
	\end{axis}
	\end{tikzpicture}
	\caption{\small 
		The shaded region is the set of points $(T,\alpha)$ where 
		\nameref{alg:TSB} is preferable to SBG-Alg assuming $\alpha_\pm=\alpha\in(0,1)$ 
		and all other parameters are as in the USD/JPY (v1) 
		currency pair in Table~\ref{tab:fit}. 
	}\label{fig:SBG-TSB}
\end{center}
\end{figure}
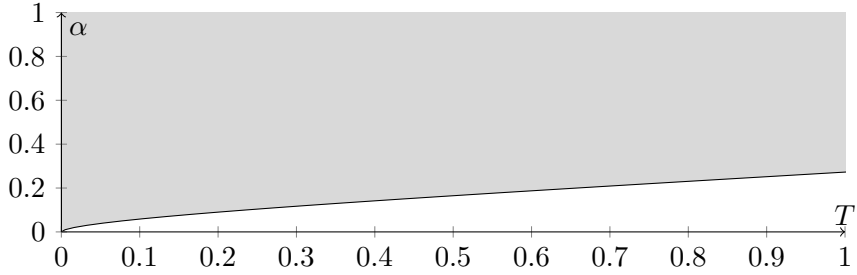

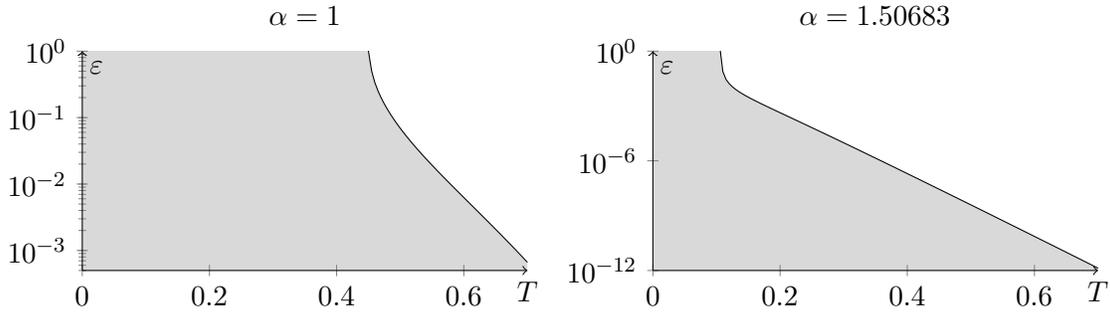
\begin{figure}[ht]
\begin{center}
	\begin{tikzpicture} 
	\begin{semilogyaxis} 
	[
	title={$\alpha = 1$},
	ymin=0.0005,
	ymax=1,
	xmin=0,
	xmax=.7,
	xlabel={$T$},
	ylabel={$\varepsilon$},
	width=7.5cm,
	height=4.5cm,
	axis on top=true,
	axis x line=bottom, 
	axis y line=left,
	axis line style={->},
	x label style={at={(axis description cs:1,0.175)},anchor=north},
	ylabel style={at={(axis description cs:0.24,1)},anchor=north,rotate=-90},
	legend style={at={(.18,.2)},anchor=south east}
	]
	
	\addplot[
	name path=f, 
	solid,
	mark options={scale=.8},
	color=black,
	]
	coordinates {(0,2)(.44,2)(0.45,1.0)(0.455,0.5043)(0.46,0.3379)(0.465,0.2526)(0.47,0.1984)(0.475,0.1605)(0.48,0.1324)(0.485,0.1109)(0.49,0.09392)(0.495,0.08026)(0.5,0.06908)(0.505,0.05983)(0.51,0.05207)(0.515,0.04552)(0.52,0.03993)(0.525,0.03515)(0.53,0.03102)(0.535,0.02744)(0.54,0.02433)(0.545,0.0216)(0.55,0.01922)(0.555,0.01712)(0.56,0.01526)(0.565,0.01363)(0.57,0.01218)(0.575,0.01089)(0.58,0.00974)(0.585,0.008719)(0.59,0.007808)(0.595,0.006995)(0.6,0.006267)(0.605,0.005617)(0.61,0.005034)(0.615,0.004511)(0.62,0.004043)(0.625,0.003623)(0.63,0.003245)(0.635,0.002907)(0.64,0.002603)(0.645,0.00233)(0.65,0.002084)(0.655,0.001864)(0.66,0.001666)(0.665,0.001488)(0.67,0.001329)(0.675,0.001186)(0.68,0.001057)(0.685,0.0009422)(0.69,0.000839)(0.695,0.0007464)(0.7,0.0006636)
	}; 
	
	\addplot[name path=axis, draw=none] coordinates
	{(0,.00001)(0.7,0.0001)};
	
	\addplot [
	thick,
	color=gray!30,
	fill=gray!30
	]
	fill between[
	of=f and axis,
	];
	
	\end{semilogyaxis}
	\end{tikzpicture}
	\begin{tikzpicture} 
	\begin{semilogyaxis} 
	[
	title={$\alpha = 1.50683$},
	ymin=1e-12,
	ymax=1,
	xmin=0,
	xmax=.7,
	xlabel={$T$},
	ylabel={$\varepsilon$},
	width=7.5cm,
	height=4.5cm,
	axis on top=true,
	axis x line=bottom, 
	axis y line=left,
	axis line style={->},
	x label style={at={(axis description cs:1,0.175)},anchor=north},
	ylabel style={at={(axis description cs:0.24,1)},anchor=north,rotate=-90},
	legend style={at={(.18,.2)},anchor=south east}
	]
	
	\addplot[
	name path=f, 
	solid,
	mark options={scale=.8},
	color=black,
	]
	coordinates {(0,2)(.105,2)(0.11,0.07675)(0.115,0.02843)(0.12,0.01657)(0.125,0.01119)(0.13,0.008125)(0.135,0.006151)(0.14,0.004783)(0.145,0.003787)(0.15,0.003036)(0.155,0.002457)(0.16,0.002001)(0.165,0.001638)(0.17,0.001345)(0.175,0.001108)(0.18,0.0009146)(0.185,0.000756)(0.19,0.0006255)(0.195,0.0005179)(0.2,0.000429)(0.205,0.0003555)(0.21,0.0002946)(0.215,0.0002442)(0.22,0.0002023)(0.225,0.0001677)(0.23,0.0001389)(0.235,0.0001151)(0.24,9.528e-5)(0.245,7.888e-5)(0.25,6.529e-5)(0.255,5.403e-5)(0.26,4.469e-5)(0.265,3.696e-5)(0.27,3.056e-5)(0.275,2.526e-5)(0.28,2.087e-5)(0.285,1.724e-5)(0.29,1.424e-5)(0.295,1.176e-5)(0.3,9.704e-6)(0.305,8.008e-6)(0.31,6.607e-6)(0.315,5.45e-6)(0.32,4.494e-6)(0.325,3.705e-6)(0.33,3.054e-6)(0.335,2.517e-6)(0.34,2.074e-6)(0.345,1.708e-6)(0.35,1.407e-6)(0.355,1.159e-6)(0.36,9.541e-7)(0.365,7.854e-7)(0.37,6.464e-7)(0.375,5.319e-7)(0.38,4.377e-7)(0.385,3.6e-7)(0.39,2.961e-7)(0.395,2.436e-7)(0.4,2.003e-7)(0.405,1.647e-7)(0.41,1.354e-7)(0.415,1.113e-7)(0.42,9.143e-8)(0.425,7.513e-8)(0.43,6.173e-8)(0.435,5.071e-8)(0.44,4.166e-8)(0.445,3.421e-8)(0.45,2.81e-8)(0.455,2.307e-8)(0.46,1.894e-8)(0.465,1.555e-8)(0.47,1.277e-8)(0.475,1.048e-8)(0.48,8.601e-9)(0.485,7.059e-9)(0.49,5.792e-9)(0.495,4.752e-9)(0.5,3.899e-9)(0.505,3.199e-9)(0.51,2.624e-9)(0.515,2.152e-9)(0.52,1.765e-9)(0.525,1.447e-9)(0.53,1.187e-9)(0.535,9.732e-10)(0.54,7.979e-10)(0.545,6.542e-10)(0.55,5.363e-10)(0.555,4.396e-10)(0.56,3.603e-10)(0.565,2.953e-10)(0.57,2.42e-10)(0.575,1.983e-10)(0.58,1.625e-10)(0.585,1.332e-10)(0.59,1.091e-10)(0.595,8.939e-11)(0.6,7.323e-11)(0.605,5.999e-11)(0.61,4.914e-11)(0.615,4.025e-11)(0.62,3.297e-11)(0.625,2.7e-11)(0.63,2.211e-11)(0.635,1.811e-11)(0.64,1.483e-11)(0.645,1.214e-11)(0.65,9.942e-12)(0.655,8.14e-12)(0.66,6.664e-12)(0.665,5.456e-12)(0.67,4.466e-12)(0.675,3.656e-12)(0.68,2.993e-12)(0.685,2.45e-12)(0.69,2.005e-12)(0.695,1.641e-12)(0.7,1.343e-12)
	}; 
	
	\addplot[name path=axis, draw=none] coordinates
	{(0,1e-12)(0.7,1e-12)};
	
	\addplot [
	thick,
	color=gray!30,
	fill=gray!30
	]
	fill between[
	of=f and axis,
	];
	
	\end{semilogyaxis}
	\end{tikzpicture}
	\caption{\small 
		The shaded region is the set of points $(T,\varepsilon)$ where \nameref{alg:TSB} is preferable to SBG-Alg where $\alpha_\pm=\alpha\in[1,2)$ and all other parameters are as in the MCD stock in Table~\ref{tab:fit}. 
	}\label{fig:SBG-TSB2}
\end{center}
\end{figure}

\subsection{Concluding remarks}
\label{subsec:concluding_rem}

\nameref{alg:TSB} is an easily implementable algorithm for which optimal MLMC (and even unbiased) estimators exist. \nameref{alg:TSB} combines the best of both worlds: it is applicable to \emph{all} tempered stable processes (as is SBG-Alg in~\cite{SBG}), while preserving the \emph{geometric} convergence of SB-Alg in~\cite{LevySupSim}. The only downside of \nameref{alg:TSB} is the enlarged variance by the factor $\exp((\mu_{2\blambda}-2\mu_\blambda)T)$. This factor is typically small (see discussion in Subsection~\ref{subsec:MC} above) and, when it is large, easily implementable variance reduction techniques exist (see details in Subsection~\ref{subsec:var_red} above). These facts favour the use of the MLMC estimator based on \nameref{alg:TSB} over its competitors when $(\mu_{2\blambda}-2\mu_\blambda)T$ is not large (for more concrete rules of thumb, see the concluding paragraphs of Subsections~\ref{subsec:MCvMLMC}, \ref{subsec:SBA_Masuda} and \ref{subsec:SBG_comparison} above).

In practice, when implementing the MC and MLMC estimators of \nameref{alg:TSB}, the leading constant of the bias and variance decay are hard to compute and often overestimate the error. The general practice in this situation (see, e.g.~\cite[Sec.~2]{MR3349310}) is to numerically estimate these constants along with the integers $n(\varepsilon)$, $N$ and $(N_k)_{k\in\N}$ in Proposition~\ref{prop:MC_MLMC}. Such an estimation requires few simulations for the first few levels and some extrapolation but typically performs well in practice. This is particularly true in our setting as the MLMC estimator is optimal, see~\cite[Sec.~3]{MR3349310} for a detailed discussion and a generic MATLAB implementation. 

\section{Proofs}
\label{sec:Proofs}

Let us introduce the geometric rate $\vartheta_p$ used in Proposition~\ref{prop:error-rates} above. Let $\beta$ be the \textit{Blumenthal-Getoor index}~\citep{MR0123362}, defined as 
\begin{equation}\label{def:I0_beta}
\beta:=\inf\{p>0:I_0^p<\infty\},\quad\text{where}\quad 
I_0^p:=\int_{(-1,1)}|x|^p\nu(\D x),\quad\text{for any }p\geq 0,
\end{equation}
and note that $\beta\in[0,2]$ since $I_0^2<\infty$. Moreover, $I_0^1<\infty$ if and only if the jumps of $X$ have finite variation. In the case of (tempered) stable processes, $\beta$ is the greatest of the two activity indices of the L\'evy measure. Note that $I_0^p<\infty$ for any $p>\beta$ but $I_0^\beta$ can be either finite or infinite. If $I_0^\beta=\infty$ we must have $\beta<2$ and can thus pick $\delta\in(0,2-\beta)$, satisfying $\beta+\delta<1$ whenever $\beta<1$, and define
\begin{equation}\label{eq:BG+}
\beta_* := \beta+\delta\cdot \1_{\{I_0^\beta=\infty\}}\in[\beta,2].
\end{equation}
The index $\beta_*$ is either equal to $\beta$ or arbitrarily close to 
it. In either case we have $I_0^{\beta_*}<\infty$. 
Define $\alpha\in[\beta,2]$ and $\alpha_*\in[\beta_*,2]$ by
\begin{equation}\label{eq:alpha}
\alpha := 2\cdot \1_{\sigma\neq0} +\1_{\sigma=0}\begin{cases}
1, & I_0^1<\infty\text{ and }b_0\neq0\\
\beta, & \text{otherwise},
\end{cases}
\quad\text{and}\qquad
\alpha_* := \alpha+(\beta_*-\beta)\cdot \1_{\alpha=\beta}.
\end{equation}
Finally, we may define 
\begin{equation}\label{eq:eta}
\vartheta_p := \log\Big(1+\1_{p>\alpha}+\frac{p}{\alpha_*} \cdot \1_{p\leq\alpha}\Big)\in(0,\log2],
\quad\text{for any}\quad p>0,
\end{equation}
and note that $\vartheta_p\geq \log(3/2)$ for $p\geq 1$. 

In order to prove Proposition~\ref{prop:error-rates} for barrier-type 1 payoffs, we need to ensure that $\ov{X}_T$ has a sufficiently regular distribution function under $\p_{p\blambda}$. The following assumption will help us establish that in certain cases of interest. 

\begin{assumption*}[S]
\label{asm:s}
Under $\p_\bzero$, the L\'evy process $X=(X_t)_{t\in[0,T]}$ is in the domain of attraction of an $\alpha$-stable process as $t\to0$ with $\alpha\in(1,2]$. Put differently, there exists a positive function $g$ such that the law of $X_t/g(t)$, under $\p_\bzero$, converges in distribution to an $\alpha$-stable law for some $\alpha\in(1,2]$ as $t\to0$. 
\end{assumption*}

When $X$ is tempered stable, Assumption~(\nameref{asm:s}) holds 
trivially if $\max\{\alpha_+,\alpha_-\}>1$ or $\sigma\ne 0$. 
The index $\alpha$ in Assumption~(\nameref{asm:s}) 
necessarily agrees with the one in~\eqref{eq:alpha}, see~\cite[Subsec.~2.1]{ZoomIn}. For further sufficient and necessary conditions for Assumption~(\nameref{asm:s}), we refer the reader to~\cite{MR3784492,ZoomIn}. In particular, Assumption~(\nameref{asm:s}) remains satisfied if the L\'evy measure $\nu$ is modified away from $0$ or the law of $X$ is changed under an equivalent change of measure, see~\cite[Subsec.~2.3.4]{ZoomIn}. 

\begin{lem}
\label{lem:temper_CDF}
For any Borel set $A\subset\R\times\R_+\times[0,T]$ and $p>1$ we have 
\begin{equation}
\label{eq:CDF_Holder}
\p_{\blambda}(\chi\in A)
\le e^{(\mu_{p\blambda}-p\mu_\blambda)T/p}
	\p_\bzero(\chi\in A)^{1-1/p},
\end{equation}
where the constants $\mu_\blambda$ and $\mu_{p\blambda}$ are defined in~\eqref{eq:RN-drift}. Moreover, if $I_0^1<\infty$, then we also have
\begin{equation}
\label{eq:CDF_Lip}
\p_{\blambda}(\chi\in A)
\le e^{(\gamma_\blambda^\pp+\gamma_\blambda^\pn)T}
	\p_\bzero(\chi\in A),
\end{equation}
where the constants $\gamma^\ppm_\blambda$ are defined in~\eqref{eq:gamma_pm}. 
\end{lem}

The proofs of Lemma~\ref{lem:temper_CDF} 
and Proposition~\ref{prop:error-rates} rely on the identity 
$\Upsilon_\blambda^p=\Upsilon_{p\blambda}e^{(\mu_{p\blambda}-p\mu_\blambda)T}$, 
valid for any $\blambda\in\R_+^2$ and $p\ge 1$. 

\begin{proof}[Proof of Lemma~\ref{lem:temper_CDF}]
Fix the Borel set $A$. By Theorem~\ref{thm:chi_exp_temp} and H\"older's inequality with $p>1$, we get
\[
\p_\blambda(\chi\in A)=\e_\bzero\big[\Upsilon_\blambda\1_{\{\chi\in A\}}\big]
\le\e_\bzero[\Upsilon_\blambda^p]^{1/p}\e_\bzero\big[\1_{\{\chi\in A\}}^q\big]^{1/q}
=e^{(\mu_{p\blambda}-p\mu_\blambda)T/p}\p_\bzero(\chi\in A)^{1/q},
\]
where $1/q=1-1/p$, implying~\eqref{eq:CDF_Holder}. 
If $I_0^1<\infty$, then $\mu_{p\blambda}-p \mu_\blambda 
= p (\gamma^\pp_{\blambda}+\gamma^\pn_{\blambda}) 
- (\gamma^\pp_{p\blambda}+\gamma^\pn_{p\blambda})
\le p (\gamma^\pp_{\blambda}+\gamma^\pn_{\blambda})$. 
Thus, taking $p\to\infty$ (and hence $q\to 1$) 
in~\eqref{eq:CDF_Holder} yields~\eqref{eq:CDF_Lip}.
\end{proof}

\begin{proof}[Proof of Proposition~\ref{prop:error-rates}]
Theorem~\ref{thm:chi_exp_temp} implies that all the expectations 
in the statement of Proposition~\ref{prop:error-rates} can be 
replaced with the expectation 
$\e_{p\blambda}[|g(\chi)-g(\chi_n)|]$. 
Since $\blambda\ne\bzero$, implying that 
$\min\{\e_{p\blambda}[\max\{X_t,0\}],\e_{p\blambda}[\max\{-X_t,0\}]\}<\infty$, 
\cite[Prop.~1]{LevySupSim} yields the result for Lipschitz payoffs. 
By the assumption in Proposition~\ref{prop:error-rates} for the 
	locally Lipschitz case, the L\'evy measure $\nu_{p\blambda}$ in~\eqref{eq:lambda} satisfies 
the assumption in~\cite[Prop.~2]{LevySupSim}, implying the result for 
locally Lipschitz payoffs. The result for barrier-type 2 payoffs 
follows from a direct application of~\cite[Lem.~14 \& 15]{SBG} 
and~\cite[Thm~2]{LevySupSim}. 

The result for barrier-type 1 payoffs follows 
	from~\cite[Prop.~3 \& Rem.~6]{LevySupSim} if we show the existence of a 
constant $K'$ satisfying 
$\p_{p\blambda}(M<\ov{X}_T\le M+x)\le K'x^\gamma$ for all $x>0$. 
If $\gamma\in(0,1)$, such $K'$ exists by~\eqref{eq:CDF_Holder} in 
Lemma~\ref{lem:temper_CDF} above with $p=(1-\gamma)^{-1}>1$ 
and $A=\R\times(M,M+x]\times[0,T]$. 
If $\gamma=1$ and $I_0^1<\infty$, the existence of $K'$ 
follows from the assumption in Proposition~\ref{prop:error-rates} 
and~\eqref{eq:CDF_Lip} in Lemma~\ref{lem:temper_CDF}.
If $\gamma=1$ and Assumption~(\nameref{asm:s}) holds, 
then~\cite[Thm.~5.1]{ZoomIn} implies the existence of $K'$. 
\end{proof}

\begin{lem}
\label{lem:V_k_DCT}
Let the payoff $g$ be as in Proposition~\ref{prop:error-rates} with 
$p=2$ and $\e_\bzero[g(\chi)^2\Upsilon_\blambda^2]<\infty$. 
Let $n=n(\varepsilon)$, $N$ and $N_1,\ldots,N_n$ be as in 
Proposition~\ref{prop:MC_MLMC}, then the following limits hold 
as $\varepsilon\to0$: 
\begin{align}
\label{eq:V_k_DCT1}
&\qquad\qquad
\varepsilon^2N_k
\to 2\sqrt{\V_\bzero[D_{k,1}^g]/k}
	\sum_{j=1}^\infty\sqrt{j\V_\bzero[D_{j,1}^g]}\in(0,\infty),
\qquad k\in\N,\\
\label{eq:V_k_DCT2}
&\frac{\V_\bzero[\hat\theta_\MC^{g,n(\varepsilon)}]}{\varepsilon^{2}/2}
=\frac{\V_\bzero[\theta_1^{g,n(\varepsilon)}]}{\varepsilon^2 N/2}\to 1,
\qquad\text{and}\qquad
\frac{\V_\bzero[\hat\theta_\ML^{g,n(\varepsilon)}]}{\varepsilon^{2}/2}
=\sum_{k=1}^n \frac{\V_\bzero[D_{k,1}^g]}{\varepsilon^2 N_k/2}
\to 1.
\end{align}
\end{lem}

\begin{proof}
Since $x+1\ge\lceil x\rceil\ge x$, we have 
$B_{k}(\varepsilon)\ge \varepsilon^2N_k\ge B_{n,k}(\varepsilon)$, where 
\[
B_k(\varepsilon):=
\varepsilon^2
+2\sqrt{\V_\bzero[D_{k,1}^g]/k}
\sum_{j=1}^\infty\sqrt{j\V_\bzero[D_{j,1}^g]},
\qquad
\text{and}
\qquad
B_{n,k}(\varepsilon):=
2\sqrt{\V_\bzero[D_{k,1}^g]/k}
	\sum_{j=1}^n\sqrt{j\V_\bzero[D_{j,1}^g]},
\] 
implying the limit in~\eqref{eq:V_k_DCT1} (note that since $\V_\bzero[D_{k,1}^g]\le 2\e_\bzero[(\Delta_{k}^g)^2+(\Delta_{k-1}^g)^2] =\Oh(e^{-\vartheta_gk})$ for some $\vartheta_g>0$, the limiting value in~\eqref{eq:V_k_DCT1} is finite). The first limit in~\eqref{eq:V_k_DCT2} follows similarly: $\varepsilon^2/2+\V_\bzero[g(\chi_n)\Upsilon_\blambda]\ge \varepsilon^2N/2\ge \V_\bzero[g(\chi_n)\Upsilon_\blambda]$, where $\V_\bzero[\theta_1^{g,n}] =\V_\bzero[g(\chi_n)\Upsilon_\blambda] \to\V_\bzero[g(\chi)\Upsilon_\blambda]>0$ as $\varepsilon\to 0$ by the convergence in $L^2$ of Proposition~\ref{prop:error-rates}. By the same inequalities, we obtain 
\[
\sum_{k=1}^n\frac{\V_\bzero[D_{k,1}^g]}{B_k(\varepsilon)/2}
\le\sum_{k=1}^n
	\frac{\V_\bzero[D_{k,1}^g]}{\varepsilon^2N_k/2}
\le\sum_{k=1}^n
	\frac{\V_\bzero[D_{k,1}^g]}{B_{n,k}(\varepsilon)/2}
=1. 
\]
The left-hand side converges to $1$ by the monotone convergence 
theorem with respect to the counting measure, implying the second 
limit in~\eqref{eq:V_k_DCT2} and completing the proof. 
\end{proof}

\begin{proof}[Proof of Theorem~\ref{thm:CLT}]
We first  establish the CLT for the MLMC estimator
$\hat\theta_\ML^{g,n(\varepsilon)}$, where $n=n(\varepsilon)$ is as 
stated in the theorem and the numbers of samples 
$N_1,\ldots,N_n$ are given in~\eqref{eq:N_k}. By~\eqref{eq:bias} 
and~\eqref{eq:MLMC} we have 
\begin{equation*}
\sqrt{2}\varepsilon^{-1}\left(\hat\theta_\ML^{g,n(\varepsilon)} 
	-\e_\blambda[g(\chi)]\right)
=\sqrt{2}\varepsilon^{-1} \e_\bzero[\Delta_{n(\varepsilon)}^g]+
\sum_{k=1}^{n(\varepsilon)}\sum_{i=1}^{N_k}\zeta_{k,i}, 
\enskip\text{where } 
\zeta_{k,i}:=\frac{\sqrt{2}}
	{\varepsilon N_k}\big(D_{k,i}^g-\e_\bzero[D_{k,i}^g]\big).
\end{equation*}
By  assumption we have $c_0>1/\vartheta_g$. Thus, the limit $\sqrt{2}\varepsilon^{-1} \e_\bzero[\Delta_{n(\varepsilon)}^g]=\Oh(\varepsilon^{-1+c_0\vartheta_g}) \to 0$ as $\varepsilon\to0$ follows from Proposition~\ref{prop:error-rates}. Hence the CLT in~\eqref{eq:CLT} for the estimator $\hat\theta_\ML^{g,n(\varepsilon)}$ follows if we prove
\[
\sum_{k=1}^{n(\varepsilon)}\sum_{i=1}^{N_k}\zeta_{k,i} \overset{d}{\to}Z \quad\text{as $\varepsilon\to0$,}
\]
where $Z$ is a normal random variable with mean zero and unit variance. Thus, by~\cite[Thm.~5.12]{MR1876169}, it suffices to note that $\zeta_{k,i}$ have zero mean $\e_\bzero[\zeta_{k,i}]=0$, 
\[
\sum_{k=1}^{n(\varepsilon)} 
	\sum_{i=1}^{N_k}\e_\bzero[\zeta_{k,i}^2] 
= \sum_{k=1}^{n(\varepsilon)}\frac{2}{\varepsilon^2N_k} 
	\V_\bzero[D_{k,1}^g]\to 1 
\quad\text{as }\varepsilon\to0,
\] 
which holds by~\eqref{eq:V_k_DCT2}, and establish the Lindeberg condition: for any $r>0$ the following limit holds $\sum_{k=1}^{n(\varepsilon)}\sum_{i=1}^{N_k} \e_\bzero[\zeta_{k,i}^2\1_{\{|\zeta_{k,i}|>r\}}]\to 0$ as $\varepsilon\to0$. 

To prove the Lindeberg condition first note that 
\begin{equation}
\label{eq:DOM_N_k_bound}
\sum_{i=1}^{N_k}\e_\bzero[\zeta_{k,i}^2\1_{\{|\zeta_{k,i}|>r\}}]
	=N_k\e_\bzero[\zeta_{k,1}^2\1_{\{|\zeta_{k,i}|>r\}}]
		\le N_k\e_\bzero[\zeta_{k,1}^2]\quad \text{for any $k\in\N$.}
	\end{equation}
By~\eqref{eq:V_k_DCT1}, the bound $N_k\e_\bzero[\zeta_{k,1}^2] =2\V_\bzero[D_{k,1}^g]/(\varepsilon^2N_k)$ converges for all $k\in\N$ as $\varepsilon\to 0$ to some $c_k\ge 0$ and $\sum_{k=1}^nN_k\e_\bzero[\zeta_{k,1}^2]\to 1 =\sum_{k=1}^\infty c_k$. Lemma~\ref{lem:V_k_DCT} also implies that $\varepsilon N_k\to\infty$ and $\varepsilon^2N_k$ converges to a positive finite constant as $\varepsilon\to0$. Since $\V_\bzero[D_{k,1}^g]<\infty$ for all $k\in\N$, the dominated convergence theorem implies 
\[
N_k\e_\bzero[\zeta_{k,1}^2\1_{\{|\zeta_{k,i}|>r\}}]
=\frac{2\e_\bzero[(D_{k,1}^g-\e_\bzero[D_{k,1}^g])^2
	\1_{\{|D_{k,1}^g-\e_\bzero[D_{k,1}^g]|>r\varepsilon N_k/2\}}]}{\varepsilon^2 N_k}
\to 0,
\qquad\text{as }\varepsilon\to 0.
\]
Thus, the inequality in~\eqref{eq:DOM_N_k_bound} and the dominated convergence 
theorem~\cite[Thm~1.21]{MR1876169} with 
respect to the counting measure yield the Lindeberg condition, 
establishing the CLT for $\hat\theta_\ML^{g,n(\varepsilon)}$.

Let us now establish the CLT for the MC estimator 
$\hat\theta_\MC^{g,n(\varepsilon)}$, with the number of samples $N$ given in 
Proposition~\ref{prop:MC_MLMC}.  As before, 
by Proposition~\ref{prop:error-rates} and the definition of $n(\varepsilon)$ in the theorem,
the bias 
satsfies 
$\sqrt{2}\varepsilon^{-1} \e_\bzero[\Delta_{n(\varepsilon)}^g]=\Oh(\varepsilon^{-1+c_0\vartheta_g}) \to 0$
as $\varepsilon\to0$.
 Thus, 
by~\cite[Thm.~5.12]{MR1876169}, it suffices to show that 
$2\V_\bzero[g(\chi_n)\Upsilon_\blambda]/(\varepsilon^2N)\to1$ 
as $\varepsilon\to0$ and the Lindeberg condition holds: for any $r>0$, 
\[
C(\varepsilon)
:=\sum_{i=1}^N\e_\bzero[|\zeta_{i,n(\varepsilon)}'|^2
	\1_{\{|\zeta_{1,n(\varepsilon)}'|>r\}}]\to 0,
\qquad\text{as}\quad\varepsilon\to 0,
\qquad\text{where}\quad
\zeta_{i,n}':=\frac{\sqrt{2}}{\varepsilon N}
	(\theta_i^{g,n}-\e_\bzero[\theta_i^{g,n}]).
\]
The limit 
$2\V_\bzero[g(\chi_n)\Upsilon_\blambda]/(\varepsilon^2N)\to1$ 
as $\varepsilon\to0$ follows from 
Lemma~\ref{lem:V_k_DCT}. To establish the Lindeberg condition, 
let $\widetilde\theta_\varepsilon
:=\theta_1^{g,n(\varepsilon)}-\e_\bzero[\theta_1^{g,n(\varepsilon)}]$ 
and note that 
\[
C(\varepsilon)
=N\e\big[|\zeta_{1,n(\varepsilon)}'|^2
	\1_{\{|\zeta_{1,n(\varepsilon)}'|>r\}}\big]
=2\e\big[|\widetilde\theta_\varepsilon|^2
	\1_{\{|\widetilde\theta_\varepsilon|>r\varepsilon N/2\}}\big]
	/(\varepsilon^2N).
\] 
Since $\theta_1^{g,n}\overset{L^2}{\to}g(\chi)\Upsilon_\blambda$ 
as $n\to\infty$ by Proposition~\ref{prop:error-rates}, we get 
$\widetilde\theta_\varepsilon\overset{L^2}{\to}
g(\chi)\Upsilon_\blambda-\e_\bzero[g(\chi)\Upsilon_\blambda]$ 
as $\varepsilon\to 0$. By Lemma~\ref{lem:V_k_DCT}, 
$\varepsilon^2N\to 2\V_\bzero[g(\chi)\Upsilon_\blambda]>0$ and 
the indicator function in the integrand vanishes in probability 
since $\varepsilon N\to\infty$. Thus, 
the dominated convergence theorem~\cite[Thm~1.21]{MR1876169} 
yields $C(\varepsilon)\to 0$, completing the proof. 
\end{proof}

\begin{proof}[Proof of Lemma~\ref{lem:asymp_cost}]
(a) Note that the density of $\ell_n$ is given by 
$x\mapsto\log(1/x)^{n-1}/(n-1)!$, $x\in(0,1)$. 
Note that $x^{-1}=e^{\log(1/x)} = \sum_{k=0}^{\infty} \log(1/x)^k/k!$, 
hence $\int_0^1x^{-1}\phi(x)\D x=\sum_{k=1}^\infty\e[\phi(\ell_k)]$. This yields 
\begin{align*}
n + \int_0^1\frac{1}{x}\big(e^{cx}-1\big)\D x
- \sum_{k=1}^n \e[e^{c\ell_k}]
&= \sum_{k=n+1}^\infty \e [e^{c\ell_k}-1]
= \sum_{k=n+1}^\infty \sum_{j=1}^\infty \frac{c^j}{j!}\e[\ell_k^j]\\
&=\sum_{j=1}^\infty\frac{c^j}{j!}\sum_{k=n+1}^\infty (1+j)^{-k}
=\sum_{j=1}^\infty\frac{c^j}{j!j}(1+j)^{-n}.
\end{align*}
Since $\int_0^1 x^{-1}(e^{cx}-1)\D x = \sum_{j=1}^\infty c^j/(j!j)$ 
and $(1+j)^{-n}\le 2^{-n}$ for all $j\ge 1$, 
the inequality in~\eqref{eq:asymp_cost2} holds, 
implying (a). 

(b) Note that $\int_0^1x^{-1}(e^{cx}-1)\D x
=\int_0^c x^{-1}(e^{x}-1)\D x$ and apply 
l'H\^{o}pital's rule. 
\end{proof}

\appendix

\section{Simulation of stable laws}

In this section we adapt the Chambers-Mellows-Stuck simulation of the 
increments of a L\'evy process $Z$ with generating triplet $(0,\nu,b)$, where 
\[
\frac{\nu(\D x)}{\D x}
=\frac{c_+}{x^{\alpha+1}}\cdot\1_{(0,\infty)}(x)
+ \frac{c_-}{|x|^{\alpha+1}}\cdot\1_{(-\infty,0)}(x),
\]
for arbitrary $(c_+,c_-)\in\R_+^2\setminus\{\bzero\}$ and $\alpha\in(0,2)$. 
First, we introduce the constant $\upsilon$ given by (see (14.20)--(14.21) 
in~\cite[Lem.~14.11]{MR3185174}):
\[
\upsilon = \int_1^\infty x^{-2}\sin(x)\D x
	+\int_0^1 x^{-2}(\sin(x)-x)\D x=1-\gamma.
\]
Then the characteristic function of $Z_t$ is given by 
(see~\cite[Thm~14.15]{MR3185174}) 
\begin{align*}
\e[e^{iuZ_t}]&=\exp\big(t\Psi(u)\big),\qquad u\in\R,\\
\Psi(u)&=i\mu u-
\begin{cases}
\varsigma|u|^\alpha\big(1-i\theta\tan\big(\frac{\pi\alpha}{2}\big)\sgn(u)\big),
	&\alpha\in(0,2)\setminus\{1\},\\
\varsigma|u|\big(1+i\theta\frac{2}{\pi}\sgn(u)\log|u|\big),
	&\alpha=1,
\end{cases}
\end{align*}
where $\theta=(c_+-c_-)/(c_++c_-)$ and the constants $\mu$ and $\varsigma$ 
are given by 
\begin{equation}\label{eq:mu_varsigma}
(\mu,\varsigma) =\begin{cases}
\big(b+\frac{c_--c_+}{1-\alpha},
	-(c_++c_-)\Gamma(-\alpha)\cos\big(\frac{\pi\alpha}{2}\big)\big),
&\alpha\in(0,2)\setminus\{1\},\\
\big(b+\upsilon(c_+-c_-),
	\frac{\pi}{2}(c_++c_-)\big),
&\alpha=1.
\end{cases}
\end{equation}

Finally, we define Zolotarev's function
\[
A_{a,r}(u) = 
\big(1+\theta^2\tan^2\big(\tfrac{\pi\alpha}{2}\big)\big)^{\frac{1}{2\alpha}}
	\frac{\sin(a(r+u))\cos(ar+(a-1)u)^{1/a-1}}{\cos(u)^{1/a}},
\qquad u\in\big(-\tfrac{\pi}{2},\tfrac{\pi}{2}\big).
\]

\begin{algorithm}
\caption{Algorithm 2. (Chambers-Mallows-Stuck) Simulation of $Z_t$ with 
	triplet $(0,\nu,b)$}
\label{alg:cms}
\begin{algorithmic}[1]
	\Require{Parameters $(c_\pm,\alpha,b)$ and time horizon $t>0$}
	\State{Compute $\theta=(c_+-c_-)/(c_++c_-)$ and $(\mu,\varsigma)$ 
		in~\eqref{eq:mu_varsigma}}
	\State{Sample $U\sim\U(-\frac{\pi}{2},\frac{\pi}{2})$ and $E\sim\Exp(1)$}
	\If{$\alpha\ne 1$}
	\State{Compute $\delta=\arctan\big(\theta\tan\big(\tfrac{\pi\alpha}{2}\big)\big)/\alpha$
	and \Return $(\varsigma t)^{1/\alpha}A_{\alpha,\delta}(U)E^{1-1/\alpha}+\mu t$}
	\Else
	\State{\Return 
		$\tfrac{2}{\pi}\varsigma t\Big(\big(\tfrac{\pi}{2}+\theta U\big)\tan(U)
			-\theta\log\Big(\frac{\pi E\cos(U)}{\varsigma t(\pi+2\theta U)}
			\Big)\Big)+\mu t$}
	\EndIf
\end{algorithmic}
\end{algorithm}

\bibliographystyle{abbrvnat}
\bibliography{References}

\section*{Acknowledgement} 
\noindent JGC and AM are supported by EPSRC grant EP/V009478/1 and The Alan Turing Institute under the EPSRC grant EP/N510129/1; 
AM was supported by the Turing Fellowship funded by the Programme on Data-Centric Engineering of Lloyd's Register Foundation and the EPSRC grant EP/P003818/1;
JGC was supported by CoNaCyT scholarship 2018-000009-01EXTF-00624 CVU 699336.
\end{document}